\documentclass[svgnames]{lmcs} 
\pdfoutput=1
\usepackage[utf8]{inputenc}

\usepackage{lastpage}
\lmcsdoi{21}{3}{21}
\lmcsheading{}{\pageref{LastPage}}{}{}%
{Jan.~30,~2023}{Aug.~26,~2025}{}


\usepackage{amsmath}
\usepackage{amssymb} 
\usepackage{hyperref}
\usepackage{subfig}
\usepackage{tikz}
\usepackage{xspace}
\usepackage{latexsym}
\usepackage[normalem]{ulem}
\usetikzlibrary{arrows, automata, patterns.meta}

\usetikzlibrary{decorations.pathmorphing}
\usetikzlibrary{calc}

\newcommand{\range}{\mathtt{range}}
\newcommand{\zones}{\mathtt{zones}}

\makeatletter
\DeclareRobustCommand{\cev}[1]{%
  \mathpalette\do@cev{#1}%
}
\newcommand{\do@cev}[2]{%
  \fix@cev{#1}{+}%
  \reflectbox{$\m@th#1\vec{\reflectbox{$\fix@cev{#1}{-}\m@th#1#2\fix@cev{#1}{+}$}}$}%
  \fix@cev{#1}{-}%
}
\newcommand{\fix@cev}[2]{%
  \ifx#1\displaystyle
    \mkern#23mu
  \else
    \ifx#1\textstyle
      \mkern#23mu
    \else
      \ifx#1\scriptstyle
        \mkern#22mu
      \else
        \mkern#22mu
      \fi
    \fi
  \fi
}

\makeatother

\newcommand{\RED}[1]{\textcolor{red}{#1}}
\newcommand{\BROWN}[1]{\textcolor{brown}{#1}}

\newcommand{\MAGENTA}[1]{\textcolor{magenta}{#1}}
\newcommand{\ORANGE}[1]{\textcolor{orange}{#1}}

\newcounter{dgnot} 
\newenvironment{dgnot}[1][]{\refstepcounter{dgnot}\par\medskip
   \noindent \textbf{\RED{NfD~\thedgnot.}  #1} \rmfamily}{\medskip}

\newcounter{mknot} 
\newenvironment{mknot}[1][]{\refstepcounter{mknot}\par\medskip
   \noindent \textbf{\MAGENTA{NfM~\themknot.}  #1} \rmfamily}{\medskip}

\newcounter{vcnot} 
\newenvironment{vcnot}[1][]{\refstepcounter{vcnot}\par\medskip
   \noindent \textbf{\BROWN{NfV~\thevcnot.}  #1} \rmfamily}{\medskip}

\newcounter{eknot} 
\newenvironment{eknot}[1][]{\refstepcounter{vcnot}\par\medskip
   \noindent \textbf{\ORANGE{NbE~\thevcnot.}  #1} \rmfamily}{\medskip}

\newcommand{\calC}{\mathcal{C}}
\newcommand{\calD}{\mathcal{D}}

\newcommand{\calI}{\mathcal{I}}
\newcommand{\calJ}{\mathcal{J}}
\newcommand{\calK}{\mathcal{K}}

\newcommand{\calM}{\mathcal{M}}
\newcommand{\calN}{\mathcal{N}}

\newcommand{\calP}{\mathcal{P}}

\newcommand{\calS}{\mathcal{S}}
\newcommand{\calT}{\mathcal{T}}

\newcommand{\calV}{\mathcal{V}}

\newcommand{\deriv}{\noindent\hspace*{.20in}\vspace{0.1in}}

\newcommand{\hint}[2]{\newline#1\hspace*{.25in} [\,$#2$\,]\vspace{0.1in}\newline\hspace*{.20in}
\vspace{0.1in}}

\newcommand{\blankline}{\vspace{1.0\baselineskip}}

\newcommand{\SET}[1]{\{#1\}}
\newcommand{\ZET}[2]{\SET{\, {#1} \mid {#2} \,}}
\newcommand{\pws}{\mathbf{\calP}}

\newcommand{\nats}{\mathbb{N}}
\renewcommand{\succ}{\mathtt{succ}}

\newcommand{\cs}{CS}
\newcommand{\qdcs}{QdCS}
\newcommand{\closure}{\calC}
\newcommand{\closureF}{\vec{\closure}}
\newcommand{\closureT}{\cev{\closure}}
\newcommand{\interior}{\calI}
\newcommand{\interiorF}{\vec{\interior}}
\newcommand{\interiorT}{\cev{\interior}}

\newcommand{\fpthsF}{\mathtt{FPaths^F}}

\newcommand{\nmodels}{\nvDash}

\newcommand{\dbs}{DBS}
\newcommand{\cl}{$\calC$}
\newcommand{\cm}{CM}
\newcommand{\cmc}{CMC}
\newcommand{\pth}{Path}
\newcommand{\cop}{CoPa}
\newcommand{\qdcm}{QdCM}
\newcommand{\ap}{\texttt{AP}}
\newcommand{\model}{\calM}
\newcommand{\peval}{\calV}
\newcommand{\invpeval}{\peval^{-1}}

\newcommand{\eqsign}{\rightleftharpoons}
\newcommand{\cmbis}{\eqsign_{\mathtt{\cm}}}

\newcommand{\cmcbis}{\eqsign_{\mathtt{\cmc}}}

\newcommand{\copbis}{\eqsign_{\mathtt{\cop}}}
\newcommand{\copbisold}{\eqsign_{\mathtt{\cop}}^{\ref{def:CoPabisimilarityOld}}}
\newcommand{\copbisnew}{\eqsign_{\mathtt{\cop}}^{\ref{def:CoPabisimilarity}}}
\newcommand{\dbsbis}{\eqsign_{\mathtt{\dbs}}}

\newcommand{\dbseq}{\eqsign_{\mathtt{\dbs}}}

\newcommand{\iml}{{\tt IML}}

\newcommand{\imlc}{{\tt IMLC}}

\newcommand{\icrl}{{\tt ICRL}}
\newcommand{\slcs}{{\tt SLCS}}
\newcommand{\islcs}{{\tt ISLCS}}
\newcommand{\minilogica}{{\tt MiniLogicA}}
\newcommand{\topochecker}{{\tt topochecker}}
\newcommand{\voxlogica}{{\tt VoxLogicA}}
\newcommand{\graphlogica}{{\tt GraphLogicA}}
\newcommand{\form}{\Phi}
\newcommand{\ltrue}{{\tt true}}
\newcommand{\lfalse}{{\tt false}}
\newcommand{\lneg}{\neg}
\newcommand{\lior}{\bigvee}
\newcommand{\liand}{\bigwedge}
\newcommand{\ldist}{\calD}
\newcommand{\lnear}{\calN}
\newcommand{\lreach}{\rho}
\newcommand{\lsurr}{\calS}

\newcommand{\lnearF}{\vec{\lnear}}
\newcommand{\lnearT}{\cev{\lnear}}
\newcommand{\ltothru}{\vec{\rho}}
\newcommand{\lfromthru}{\cev{\rho}}

\newcommand{\lstothru}{\vec{\zeta}}
\newcommand{\lsfromthru}{\cev{\zeta}}

\newcommand{\imleq}{\simeq_{\iml}}

\newcommand{\icrleq}{\simeq_{\icrl}}
\newcommand{\imlceq}{\simeq_{\imlc}}

\newcommand{\closedefi}{\hfill$\bullet$}

\newcommand{\cnrf}[1]{\hat{#1}}

\newcommand{\sem}[1]{\mkern1mu [\![#1]\!]}

\newcommand{\lts}{LTS}

\providecommand*{\donothing}[1]{}

\begin{document}


\title{On Bisimilarity for Quasi-Discrete Closure Spaces}



\thanks{
Research partially supported by  
MUR project 
PRIN 2020TL3X8X ``T-LADIES'', and partially funded by 
European Union - Next Generation EU - MUR-PRIN 20228KXFN2 ``Stendhal'', 
bilateral project between CNR (Italy) and SRNSFG (Georgia) ``Model Checking for Polyhedral Logic'' (\#CNR-22-010), and  
European Union - Next Generation EU, in the context of The National Recovery and Resilience Plan, Investment 1.5 Ecosystems of Innovation, Project “Tuscany Health Ecosystem” (THE), CUP: B83C22003920001.
The authors are listed in alphabetical order, as they equally contributed to the work presented in this paper.}

\author[V.~Ciancia]{Vincenzo Ciancia\lmcsorcid{0000-0003-1314-0574}}[a]	
\author[D.~Latella]{Diego Latella\lmcsorcid{0000-0002-3257-9059}}[b]	
\author[M.~Massink]{Mieke Massink\lmcsorcid{0000-0001-5089-002X}}[a]	
\author[E.~P.~de Vink]{Erik P.~de Vink\lmcsorcid{0000-0001-9514-2260}}[c]

\address{Istituto di Scienza e Tecnologie dell'Informazione ``A. Faedo'', Consiglio Nazionale delle Ricerche, Pisa, Italy}	
\email{Vincenzo.Ciancia@cnr.it, Mieke.Massink@cnr.it}  

\address{Formerly with Istituto di Scienza e Tecnologie dell'Informazione ``A. Faedo'', Consiglio Nazionale delle Ricerche, Pisa, Italy (ret.)}	
\email{Diego.Latella@actiones.eu}  

\address{Department of Mathematics and Computer Science,
  Eindhoven University of Technology, Eindhoven, The Netherlands}
\email{evink@win.tue.nl}





\begin{abstract}
Closure spaces, a generalisation of topological spaces, have shown to be a convenient
theoretical framework for spatial model-checking.
  The closure operator of closure spaces and quasi-discrete closure
  spaces induces a notion of neighbourhood akin to that of topological
  spaces that build on open sets. For closure models and
  quasi-discrete closure models, in this paper we present three notions of bisimilarity that are
  logically characterised by corresponding modal logics with spatial
  modalities: (i)~\cm-bisimilarity for closure models (\cm{s}) is
  shown to generalise topo-bisimilarity for topological
  models and to be an instantiation of neighbourhood bisimilarity, when \cm{s} are seen as  (augmented) neighbourhood models.  \cm-bisimilarity corresponds to equivalence with respect to
  the infinitary modal logic~\iml{} that includes the
  modality~$\lnear$ for ``being near to''. (ii) \cmc-bisimilarity, with \cmc{} standing for \cm-bisimilarity \emph{with converse}, refines
  \cm-bisimilarity for quasi-discrete closure spaces, carriers of
  quasi-discrete closure models. Quasi-discrete closure models come
  equipped with two closure operators, $\closureF$ and~$\closureT$,
  stemming from the binary relation underlying closure and its
  converse.  \cmc-bisimilarity, is captured by the infinitary modal
  logic~\imlc{} including two modalities, $\lnearF$ and~$\lnearT$,
  corresponding to the two closure operators.
  (iii)~\cop-bisimilarity on quasi-discrete closure models, which is weaker
  than \cmc-bisimilarity, is based on the notion of
  \emph{compatible paths}. The logical counterpart of
  \cop-bisimilarity is the infinitary modal logic~\icrl{} with
  modalities $\lstothru$ and~$\lsfromthru$ whose semantics relies on
  forward and backward paths, respectively. It is shown that
  \cop-bisimilarity for quasi-discrete closure models relates to
  divergence-blind stuttering equivalence for Kripke models.
\end{abstract}

\maketitle


\section{Introduction}
\label{sec:Introduction}

Space, like time, has been the object of human reasoning probably since the start of humankind. Though both concepts have been profoundly studied in philosophy and mathematics, in the area of formal methods for computer science time has received much more attention than space. Much work can be found on temporal logics and its application to in formal verification techniques such as model checking. Spatial logics have received more attention only very recently, also in combination with time (see for example~\cite{Ai+07}). In this paper we aim at the development of a uniform, coherent and self-contained framework for the design of logic based, automatic verification techniques, notably model checking, for representations of space and the analysis thereof. We will do so by adopting {\em closure spaces} (\cs{s}), a generalisation of topological spaces, as a
general underlying theory, for space and proposing, for the spatial setting, notions and techniques that have often been inspired by similar ones in the context of modal logic and concurrency theory.

In the well-known topological interpretation of modal logic, a point
in space satisfies the formula~$\Diamond \, \form$ whenever it belongs
to the \emph{topological closure} of the set $\sem{\form}$ of all the
points satisfying formula~$\form$ (see e.g.,~\cite{vBB07}).
The approach dates back to Tarski and McKinsey~\cite{MT44:am}, who
proposed topological spaces and their 
{\em interior operator}~$\Box$ --- i.e., the boolean dual of $\Diamond$ ---
as the fundamental basis for reasoning about space.
However, the idempotence property of topological closure renders
topological spaces to be too restrictive for specific applications.
For instance, discrete structures useful for certain representations
of space, like general graphs or adjacency graphs, including 2D and 3D
images, cannot be captured topologically.
To that purpose, a more liberal notion of space, namely that of
closure spaces, has been proposed in the
literature~\cite{Smy95:tcs,Gal03}.
Closure spaces do not require idempotence of the closure operator
(see~\cite{Gal03} for an in-depth treatment of the subject).

For an interpretation of modal logic 
on {\em closure models} (\cm{s}), i.e., models based on \cs{s},
\cite{Ci+14,Ci+16}~introduced the \emph{Spatial Logic for Closure
  Spaces} (\slcs) that includes the modality of the surround
operator~$\lsurr$.
A point $x$ satisfies
$\form_1 \, \lsurr \, \form_2$ if it lays in a set
$A \subseteq \sem{\form_1}$ while the external border of~$A$
consists of points satisfying $\form_2$, i.e.,\ the point~$x$ lays in an
area satisfying~~$\form_1$ which is surrounded by points
satisfying~$\form_2$. A model-checking algorithm has been proposed in
~\cite{Ci+14,Ci+16} that has been implemented in the tool
\topochecker~\cite{Ci+18,Ci+15} and in
\voxlogica~\cite{Be+19}, a tool specialised for spatial model-checking
digital images, that can be modelled as \emph{adjacency spaces}, a
special case of closure spaces.

The logic and the above mentioned model checkers have been applied to several case
studies~\cite{Ci+16,Ci+15,Ci+18} including a declarative approach to
medical image analysis~\cite{Be+19,Be+19a,Ba+20,Be+21}, 
where the logic has been extended with a {\em distance} operator $\ldist^I$ in order to be able to express properties of points in {\em distance \cm{s}}, i.e., \cm{s} based on \cs{s} equipped with a suitable {\em distance} function. A point $x$ in the space satisfies formula $\ldist^I \form$
 if its distance from the set of points satisfying $\form$ falls in interval $I$.
An encoding of
the discrete Region Connection Calculus RCC8D of~\cite{Ra+13} into the
collective variant of \slcs{} has been proposed in~\cite{Ci+19b}. The
logic has also inspired other approaches to spatial reasoning in the
context of signal temporal logic and system
monitoring~\cite{NBBL22,NBCLM18} and in the verification of
cyber-physical systems \cite{TKG17}.

 A key question, when reasoning about modal logics and their models, is
the relationship between logical equivalences and notions of
bisimilarity defined on their underlying models. This is of great  importance not
only from a theoretical perspective, but also from a practical point of view. In fact,
the existence of such bisimilarities, and 
the existence of a logical characterisation thereof, 
often referred to as the ``Hennessy-Milner'' property of the bisimilarity,
makes it possible to exploit minimisation
procedures for bisimilarity with the purpose of 
model reduction, which may help making 
model-checking more efficient --- and spatial models are notoriously large.
For instance, in~\cite{Ci+23} an encoding of a class of closure spaces
into labelled transition systems has been proposed such that two points in
space are bisimilar --- for an appropriate notion of spatial bisimilarity --- if and only if their
images in the associated labelled transition system are branching equivalent.
Thus, model-checking a formula of the logic characterising spatial bisimilarity can
be safely performed on a model that has been minimised using efficient tools for branching equivalence.

In the present paper we study three different notions of bisimilarity for
\cm{s} and their
associated logics. The first one is \emph{\cm-bisimilarity}, that is
an adaptation for closure models of classical topo-bisimilarity for
topological models as defined in~\cite{vBB07}. Actually,
\cm-bisimilarity is an instantiation to closure models of monotonic
bisimulation on neighbourhood models~\cite{HKP09,Pac17}. In fact, it
is defined using the monotonic interior operator of closure models, 
which characterises the neighbourhoods of a point,
thus making closure models an instantiation of monotonic
neighbourhood models.  We show that,
as one might expect,
 \cm-bisimilarity is logically
characterised by 
a logic with a single modality, namely $\lnear$ (to be read as ``near''), expressing {\em proximity},
and infinite conjunction. $\lnear$ coincides with the $\Diamond$  modality in the topological interpretation of modal logic of Tarski and McKinsey. We call the resulting language Infinitary Modal Logic,~\iml.

For {\em quasi-discrete} closure models (\qdcm{s}), i.e., closure models where every point has a minimal neighbourhood, it is convenient to use a more intuitive, and simpler, definition of \cm-bisimilarity, given directly in terms of the 
closure operator --- instead of the interior operator. Such a definition is reminiscent of the definition of strong equivalence for labelled transition systems~\cite{Mil89}. The direct use of the closure operator in the definition of \cm-bisimilarity simplifies several proofs.

We then present a refinement of \cm-bisimilarity, specialised  for \qdcm{s}. In quasi-discrete closure spaces, the closure of a set of points --- and also its interior~--- can be expressed using an underlying binary relation. This gives rise to both a {\em direct} closure and interior of a set, and a {\em converse} closure and interior, the latter being obtained using the converse of the  binary relation. In turn, this induces a refined notion of bisimilarity, {\em \cm-bisimilarity with converse}, abbreviated as \cmc-bisimilarity, which is shown to be strictly stronger than \cm-bisimilarity.

We also present a closure-based definition 
for \cmc-bisimilarity, originally proposed in~\cite{Ci+20}, that resembles strong back-and-forth bisimilarity for labelled transition systems proposed by De Nicola, Montanari, and Vaandrager in~\cite{De+90}.  
In order to capture \cmc-bisimilarity logically, we extend \iml{} with the
converse of its unary modal operator and show that the resulting
logic~\imlc{} characterises \cmc-bisimilarity.  
We recall here that in~\cite{Ci+20} a minimisation
algorithm for \cmc-bisimulation and the related tool \minilogica{} have
been proposed as well.

\cm-bisimilarity and \cmc-bisimilarity play an
important role as they are counterparts of classical
Topo-bisimilarity. On the other hand, they turn out to be rather
strong when one has intuitive relations on space  in mind  --- like reachability, or adjacency, of areas of points with specific features --- that may be useful for instance when dealing with models representing images.

Let us consider, as an example, the image of a maze shown in Figure~\ref{fig:TheMaze}. In the image, 
walls are represented in black and the exit area, situated at the border of the image, is shown in green; 
the floor is represented in white.
Suppose we want to know whether, starting from a given point in the maze --- for instance one of those shown in blue in the picture --- one can reach the exit area. 

\begin{figure}
\def\samplesz{2.5cm}
\centering
\subfloat[][]
{
\includegraphics[height=3.5cm]{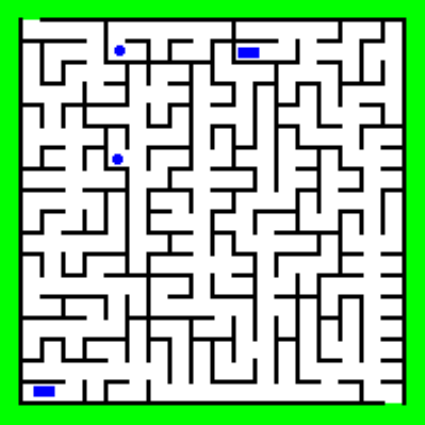}
\label{fig:TheMaze}
}\quad
\centering
\subfloat[][]
{
\includegraphics[height=1.7cm]{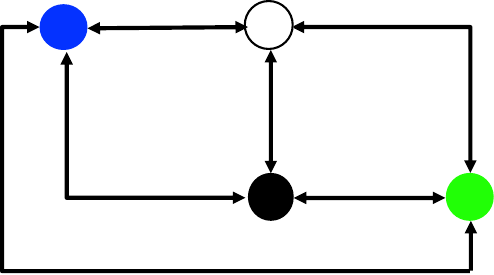}
\label{fig:PathReducedMaze}
}\quad
\centering
\subfloat[][]
{
\includegraphics[height=\samplesz]{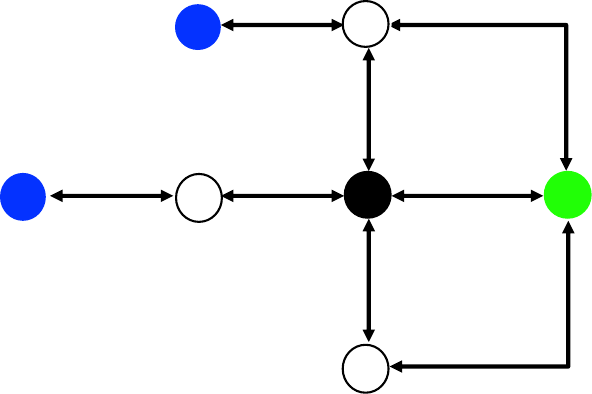}
\label{fig:CoPaReducedMaze}
}
\caption{A maze~(\ref{fig:TheMaze}), its path minimal model~(\ref{fig:PathReducedMaze}), and its  \cop-minimal models~(\ref{fig:CoPaReducedMaze}). In the maze, there are two exit gates, one at the top-left corner and the other one in the bottom-right corner.}
\label{fig:Maze}
\end{figure}

In other words, we are interested in those paths in the picture, starting in blue points, 
leading to green points {\em passing only} through white points.
In~\cite{Ci+21} we introduced the notion of {\em spatial path bisimilarity}. Essentially, 
two points are path bisimilar if (they satisfy the same set of atomic proposition letters and) for every path starting in one of the two points, there is a path starting in the other point such that the end-points of the two paths are path bisimilar. This is illustrated in more detail in Figure~\ref{fig:Maze}.
Figure~\ref{fig:PathReducedMaze} shows the minimal model for the maze shown in Figure~\ref{fig:TheMaze} according to path bisimilarity. We  see that all blue points are equivalent --- they all collapse to a single blue point in the minimal model --- and so are all white points. In other words, we are unable to distinguish those blue (white) points from which one can reach an exit  
from those from which one cannot, due to the (black) walls. This is not satisfactory; as we said before, we would like to differentiate the blue points (also) on the basis of whether they can reach an exit, considering the presence of walls.
We are then looking for  a notion of equivalence that is not only based on reachability, but also takes the structure of paths into consideration. Reachability should not be unconditional, i.e., the  relevant paths should share some common structure. For that purpose, we resort to a notion of ``mutual compatibility'' between relevant paths that essentially requires each of them to be composed by a sequence of non-empty ``zones''. The total number of zones in each of the  paths must be the same, while the zones may be of arbitrary (but not zero) length. Each element of one path in a  given zone is required to be related by the bisimulation relation to all the elements in the corresponding zone in the other path, as illustrated in Figure~\ref{fig:Zones}.

\begin{figure}
  \centering
  \includegraphics[height=2.3cm]{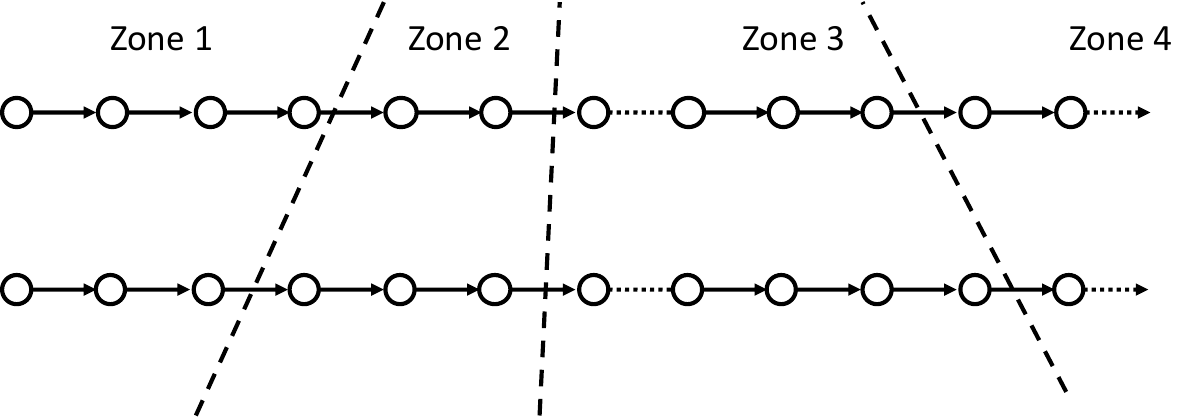}
  \caption{Zones of compatible paths.}
  \label{fig:Zones}
\end{figure}

This idea of compatibility of paths gives rise to the third notion of
bisimulation, namely {\em compatible path bisimulation}, or 
\cop-bisimulation for short, which is strictly stronger than path bisimilarity
and, for quasi-discrete closure model{s}, strictly weaker than
\cmc-bisimilarity.

The minimal model, according to \cop-bisimilarity, for the image of the maze shown in Figure~\ref{fig:TheMaze}, is shown
Figure~\ref{fig:CoPaReducedMaze}. It is worth noting that this model distinguishes the blue points
from which one can reach green ones passing only by white points from those from which one cannot. 
Similarly, white points through which an exit can be reached from a blue point (without passing through a black point) are distinguished both from those that cannot be reached from blue points and from those through which no green point can be reached. 
Finally, we note that the \cop-minimal model also shows that there is no white area that does not contain a blue region from which a green exit cannot be reached (i.e., a  completely white area  surrounded by walls).

We provide a
logical characterisation of \cop-bisimilarity by means of  the 
{\em Infinitary Compatible Reachability Logic} (\icrl). This logic involves the modalities of forward and backward
{\em conditional} reachability, $\lstothru$ and~$\lsfromthru$. 
The notion of \cop-bisimulation is reminiscent of that of 
{ stuttering equivalence} for Kripke models,
going back to~\cite{BCG88}, although in a different context and with
different definitions as well as underlying notions. We show the
coincidence of \cop-bisimilarity and an adaptation to the \qdcm{s} setting of divergence-blind
stuttering equivalence~\cite{Gr+17}.

It is worth noting that, in the context of space, and in particular when dealing with notions of directionality (e.g., one-way roads, public area gates, expansion of/radiation from certain substances), it is essential to be able to distinguish between the concept of “reaching” and that of “being reached”. For instance, one might be interested to express the fact that, from a certain location, 
via a safe corridor, a rescue area can be reached (forward reachability) that cannot be reached by (backward reachability) smoke generated in another area. This kind of situations have no obvious counterpart in the temporal domain, including temporal logics with ``past'' operators, where there can be more than one future, like in the case of branching-time logics, but there is typically only one, fixed, past. 
In fact, for the great majority of temporal logics in the literature,
it is typically assumed that there is  only one, fixed, past, i.e. the one that has occurred, whereas, at least in the case of  branching-time temporal logics, more than one possible future is allowed. In particular, for branching-time temporal logics, the model of time is typically a {\em tree}, i.e. a partial ordering $(\calT,\preceq)$ where every instant $t\in \calT$ has a {\em linearly} ordered set of $\prec$-predecessors in $\calT$ as remarked also in the Stanford Encyclopedia of Philosophy: \begin{quotation}
``Backward-linearity captures the idea that the past is fixed; forward-branching reflects the openness of the future. At each point in time, there may be more than one path leading towards the future, and each such path represents a future possibility.''~\cite{GoR23}.
\end{quotation}
There are only a few notable exceptions to this interpretation of the past, e.g.,~\cite{KuR97,Sti93}, although we are not aware of the existence of any supporting verification tool, such as model-checkers, for the logics presented in these works. 
We point out that doubly labelled transition systems, i.e., transition systems where
both states and transitions are labelled, in conjunction with languages like the Modal $\mu$Calculus, could be used for encoding both forward and backward spatial reachability.
One could, for instance, encode a finite closure model into a labelled transition system using two, ad hoc, transition labels, one for ``forward'' transitions, used for simulating forward proximity a reachability, and the other for ``backward'' ones for simulating backward proximity and reachability.
However, as we mentioned before, this paper focuses on the commonalities of various bisimilarities for closure models, with particular emphasis on \qdcm{s}. As we remarked before, the main objective of our work is, therefore, the development of a {\em uniform}, {\em coherent} and {\em self-contained} theoretical framework for the design and application of logic-based, automated verification techniques, notably model-checking, to representations of space and the analysis thereof.

Many notions we have developed for our theory have their conceptual origin in other areas of research. 
For instance, \cm-bisimilarity is closely
related to neighbourhood bisimilarity, while the logic \imlc{}, inspired
by the work of De Nicola et al.\ on labelled transition systems~\cite{De+90}, is reminiscent of modal logic with a direct and a converse modality.

Finally, \cop-bisimilarity is reminiscent of stuttering equivalence and the forward and backward reachability operators 
of \icrl{} could be simulated using the Modal $\mu$Calculus over doubly labelled transition systems, as we noted before.

One could wonder why we are developing such a theoretical framework, affording a task that looks like rephrasing, in the language of topology,  notions and results already developed in other contexts. Let us address this question in more detail. At a methodological level, as we have already stated above, we need a uniform, coherent and self-contained theoretical framework for spatial models in order to develop efficient analysis techniques and tools, and for optimising them with the support of a solid mathematical basis.
This is, for instance, the case for model reduction based on minimisation modulo bisimilarity.
Of course, a user could map, every time,  and according to specific needs, \cm{s} and related logics to different models. Thus one could use Kripke models or neighbourhood models and classical modal logic with a single modality for \cm-bisimilarity and \iml{.} Similar structures and
 modal logic with two modalities could be used for \cmc-bisimilarity and \imlc{.} Doubly labelled transition systems and the Modal $\mu$Calculus could be used for \cop-bisimilarity and \icrl{.}
 However, this would have a negative impact on the {\em usability} of our framework since the knowledge needed for using it would have to span over many different theories, making reasoning about space more complicated. In addition, the fact that such a framework would not
 be self-contained, would have a negative impact also on the development of tools supporting model analysis, since different kinds of model reduction techniques require ad hoc encodings from models of space to different target models (see, e.g.,~\cite{Ci+23}). Therefore, we prefer to develop a single, coherent, framework, with all (and only) those conceptual tools that are necessary and useful for our purposes, the latter being the modelling of space and the analysis of such models via model-checking.

 It is finally worth pointing out that in~\cite{Ci+23a} we adapted  the concept of
 path compatibility, discussed in depth in the present paper, to so called {\em $\pm$-paths} and we
 introduced a novel notion of bisimilarity, namely {\em $\pm$-bisimilarity}. The latter applies to {\em 
 cell poset models}, that are discrete representations of polyhedra models, i.e., {\em continuous} models of space (the interested reader is referred to~\cite{Bez+22,Ci+23a} for details). 
 We have shown that $\pm$-bisimilarity on cell poset models
 coincides with logical equivalence on such models. The latter coincides in turn  with {\em simplicial bisimilarity} on polyhedra models. 
In~\cite{Be+24a} a weaker logic and weaker versions of simplicial bisimilarity and 
$\pm$-bisimilarity have been proposed. This 
 logic as well as the related bisimilarities allow for
efficient model minimisation while offering adequate expressive power~\cite{Be+24b}.
 
It is  worth noting that \cmc-bisimilarity on cell poset models is included in $\pm$-bisimilarity, so that, it is safe\footnote{Although, obviously not optimal, but maybe cheap, given that efficient algorithms that implement reduction via \cmc-bisimilarity are rather easy to implement.} to use \cmc-bisimilarity for model reduction. This clearly illustrates that the availability of a uniform framework for dealing with space, both discrete and continuous, is quite profitable and that developing such a framework is worthwhile.

The present paper builds upon ideas and notions originally presented by the authors in~\cite{Ci+20,Ci+21,Ci+22a}. In particular,  this paper contains all the proofs of the results presented in~\cite{Ci+22a}. Furthermore, an alternative, more intuitive, definition of \cop-bisimilarity is proposed that is based on an explicit notion of path compatibility.
Finally,  additional results relating our notions with those in the literature are proved as well.

\paragraph{\em Contributions.} Summarising, in this paper 
we develop a uniform, coherent and self-contained theoretical framework for 
modelling space and for the design and application of  efficient, logic based, model analysis techniques. In particular, 
we present three notions of bisimilarity that are logically characterised by corresponding modal logics with spatial modalities:
\begin{itemize}
\item \cm-bisimilarity for closure models  (\cm{s}), which generalises topo-bisimilarity for topological models and coincides with neighbourhood bisimilarity when \cm{s} are seen as (augmented)
neighbourhood models. \cm-bisimilarity corresponds to equivalence with respect to the infinitary modal logic IML that includes the proximity modality~$\lnear$ standing for “being near”.
\item \cmc-bisimilarity, which refines
  \cm-bisimilarity for quasi-discrete closure spaces, carriers of
  quasi-discrete closure models. \cmc-bisimilarity, is captured by the infinitary modal
  logic~\imlc{} including two proximity modalities, $\lnearF$ and~$\lnearT$,
  corresponding to the forward and backward closure operators, respectively.
\item \cop-bisimilarity on quasi-discrete closure models based on the notion of
  \emph{compatible paths} and which is weaker
  than \cmc-bisimilarity. The logical counterpart of
  \cop-bisimilarity is the infinitary modal logic~\icrl{} with conditional
  reachability modalities $\lstothru$ and~$\lsfromthru$ whose semantics relies on
  forward and backward paths, respectively. It is shown that
  \cop-bisimilarity for quasi-discrete closure models relates to
  divergence-blind stuttering equivalence for Kripke structures.
\end{itemize}
We also discuss some preliminary results of a feasibility study of spatial model checking based on a suitable LTS-encoding in {\tt mCRL2}~\cite{Bu+19}
of a spatial model to obtain \cop-bisimilarity via the application of branching bisimilarity minimisation on the encoded LTS~\cite{Ci+23}.

\paragraph{\em Further related work.} Our work is also inspired by spatial logics (see~\cite{vBB07} for an
extensive overview), including seminal work dating back to Tarski and
McKinsey in the forties of the previous century. The work on {\em
  spatial model-checking} for logics with reachability originated in
\cite{Ci+14} and~\cite{Ci+16}, which includes a comparison to the work of Aiello on
spatial \emph{until} operators (see e.g.,~\cite{A02}). 
In~\cite{A03},
Aiello envisaged practical applications of topological logics with
\emph{until} to minimisation of images. Recent work in~\cite{Ci+23} builds on
and extends that vision taking \cop-bisimilarity as a suitable equivalence for spatial minimisation.

In~\cite{Li+20,Li+21} the spatial logic \slcs{} is studied from a model-theoretic perspective. In particular, in~\cite{Li+20} the authors are focused on issues of expressivity of \slcs{} in relation to topological connectedness and separation. In~\cite{Li+21} it is shown that the logic admits finite models for a sub-class
of neighbourhood models, namely the class of quasi-discrete neighbourhood models, both for topological paths and for quasi-discrete ones --- which consist of an enumeration of points --- whereas it does not enjoy the finite model property
for general neighbourhood models, regardless of the kind of paths considered.

The work in~\cite{Bez+22} and in~\cite{LoQ21} introduces
bisimulation relations that characterise spatial logics with
reachability in polyhedral models and in simplicial complexes,
respectively. It will be interesting future work to apply the
minimisation techniques we present here to such relevant classes of models.

Spatial analysis based on \qdcm{s} is also related to work on qualitative reasoning about spatial
entities (see~\cite{CohnR2008} and references therein), also known as QSR (Qualitative Spatial Reasoning). 
This is a very active area of research in which the theory of topology and that of closure spaces play an
important role. Prominent examples of that area are the region
connection calculi (RCC), such as the discrete space variant RCC8D. As mentioned before, an embedding of the latter in the
collective variant of \slcs\ was presented in~\cite{Ci+19b}. Such embedding enables the application of spatial model-checking of RCC8D via the spatial model checker \voxlogica~\cite{Be+19}.

In~\cite{HKP09}, coalgebraic bisimilarity is developed for a general kind of models, generalising the topological ones, known as Neighbourhood Frames.  At the moment, the notions of path and reachability are not part of the framework (that is, bisimilarity in neighbourhood semantics is based on a one-step relation), thus the results therein, although more general than the theory of closure spaces, cannot be directly reused in the work we present in this paper. Steps towards extending Neighbourhood Frames in this direction in the context of Closure Hyperdoctrines can be found in~\cite{CaM21}.
 
In the Computer Science literature, other kinds of spatial logics have been proposed that typically describe situations in which modal operators are interpreted {\em syntactically} against the {\em structure of agents} in a process calculus. We refer 
to~\cite{CaG00,CaCC03} for some classical examples. Along the same lines, an example is given in~\cite{TPGN15}, concerning model-checking of security aspects in cyber-physical systems, in a spatial context based on the idea of bigraphical reactive systems introduced by Milner~\cite{Mil09}. A bigraph consists of two graphs: a place graph, i.e., a forest defined over a set of nodes which is intended to represent entities and their locality in terms of a containment structure, and a link graph, a hyper graph composed over the same set of nodes representing arbitrary linking among those entities. The \qdcm{s} that are the topic of the present paper, instead, address space from a topological point of view rather than viewing space as a containment structure for spatial entities.

\paragraph{\em Structure.} The paper is organised as follows: Preliminary notions and definitions
can be found in Section~\ref{sec:Preliminaries}, while
Section~\ref{sec:CMbisimilarity} is devoted to \cm-bisimilarity and
the logic~\iml. Section~\ref{sec:CMCbisimilarity} deals with
\cmc-bisimulation and the logic~\imlc. Next,
Section~\ref{sec:COPAbisimilarity} addresses \cop-bisimilarity together with the
logic~\icrl.  
Section~\ref{sec:conclusions} presents conclusions and open questions for future work.
An additional proof is shown in the appendix.


\section{Preliminaries}
\label{sec:Preliminaries}

This section introduces  relevant notions, concepts, and the notation that are used in the sequel.
In particular, we recall the notions of closure space and quasi-discrete closure space, paths in such spaces, and closure models.

Given a set $X$, $|X|$ denotes the cardinality of $X$ and $\pws(X)$ denotes the power set of $X$;
for $Y \subseteq X$ we use $\overline{Y}$ to denote $X\setminus Y$,
i.e.,\ the complement of~$Y$ with respect to~$X$. 
For $x_0,\ldots,x_{\ell} \in X$, we let $(x_i)_{i=0}^{\ell}$ denote
the sequence $(x_0,\ldots,x_{\ell})\in X^{\ell+1}$.
For
$m_1,m_2 \in \nats$, with~$\nats$ the set of natural numbers,
we use $[m_1;m_2]$ to denote the set
$\ZET{n \in \nats}{m_1 \leqslant n \leqslant m_2}$ and $[m_1;m_2)$ to
denote $\ZET{n \in \nats}{m_1 \leqslant n < m_2}$. The sets
$(m_1;m_2]$ and $(m_1;m_2)$ are defined in a similar way.
For a function $f:X \to Y$ and $A \subseteq X$ and $B \subseteq Y$ we have
$f(A) = \ZET{f(a)}{a \in A}$ and $f^{-1}(B) = \ZET{a}{f(a) \in B}$, respectively. We let $\range(f)$ denote the
range~$f(X)$ of~$f$. The \emph{restriction} of~$f$ to~$A$ will be
denoted by $f|A$, as usual. For a binary relation
$R \subseteq X \times X$, we use~$R^{-1}$ to denote the converse relation
$\ZET{(x_1,x_2)}{(x_2,x_1) \in R}$, $R^=$ to denote the reflexive closure of $R$,
and we let $\cnrf{R}$ be defined as $(R^=)^{-1}$. We let $R[x]$ denote the set $\ZET{x'\in X}{(x,x')\in R}$.
In the sequel, we assume a set $\ap$ of \emph{atomic proposition letters} is given.

\begin{defi}[Kripke model]\label{def:KripkeModel}
A {\em Kripke frame} is a pair $(X,R)$ where $X$ is a set and $R\subseteq X \times X$ is a binary relation over $X$. A {\em Kripke model} is a triple $(X,R,\peval)$ where $(X,R)$ is a Kripke frame
and $\peval: \ap \to \pws(X)$ is an evaluation function for proposition letters  mapping each $p \in \ap$ into the set of elements of $X$ satisfying $p$.
\closedefi
\end{defi}

\begin{defi}[Kripke bisimulation]\label{def:KripkeBisimulation}
  Given a Kripke model $\calK = (X, R, \peval)$, a symmetric relation
  $B \subseteq X \times X$ is called a \emph{bisimulation}
  for~$\calK$ if for all $x_1, x_2 \in X$ such that~$(x_1,x_2)\in B$ the
  following holds:
  \begin{enumerate}
  \item  $x_1 \in \peval(p)$ if and only if  $x_2 \in \peval(p)$, for all $p\in\ap$.
  \item If $(x_1,x'_1) \in R$ for some $x'_1 \in X$, then 
  $x'_2 \in X$ exists such that  $(x_2,x'_2) \in R$ and $(x'_1,x'_2)\in B$.
  \end{enumerate}
  Two points $x_1, x_2 \in X$ are called bisimilar in~$\calK$ if
  $(x_1,x_2)\in B$ for some bisimulation~$B$ for~$\calK$. 
\closedefi
\end{defi}

\begin{defi}[Neighbourhood model]\label{def:NeighbourhoodModel}
A {\em neighbourhood frame} is a pair $(X,\nu)$ where $X$ is a set and $\nu:X \to \pws(\pws(X))$
is a neighbourhood function mapping each element of $X$ into its collection of neighbourhoods.
A neighbourhood frame is {\em monotonic} if, for all $x\in X$ and all $U,U' \subseteq X$ it holds that
$U \subseteq U'$ and $U \in \nu(x)$ implies $U' \in \nu(x)$. A monotonic frame is  {\em augmented} if, for all 
$x\in X$, $\bigcap \nu(x) \in \nu(x)$.
A {\em neighbourhood model} is a triple $(X,\nu,\peval)$ where $(X,\nu)$  is a  neighbourhood frame
and $\peval: \ap \to \pws(X)$ is an evaluation function mapping each proposition letter $p \in \ap$ into the set of elements of $X$ satisfying $p$.  Neighbourhood model $(X,\nu,\peval)$ is monotonic/augmented
if $(X,\nu)$ is so.
\closedefi
\end{defi}

\begin{defi}[Monotonic neighbourhood bisimulation]\label{def:MonoNeiBisimulation}
  Given a monotonic neighbourhood model $\calJ=(X,\nu,\peval)$, a symmetric relation
  $B \subseteq X \times X$ is called a \emph{bisimulation}
  for~$\calJ$ if for all $x_1, x_2 \in X$ such that~$(x_1,x_2)\in B$ the
  following holds:
  \begin{enumerate}
  \item  $x_1 \in \peval(p)$ if and only if  $x_2 \in \peval(p)$, for all $p\in\ap$.
  \item For all $U_1\in \nu(x_1)$,  $U_2 \subseteq X$ exists such that
   $U_2\in \nu(x_2)$ and, for all $x'_2 \in U_2$,  $x'_1 \in U_1$ exists such that $(x'_1,x'_2)\in B$.
  \end{enumerate}
  Two points $x_1, x_2 \in X$ are called bisimilar in~$\calJ$ if
  $(x_1,x_2)\in B$ for some bisimulation~$B$ for~$\calJ$. 
\closedefi
\end{defi}

\noindent
Given a Kripke frame $(X, R)$, one can define a corresponding augmented neighbourhood frame
$(X,\nu_R)$ by putting $\nu_R(x)=\ZET{U\subseteq X}{R[x]\subseteq U}$. The correspondence extends to Kripke and neighbourhood models in the obvious way. The following lemma relates the two notions of a model~\cite{HKP09}:
\begin{lem}\label{KbisAndNbis}
Given a Kripke model $\calK=(X, R, \peval)$ and a neighbourhood model
$\calJ=(X,\nu_R,\peval)$ with $\nu_R(x)=\ZET{U\subseteq X}{R[x]\subseteq U}$ for all $x\in X$,
the following holds: for any two  $x_1,x_2\in X$,
$x_1$ and $x_2$ are Kripke bisimilar as elements of $\calK$ if and only if they are 
neighbourhood bisimilar as elements of $\calJ$.
\end{lem}

\noindent
Our  framework  for modelling space is based on the notion of {\em \v{C}ech Closure Space}~\cite{Cec66}, \cs{} for short, also known as {\em pre-topological space}, that provides a convenient common framework for the study of several kinds of spatial models, including models of both discrete and continuous space~\cite{SmW07}. We briefly recall several definitions and results on
\cs{s}, most of which are borrowed from~\cite{Gal03}.

\begin{defi}[Closure Space]\label{def:ClosureSpace}
  A \emph{closure space}, \cs{} for short, is a pair $(X,\closure)$
  where $X$ is a  set (of {\em points}) and {\em
    $\closure: \pws(X) \to \pws(X)$} is a function, referred
    to as the closure operator, satisfying the following axioms:
\begin{enumerate}
\item $\closure(\emptyset)=\emptyset$;
\item $A \subseteq \closure(A)$ for all  $A \subseteq X$;
\item \label{CS3}
  $\closure(A_1 \cup A_2) = \closure(A_1) \cup \closure(A_2)$ for
  all $A_1,A_2\subseteq X$. 
\closedefi
\end{enumerate}
\end{defi}

\noindent
  A closure operator~$\closure$ is monotone: if $A_1 \subseteq A_2$
  then
  $\closure(A_2) = \closure(A_1 \cup A_2) = \closure(A_1) \cup
  \closure(A_2)$, hence $\closure(A_1) \subseteq \closure(A_2)$.
It is also worth pointing out that topological spaces 
coincide with the sub-class of \cs{s} for which the closure operator
is idempotent, i.e.,
\begin{enumerate}
\item[(4)]$\closure(\closure(A))= \closure(A)$,  for all~$A \subseteq X$.
\end{enumerate}
Axioms (1)--(4) are commonly known as the Kuratowski closure axioms for topological spaces, after K. Kuratowski, who first formalised them~\cite{Kur1922}.

Closure models are the basic structures for  interpreting spatial logics on closure spaces:
\begin{defi}[Closure model]
  \label{def:ClosureModel}
  A \emph{closure model}, \cm{} for short, is a tuple
  $\model = (X,\closure, \peval)$, with $(X,\closure)$ a \cs{,} and
  $\peval: \ap \to \pws(X)$ the (atomic proposition letters) valuation function,
  assigning to each $p \in \ap$ the set of points where $p$~holds.
  \closedefi
\end{defi}

Below, we recall some basic notions and results concerning \cs{s} and \cm{s}.

The \emph{interior} operator
$\interior : \pws(X) \to \pws(X)$ of CS~$(X,\closure)$ is the
dual of the closure operator:
$\interior(A) = \overline{\closure(\overline{A})}$
for~$A \subseteq X$.
It holds that $\interior(X) = X$, $\interior(A) \subseteq A$, and
$\interior(A_1 \cap A_2) = \interior(A_1) \cap \interior(A_2)$.
Like the closure operator, the interior operator is monotone.
A \emph{neighbourhood} of a point $x \in X$ is any set
$A \subseteq X$ such that $x \in \interior(A)$.
A \emph{minimal} neighbourhood of a point $x \in X$ is a neighbourhood
$A$ of $x$ such that $A \subseteq A'$ for any other neighbourhood~$A'$
of~$x$.

We will use the following property of closure spaces:
\begin{thmC}[(Theorem~14.B.6 of~\cite{Cec66})]
  \label{thm-point-in-closure}
  Let $(X,\closure)$ be a CS\@. For $x \in X$, $A \subseteq X$, it
  holds that $x \in \closure(A)$ iff $U \cap A \neq \emptyset$ for
  each neighbourhood~$U$ of~$x$.
\end{thmC}

\noindent
Given a closure space $(X, \closure)$, one can define a corresponding neighbourhood space
$(X,\nu_{\closure})$ by putting $\nu_{\closure}(x)=\ZET{U\subseteq X}{x \in \interior(U)}$. The relationship extends to closure and neighbourhood models in the obvious way.\\

\noindent
In the sequel we will mainly work with quasi-discrete closure spaces,
rather than closure spaces in general.

\begin{defi}[Quasi-discrete closure space]
  \label{def:QDClosureSpace}
  A \emph{quasi-discrete closure space}, \qdcs{} for short, is a \cs{}
  $(X,\closure)$ such that any of the following equivalent conditions
  holds:
  \begin{enumerate}
  \item each $x \in X$ has a minimal neighborhood;
  \item for each $A \subseteq X$ it holds that
    $\closure(A) = \bigcup_{x\in A}\closure(\SET{x})$.
    \closedefi
  \end{enumerate}
\end{defi}

Given a relation $R \subseteq X \times X$, let the function
$\closure_{R}: \pws(X) \to \pws(X)$ be such that
\begin{displaymath}
  \closure_R(A) = A \cup \ZET{x \in X}{\exists \mkern1mu a \in A \colon
    (a,x) \in R}
\end{displaymath}
for all~$A \subseteq X$. It is easy to see that, for
any relation~$R$, function~$\closure_{R}$ satisfies all the axioms
of Definition~\ref{def:ClosureSpace} and so $(X, \closure_{R})$ is
a~\cs.
  As a consequence, for directed and undirected graphs, the edge 
  relation gives rise to a \qdcs.
  It follows from condition~\ref{CS3} of
  Definition~\ref{def:ClosureSpace} that every finite closure space
  is quasi-discrete.

\blankline

\noindent
The following theorem  is a standard result in the theory of
\cs{s}.

\begin{thmC}[(Theorem 1 of~\cite{Gal03})]\label{thm:QDR}
  A \cs{} $(X, \closure)$ is quasi-discrete iff there is a relation
  $R \subseteq X \times X$ such that $\closure = \closure_{R}$.  
\end{thmC}

Under the conditions of Theorem~\ref{thm:QDR}, we say that $R$ {\em generates} the \qdcs{.}

The definition of closure model is extended to \qdcs{s} in the obvious way.
A \emph{quasi-discrete closure model} (\qdcm{} for short) is a \cm{}
$(X, \closure,\peval)$ where $(X,\closure)$ is a \qdcs.  For a
closure model $\model=(X,\closure,\peval)$ we often write
$x \in \model$ when~$x \in X$. 

In the sequel, whenever a CS $(X, \closure)$ is quasi-discrete, and $R$
is the binary relation on $X$ which generates it, we will often write
$\closureF$ for denoting $\closure_R$, and, consequently,
$(X, \closureF)$ for denoting the closure space, abstracting from the
specification of~$R$. Moreover, we let $\closureT$ denote $\closure_{R^{-1}}$.
In the \qdcs{} $(X,\closureF\,)$ of Figure~\ref{fig:base-set} a  
set $A \subseteq X$ is shown in red. $\closureF(A)$ and $\closureT(A)$ are shown in blue and in green in  Figure~\ref{fig:forward-closure} and Figure~\ref{fig:backward-closure}, respectively. 

\begin{figure}
\centering
\subfloat[][]{
  \scalebox{0.50}{%
%
\begin{tikzpicture}[%
  blackstate/.style = {draw, circle, minimum size=2mm, fill=black},
  redstate/.style = {draw=red, circle, minimum size=2mm, fill=red},
  myedge/.style = {ultra thick, ->, >=stealth'},
  ]

  
  \node [blackstate] (s15) at (1,5) {} ;
  \node [blackstate] (s14) at (1,4) {} ;
  \node [redstate]   (s24) at (2,4) {} ;
  \node [blackstate] (s22) at (2,2) {} ;
  \node [blackstate] (s36) at (3,6) {} ;
  \node [redstate]   (s35) at (3,5) {} ;
  \node [redstate]   (s34) at (3,4) {} ;
  \node [redstate]   (s33) at (3,3) {} ;
  \node [blackstate] (s32) at (3,2) {} ;
  \node [blackstate] (s42) at (4,2) {} ;
  \node [redstate]   (s43) at (4,3) {} ;
  \node [redstate]   (s44) at (4,4) {} ;
  \node [redstate]   (s45) at (4,5) {} ;
  \node [blackstate] (s55) at (5,5) {} ;
  \node [blackstate] (s54) at (5,4) {} ;
  \node [blackstate] (s53) at (5,3) {} ;
  
  \draw [myedge] (s24) -- (s15) ;
  \draw [myedge] (s14) -- (s24) ;
  \draw [myedge] (s24) -- (s34) ;
  \draw [myedge] (s35) -- (s45) ;
  \draw [myedge] (s34) -- (s44) ;
  \draw [myedge] (s34) -- (s35) ;
  \draw [myedge] (s34) -- (s33) ;
  \draw [myedge] (s34) -- (s33) ;
  \draw [myedge] (s33) -- (s22) ;
  \draw [myedge] (s32) -- (s42) ;
  \draw [myedge] (s32) -- (s33) ;
  \draw [myedge] (s45) -- (s55) ;
  \draw [myedge] (s44) -- (s45) ;
  \draw [myedge] (s44) -- (s54) ;
  \draw [myedge] (s43) -- (s34) ;
  \draw [myedge] (s43) -- (s42) ;
  \draw [myedge] (s53) -- (s43) ;
  \draw [myedge] (s54) -- (s53) ;
\end{tikzpicture}
  } 
  \label{fig:base-set}
} \quad
\subfloat[][] {%
  \scalebox{0.50}{%
%
\begin{tikzpicture}[%
  blackstate/.style = {draw, circle, minimum size=2mm, fill=black},
  redstate/.style = {draw=red, circle, minimum size=2mm, fill=red},
  bluestate/.style = {draw=blue, fill=blue, circle, minimum size=2mm},
  myedge/.style = {ultra thick, ->, >=stealth'},
  ]

  
  \node [bluestate]  (s15) at (1,5) {} ;
  \node [blackstate] (s14) at (1,4) {} ;
  \node [bluestate]  (s24) at (2,4) {} ;
  \node [bluestate]  (s22) at (2,2) {} ;
  \node [blackstate] (s36) at (3,6) {} ;
  \node [bluestate]  (s35) at (3,5) {} ;
  \node [bluestate]  (s34) at (3,4) {} ;
  \node [bluestate]  (s33) at (3,3) {} ;
  \node [blackstate] (s32) at (3,2) {} ;
  \node [bluestate]  (s45) at (4,5) {} ;
  \node [bluestate]  (s43) at (4,3) {} ;
  \node [bluestate]  (s42) at (4,2) {} ;
  \node [bluestate]  (s44) at (4,4) {} ;
  \node [bluestate]  (s55) at (5,5) {} ;
  \node [bluestate]  (s54) at (5,4) {} ;
  \node [blackstate] (s53) at (5,3) {} ;
  
  \draw [myedge] (s24) -- (s15) ;
  \draw [myedge] (s14) -- (s24) ;
  \draw [myedge] (s24) -- (s34) ;
  \draw [myedge] (s35) -- (s45) ;
  \draw [myedge] (s34) -- (s44) ;
  \draw [myedge] (s34) -- (s35) ;
  \draw [myedge] (s34) -- (s33) ;
  \draw [myedge] (s34) -- (s33) ;
  \draw [myedge] (s33) -- (s22) ;
  \draw [myedge] (s32) -- (s42) ;
  \draw [myedge] (s32) -- (s33) ;
  \draw [myedge] (s45) -- (s55) ;
  \draw [myedge] (s44) -- (s45) ;
  \draw [myedge] (s44) -- (s54) ;
  \draw [myedge] (s43) -- (s34) ;
  \draw [myedge] (s43) -- (s42) ;
  \draw [myedge] (s53) -- (s43) ;
  \draw [myedge] (s54) -- (s53) ;
\end{tikzpicture}
  } 
  \label{fig:forward-closure}
} \quad
\subfloat[][] {%
  \scalebox{0.50}{%
%
\begin{tikzpicture}[%
  blackstate/.style = {draw, circle, minimum size=2mm, fill=black},
  redstate/.style = {draw=red, fill=red, circle, minimum size=2mm},
  greenstate/.style = {draw=DarkGreen!80, fill=DarkGreen!70, circle, minimum size=2mm},
  myedge/.style = {ultra thick, ->, >=stealth'},
  ]

  
  \node [blackstate] (s15) at (1,5) {} ;
  \node [greenstate] (s14) at (1,4) {} ;
  \node [greenstate] (s24) at (2,4) {} ;
  \node [blackstate] (s22) at (2,2) {} ;
  \node [blackstate] (s36) at (3,6) {} ;
  \node [greenstate] (s35) at (3,5) {} ;
  \node [greenstate] (s34) at (3,4) {} ;
  \node [greenstate] (s33) at (3,3) {} ;
  \node [greenstate] (s32) at (3,2) {} ;
  \node [blackstate] (s42) at (4,2) {} ;
  \node [greenstate] (s43) at (4,3) {} ;
  \node [greenstate] (s44) at (4,4) {} ;
  \node [greenstate] (s45) at (4,5) {} ;
  \node [blackstate] (s55) at (5,5) {} ;
  \node [blackstate] (s54) at (5,4) {} ;
  \node [greenstate] (s53) at (5,3) {} ;
  
  \draw [myedge] (s24) -- (s15) ;
  \draw [myedge] (s14) -- (s24) ;
  \draw [myedge] (s24) -- (s34) ;
  \draw [myedge] (s35) -- (s45) ;
  \draw [myedge] (s34) -- (s44) ;
  \draw [myedge] (s34) -- (s35) ;
  \draw [myedge] (s34) -- (s33) ;
  \draw [myedge] (s34) -- (s33) ;
  \draw [myedge] (s33) -- (s22) ;
  \draw [myedge] (s32) -- (s42) ;
  \draw [myedge] (s32) -- (s33) ;
  \draw [myedge] (s45) -- (s55) ;
  \draw [myedge] (s44) -- (s45) ;
  \draw [myedge] (s44) -- (s54) ;
  \draw [myedge] (s43) -- (s34) ;
  \draw [myedge] (s43) -- (s42) ;
  \draw [myedge] (s53) -- (s43) ;
  \draw [myedge] (s54) -- (s53) ;
\end{tikzpicture}
  } 
  \label{fig:backward-closure}
} \quad
\subfloat[][] {%
  \scalebox{0.50}{%
%
\begin{tikzpicture}[%
  blackstate/.style = {draw, circle, minimum size=2mm, fill=black},
  redstate/.style = {draw=red, circle, minimum size=2mm, fill=red},
  orangestate/.style = {draw=orange, circle, minimum size=2mm, fill=orange},
  myedge/.style = {ultra thick, ->, >=stealth'},
  ]

  
  \node [blackstate] (s15) at (1,5) {} ;
  \node [blackstate] (s14) at (1,4) {} ;
  \node [blackstate] (s24) at (2,4) {} ;
  \node [blackstate] (s22) at (2,2) {} ;
  \node [blackstate] (s36) at (3,6) {} ;
  \node [orangestate](s35) at (3,5) {} ;
  \node [orangestate](s34) at (3,4) {} ;
  \node [blackstate] (s33) at (3,3) {} ;
  \node [blackstate] (s32) at (3,2) {} ;
  \node [blackstate] (s42) at (4,2) {} ;
  \node [blackstate] (s43) at (4,3) {} ;
  \node [orangestate](s44) at (4,4) {} ;
  \node [orangestate](s45) at (4,5) {} ;
  \node [blackstate] (s55) at (5,5) {} ;
  \node [blackstate] (s54) at (5,4) {} ;
  \node [blackstate] (s53) at (5,3) {} ;
  
  \draw [myedge] (s24) -- (s15) ;
  \draw [myedge] (s14) -- (s24) ;
  \draw [myedge] (s24) -- (s34) ;
  \draw [myedge] (s35) -- (s45) ;
  \draw [myedge] (s34) -- (s44) ;
  \draw [myedge] (s34) -- (s35) ;
  \draw [myedge] (s34) -- (s33) ;
  \draw [myedge] (s34) -- (s33) ;
  \draw [myedge] (s33) -- (s22) ;
  \draw [myedge] (s32) -- (s42) ;
  \draw [myedge] (s32) -- (s33) ;
  \draw [myedge] (s45) -- (s55) ;
  \draw [myedge] (s44) -- (s45) ;
  \draw [myedge] (s44) -- (s54) ;
  \draw [myedge] (s43) -- (s34) ;
  \draw [myedge] (s43) -- (s42) ;
  \draw [myedge] (s53) -- (s43) ;
  \draw [myedge] (s54) -- (s53) ;
\end{tikzpicture}
  } 
  \label{fig:forward-interior}
} \quad
\subfloat[][] {%
  \scalebox{0.50}{%
%
\begin{tikzpicture}[%
  blackstate/.style = {draw, circle, minimum size=2mm, fill=black},
  redstate/.style = {draw=red, circle, minimum size=2mm, fill=red},
  purplestate/.style = {draw=purple!80, circle, minimum size=2mm, fill=purple!80},
  myedge/.style = {ultra thick, ->, >=stealth'},
  ]

  
  \node [blackstate] (s15) at (1,5) {} ;
  \node [blackstate] (s14) at (1,4) {} ;
  \node [blackstate] (s24) at (2,4) {} ;
  \node [blackstate] (s22) at (2,2) {} ;
  \node [blackstate] (s36) at (3,6) {} ;
  \node [purplestate]   (s35) at (3,5) {} ;
  \node [purplestate]   (s34) at (3,4) {} ;
  \node [blackstate] (s33) at (3,3) {} ;
  \node [blackstate] (s32) at (3,2) {} ;
  \node [blackstate] (s42) at (4,2) {} ;
  \node [blackstate] (s43) at (4,3) {} ;
  \node [blackstate] (s44) at (4,4) {} ;
  \node [blackstate] (s45) at (4,5) {} ;
  \node [blackstate] (s55) at (5,5) {} ;
  \node [blackstate] (s54) at (5,4) {} ;
  \node [blackstate] (s53) at (5,3) {} ;
  
  \draw [myedge] (s24) -- (s15) ;
  \draw [myedge] (s14) -- (s24) ;
  \draw [myedge] (s24) -- (s34) ;
  \draw [myedge] (s35) -- (s45) ;
  \draw [myedge] (s34) -- (s44) ;
  \draw [myedge] (s34) -- (s35) ;
  \draw [myedge] (s34) -- (s33) ;
  \draw [myedge] (s34) -- (s33) ;
  \draw [myedge] (s33) -- (s22) ;
  \draw [myedge] (s32) -- (s42) ;
  \draw [myedge] (s32) -- (s33) ;
  \draw [myedge] (s45) -- (s55) ;
  \draw [myedge] (s44) -- (s45) ;
  \draw [myedge] (s44) -- (s54) ;
  \draw [myedge] (s43) -- (s34) ;
  \draw [myedge] (s43) -- (s42) ;
  \draw [myedge] (s53) -- (s43) ;
  \draw [myedge] (s54) -- (s53) ;
\end{tikzpicture}
  } 
  \label{fig:backward-interior}
}
\caption{
A \qdcs{} $(X,\closureF)$. Arrows in the figure represent the binary relation generating the space.
A set $A \subseteq X$ is
shown in red in Fig.~\ref{fig:base-set},  
$\closureF(A)$ is shown in blue in Fig.~\ref{fig:forward-closure},
$\closureT(A)$ is shown in green in Fig.~\ref{fig:backward-closure},
$\interiorF(A)$ is shown in orange in Fig.~\ref{fig:forward-interior}, and
$\interiorT(A)$ is shown in magenta in Fig.~\ref{fig:backward-interior}.
}
\label{fig:CloIntExample}
\end{figure}

Regarding the interior operator~$\interior$, the notations
$\interiorF$ and~$\interiorT$ are defined in the obvious way:
$\interiorF(A) = \overline{\closureF (\bar{A})}$ and
$\interiorT(A) = \overline{\closureT (\bar{A})}$. 
Again, with reference to the example of Figure~\ref{fig:CloIntExample},
$\interiorF(A)$ and $\interiorT(A)$ are shown in orange and in magenta in 
Figure~\ref{fig:forward-interior} and Figure~\ref{fig:backward-interior}, respectively.

It is worth noting that, for any \qdcs{} $(X,\closure_R)$ and $A\subseteq X$ we have that
$$\interiorF(A) = A \setminus \ZET{x \in X}{\exists \mkern1mu a \in \overline{A} \colon (a,x)\in R},$$ and, similarly,
$$\interiorT(A) = A \setminus \ZET{x \in X}{\exists \mkern1mu a \in \overline{A} \colon (x,a)\in R}.$$

Clearly, for a symmetric relation~$R$ on~$X$ the closure operators
$\closureF$ and~$\closureT$ coincide, as do $\interiorF$
and~$\interiorT$.
Finally, we will often use the simplified notation $\closureF(x)$ for
$\closureF(\SET{x})$, and similarly for~$\closureT(x)$, $\interiorF(x)$ and $\interiorT(x)$.

\blankline

\noindent
The following lemma relates the operators $\closureF$ and~$\interiorT$
(and, symmetrically, $\closureT$ and $\interiorF$).

\begin{lem}
  \label{lem:swap-closure-and-interior}
  Let $(X, \closureF)$ be a \qdcs. Then $\closureF(x) \subseteq A$ iff
  $x \in \interiorT(A)$ and $\closureT(x) \subseteq A$ iff
  $x \in \interiorF(A)$, for all $x \in X$ and $A \subseteq X$.
\end{lem}

\begin{proof}
We have the following derivation: \bigskip \\
\noindent
\begin{math}
\deriv{x \in \interiorF(A)}
\hint{\Leftrightarrow}{Def.\ of $\interiorF$; Set Theory}
x \notin \closureF(\bar{A})
\hint{\Leftrightarrow}{Def.\ of $\closureF$}
\mbox{$x \notin \bar{A}$ and there is no $\bar{a} \in \bar{A}$ such
  that $(\bar{a},x) \in R$} 
\hint{\Leftrightarrow}{Logic}
\mbox{$x \in A$ and for all $x' \in X$, if $(x',x) \in R$ then $x' \in A$}
\hint{\Leftrightarrow}{Set Theory}
\closureT(x) \subseteq A \mkern1mu .
\end{math} \\
Symmetrically we obtain $x \in \interiorT(A)$ if and only if
$\closureF(x) \subseteq A$.
\end{proof}

\noindent
Lemma~\ref{lem:AltQdCS} below provides an alternative characterisation of \qdcs{s}, based on their interpretation as augmented neighbourhood spaces. We recall that $\cnrf{R}$ stands for $(R^=)^{-1}$.
\begin{lem}\label{lem:AltQdCS}
Given a  \cs{} $(X,\closure)$ and binary relation $R \subseteq X \times X$, the following holds:
$(X,\closure)$ is the \qdcs{} generated by $R$ if and only if $(X,\nu_{\closure})=(X,\nu_{\cnrf{R}})$.
\end{lem}
\begin{proof}
Suppose $(X,\closure)$  is the \qdcs{} generated by $R$, i.e. $\closure=\closure_R$ as of Theorem~\ref{thm:QDR}. 
We have to prove that $(X,\nu_{\closure_R})=(X,\nu_{\cnrf{R}})$. This amounts to showing that
$\nu_{\closure_R} = \nu_{\cnrf{R}}$, i.e. that
for all $x \in X$, $\nu_{\closure_R}(x) = \nu_{\cnrf{R}}(x)$ holds. We show the derivation below for generic $x \in X$ and $U \subseteq X$\\[0.5em]
\noindent
$
\deriv
U \in \nu_{\closure_R}(x)
\hint{\iff}{Def. of $\nu_{\closure_R}$}
x \in \interior_R(U)
\hint{\iff}{$\interior_R(U)=\overline{\closure_R(\overline{U})}$}
x \in \overline{\closure_R(\overline{U})}
\hint{\iff}{Set Theory}
x \not\in \closure_R(\overline{U})
\hint{\iff}{Def. of $\closure_R$ as of Theorem~\ref{thm:QDR} on page~\pageref{thm:QDR}}
x\not\in (\overline{U} \cup \ZET{y \in X}{\exists z \in \overline{U}: z R y})
\hint{\iff}{Set Theory}
x\not\in \overline{U} \mbox{ and } x\not\in \ZET{y \in X}{\exists z \in \overline{U}: z R y}
\hint{\iff}{Set Theory}
x\in U \mbox{ and no }  z \mbox{ exists such that } z \not\in U \mbox{ and } z R x
\hint{\iff}{Set Theory}
x\in U \mbox{ and } \ZET{y \in X}{y R x} \subseteq U
\hint{\iff}{Set Theory; Def. of $R^=$}
\ZET{y \in X}{y R^= x} \subseteq U
\hint{\iff}{Def. of $(R^=)^{-1}$; $\cnrf{R}=(R^=)^{-1}$}
\cnrf{R}[x] \subseteq U
\hint{\iff}{Def. of $\nu_{\cnrf{R}}$}
U \in \nu_{\cnrf{R}}(x)
$\\[0.5em]
\noindent
Now suppose that $(X,\nu_{\closure})=(X,\nu_{\cnrf{R}})$. This means that 
$\nu_{\closure}=\nu_{\cnrf{R}}$. We have to show that 
$(X,\closure)=(X,\closure_R)$. 
Let us consider the neighbourhood space $(X, \nu_{\closure_R})$, noting that, from
the above derivation, we know  that 
$\nu_{\closure_R} = \nu_{\cnrf{R}}$.  This means that 
$(X,\nu_{\closure_R})=(X,\nu_{\cnrf{R}})$. But, by hypothesis, we know that 
$(X,\nu_{\cnrf{R}})=(X,\nu_{\closure})$ and so we get $(X,\nu_{\closure}) = (X,\nu_{\closure_R})$.
We now use this last equality to prove that $\closure = \closure_R$. From this, it follows that
$(X,\closure) = (X,\closure_R)$.
The above equality,  $(X,\nu_{\closure}) = (X,\nu_{\closure_R})$, means  that 
$\nu_{\closure} = \nu_{\closure_R}$, that is
$$
\ZET{U\subseteq X}{x\in \interior(U)} =\nu_{\closure}(x) = 
\nu_{\closure_R}(x) = \ZET{U\subseteq X}{x\in \interior_R(U)}
$$
for all $x\in X$, where, we recall, $\interior(U)=\overline{\closure(\overline{U})}$. Suppose $\interior \not=\interior_R$, i.e. 
$\interior(A) \not=\interior_R(A)$ for some $A \subseteq X$.
Then there is an $x \in X$ such that $x \in \interior(A) \setminus \interior_R(A)$
or $x \in \interior_R(A) \setminus \interior(A)$.
Suppose $x \in \interior(A) \setminus \interior_R(A)$. This implies that 
$A \in \ZET{U\subseteq X}{x\in \interior(U)} \setminus \ZET{U\subseteq X}{x\in \interior_R(U)}$.
Similarly, if $x \in \interior_R(A) \setminus \interior(A)$, we have
$A \in \ZET{U\subseteq X}{x\in \interior_R(U)} \setminus \ZET{U\subseteq X}{x\in \interior(U)}$.
In both cases, we would have 
$\nu_{\closure}(x) \not=  \nu_{\closure_R}(x)$, which contradicts 
$(X,\nu_{\closure}) = (X,\nu_{\closure_R})$. So, we have that, for all $A \subseteq X$
we have $\interior(A) =\interior_R(A)$, i.e. 
$\interior =\interior_R$. This means
$\closure = \closure_R$ which brings us to $(X,\closure) = (X,\closure_R)$.
\end{proof}

\noindent
In the theory of closure spaces, {\em paths} play a fundamental
role. As in topology, they are defined as {\em continuous functions} from appropriate {\em index spaces} to the relevant closure spaces. An index space is a \cs{} in turn.
\begin{defi}[Continuous function]
  \label{def:ContinuousFunction}
  Function $f : X_1 \to X_2$ is a \emph{continuous} function from
  $(X_1,\closure_1)$ to $(X_2,\closure_2)$ if and only if for all sets
  $A \subseteq X_1$ we have
  $f (\closure_1(A)) \subseteq \closure_2(f(A))$.
  \closedefi
\end{defi}

In the present paper we will use paths mainly in the context of \qdcs{s}. In this case it is sufficient to use {\em discrete paths}, i.e., paths where the index space is the \qdcs{} $(\nats,\closure_{\succ})$.
Relation $\succ$ is the \emph{successor} relation on the set of natural numbers $\nats$, i.e., $\ZET{(m,n)\in \nats \times \nats}{n=m+1}$:

\begin{defi}[Discrete path]
  \label{def:path}
  A \emph{discrete path} in \qdcs{} $(X,\closureF)$ is a continuous function from
  $(\nats,\closure_{\succ})$ to $(X,\closureF)$.
  \closedefi
\end{defi}

In the sequel, we will omit ``discrete'' when referring to discrete paths.
The following lemma states some useful properties of the closure
operator as well as of paths.

\begin{lem}
  \label{lem:FT}
  For a \qdcs{} $(X, \closureF)$ it holds for all
  $A \subseteq X$ and $x_1,x_2 \in X$ that:
  \begin{enumerate} [(i)]
  \item \label{point:FT2}
    $x_1 \in \closureT(\SET{x_2})$ if and only if
    $x_2 \in \closureF(\SET{x_1})$;
    \item $\closureT(A) = \ZET{x}{\mbox{$x \in X$ and exists $a \in A$
        such that $a \in \closureF(\SET{x})$}}$;
  \item\label{point:Paths} $\pi$ is a path over $X$ if and only if for
    all $j \neq 0$ the following holds:
    $ \pi(j) \in \closureF(\pi(j{-}1))$.
\end{enumerate}
\end{lem}

\begin{proof}
  We prove only item~\ref{point:Paths} of the lemma, the proof of the
  other points being straightforward.  We show that $\pi$~is a path over
  $(X,\closureF)$ if and only if, for all $j\not=0$, we have
  $ \pi(j) \in \closureF(\pi(j{-}1))$.  Suppose $\pi$~is a path over
  $(X,\closureF)$; the following
  derivation proves the assert:\\
  $ $\\
\noindent
$
\deriv
\pi(j)
\hint{\in}{Set Theory}
\SET{\pi(j{-}1),\pi(j)}
\hint{=}{Definition of $\pi(N)$ for $N \subseteq \nats$}
\pi(\SET{j{-}1,j})
\hint{=}{Definition of $\closure_{\succ}$}
\pi(\closure_{\succ}(\SET{j{-}1}))
\hint{\subseteq}{Continuity of $\pi$}
\closureF(\pi(j{-}1))
$

\noindent
For proving the converse we have to show that for all sets 
$N \subseteq \nats \setminus \SET{0}$ we have 
$\pi(\closure_{\succ}(N)) \subseteq \closureF(\pi(N))$. 
By definition of $\closure_{\succ}$
we have that $\closure_{\succ}(N) = N \cup \ZET{j}{j{-}1 \in N}$ and so
$\pi(\closure_{\succ}(N)) = \pi(N) \cup \pi(\ZET{j}{j{-}1 \in N})$. 
By the second axiom of closure, we have $\pi(N) \subseteq \closureF(\pi(N))$.
We show that $\pi(\ZET{j}{j{-}1 \in N}) \subseteq \closureF(\pi(N))$ as well.
Take any $j$ such that $j{-}1 \in N$; we have $\SET{\pi(j{-}1)} \subseteq \pi(N)$ since $j{-}1 \in N$, and,
by monotonicity of $\closureF$ it follows that 
$\closureF(\SET{\pi(j{-}1)}) \subseteq \closureF(\pi(N))$. Consequently, since 
$\pi(j) \in \closureF(\pi(j{-}1))$ by hypothesis, we also get 
$\pi(j) \in\closureF(\pi(N))$. Since this holds for all elements of
the set $\ZET{j}{j{-}1 \in N}$ we also have 
$\pi(\ZET{j}{j{-}1 \in N}) \subseteq \closureF(\pi(N))$.
\end{proof}

For the purposes of the present paper, it is actually sufficient to
consider only \emph{finite} paths over \qdcs{s}, i.e.,\ continuous
functions having $[0;\ell \mkern1mu ]$ as a domain, for some
$\ell\in \nats$.  For such paths, it is convenient to introduce some
notation and terminology.  Given a \qdcs{} $(X,\closureF)$ and a path
$\pi : [0;\ell \mkern1mu ] \to X$, we call~$\ell$ the \emph{length}
of~$\pi$ and often use the sequence notation $(x_i)_{i=0}^{\ell}$,
where~$x_i= \pi(i)$ for all $i \in [0;\ell
\mkern1mu]$. 
More precisely, on the basis of Lemma~\ref{lem:FT}\ref{point:Paths},
we say that 
$(x_i)_{i=0}^{\ell}$ is a \emph{forward path from}~$x_0$ if $x_{i+1} \in \closureF(x_i)$ for
$i \in [0;\ell \mkern1mu )$ and, similarly, we say that it
is a \emph{backward path from}~$x_0$ if $x_{i+1}\in \closureT(x_i)$ for
$i \in [0;\ell \mkern1mu )$. An example of forward and backward paths is shown in 
Figure~\ref{fig:ForwBackwPath}. In the sequel, we avoid to specify  ``from''  which point 
a (forward / backward) starts, when this does not cause confusion.

\begin{figure}
     \begin{tikzpicture}
       \tikzset{->}
       \tikzstyle{point}=[circle,draw=black,fill=white,inner sep=0pt,minimum width=4pt,minimum height=4pt]
          \node (p1)[point] at (0,0){ $\;x_1\;$ }; 
          \node (p2)[point, right of=p1, xshift=1cm] { $\;x_2\;$ }; 
          \node (p3)[point, right of=p2, xshift=1cm] { $\;x_3\;$ }; 
          \draw [thick] (p1) edge (p2);
          \draw [thick] (p2) edge (p3);         
     \end{tikzpicture}
\caption{\label{fig:ForwBackwPath} A \cs{} $(X,\closure_R)$ where $X=\SET{x_1, x_2, x_3}$ and 
$R=\SET{(x_1,x_2),(x_2,x_3)}$ (in the figure, only the points and relation $R$ are shown); we have that $(x_1,x_1,x_1,x_2,x_2,x_2,x_3,x_3)$ is a forward path from $x_1$ and 
$(x_3,x_3,x_2, x_1, x_1, x_1, x_1)$ is a backward path from $x_3$.}
\end{figure}

Given forward paths
$\pi' = (x'_i)_{i=0}^{\ell'}$ and $\pi'' = (x''_i)_{i=0}^{\ell''}$,
with $x'_{\ell'} = x''_0$, the \emph{concatenation} $\pi' \cdot \pi''$
of $\pi'$ with~$\pi''$ is the path $\pi: [0;\ell' + \ell''] \to X$
from~$x'_0$, such that $\pi(i) = \pi'(i)$ if $i \in [0;\ell']$ and
$\pi(i) = \pi''(i-\ell')$ if $i \in [\ell'; \ell'+ \ell'']$;
concatenation for backward paths is defined similarly.  For a path
$\pi=(x_i)_{i=0}^n$ and $k \in [0,n]$ we define the $k$-shift operator
$\_{\mkern2mu \uparrow} k$ as follows:
$\pi \mathord{\uparrow} k = (x_{j+k})_{j=0}^{n-k}$;
furthermore, a {\em (non-empty) prefix} of $\pi$ is a path $\pi|[0,k]$, for some $k\in [0;n]$.
We let $\fpthsF$ denote the set of all finite forward paths
of~$(X,\closureF)$.


\section{Bisimilarity for Closure Models}\label{sec:CMbisimilarity}

In this section, we introduce the first notion of spatial bisimilarity that we consider, namely {\em closure model bisimilarity}, \cm-bisimilarity for short, for which we also provide a logical
characterisation.
The definition of \cm-bisimilarity is an instantiation to \cm{s} of
monotonic bisimulation on neighbourhood
models~\cite{vB+17,Han03}. Consequently, it is defined using
the interior operator.

\begin{defi}[\cm-bisimilarity]
  \label{def:cm-bisimulation}
  Given a \cm{} $\model = (X, \closure, \peval)$, a symmetric relation
  $B \subseteq X \times X$ is called a \emph{\cm-bisimulation}
  for~$\model$ if for all $x_1, x_2 \in X$ such that~$(x_1,x_2)\in B$ the
  following holds:
  \begin{enumerate}
  \item \label{CM1} $x_1 \in \peval(p)$ if and only if
    $x_2 \in \peval(p)$, for all $p\in\ap$.
  \item\label{CM2} If $x_1 \in \interior(S_1)$ for a subset
    $S_1 \subseteq X$, then for some subset $S_2 \subseteq X$ it holds
    that (i)~$x_2 \in \interior(S_2)$ and (ii)~for each $s_2 \in S_2$
    there $s_1 \in S_1$ such that~$(s_1,s_2)\in B$.
  \end{enumerate}
  Two points $x_1, x_2 \in X$ are called \cm-bisimilar in~$\model$ if
  $(x_1,x_2)\in B$ for some \cm-bisimulation~$B$ for~$\model$. Notation,
  $x_1 \cmbis^{\model} x_2$.  
\closedefi
\end{defi}

\begin{exa}\label{ex:AnExample}
  Consider the closure model $\model$ depicted in
  Figure~\ref{fig:example-cm-bisim} where
  $\model = (X,\closure,\peval)$ with
  $X = \SET{x_1,x_2,x_3,y_1,y_2,z_2,t_1,t_2,t_3,u_1,u_2,v_1,v_2,v_3}$,
  $\closure$ as induced by the binary relation as depicted in the
  figure, and with atomic propositions $\mathit{red}$,
  $\mathit{blue}$, and $\mathit{green}$ and valuation
  function~$\peval$ satisfying
  $\peval(\mathit{red}) = \SET{x_1, x_2, x_3, v_1, v_2, v_3}$,
  $\peval(\mathit{blue}) = \SET{y_1, y_2, z_2, t_2}$, and
  $\peval(\mathit{green}) = \SET{y_3, u_1, u_2, t_1, t_3}$.
  
  The relation $B_{12} = \SET{ \, (x_1,x_2) ,\, (x_2,x_1) \, }$ is a
  \cm-bisimulation relation relating $x_1$ and~$x_2$. E.g., for the set
  $S_1 = \SET{x_1,y_1}$ satisfying
  $x_1 \in \SET{x_1,y_1} = \interior(S_1)$ we can choose
  $S_2 = \SET{x_2}$ to match $S_1$.
  Similarly, the relation
  $B_{13} = \SET{ \, (x_1,x_3) ,\, (x_3,x_1) \, }$ is a
  \cm-bisimulation relation relating $x_1$ and~$x_3$ despite the green
  colour~$\invpeval(\SET{y_3})$ of~$y_3$.

  Also the points $v_1$ and~$v_3$ are CM-bisimilar, but both $v_1$
  and~$v_3$ are not CM-bisimilar to~$v_2$. The set
  $S_1 = \SET{t_1, u_1, v_1}$ has~$v_1$ in its interior. The smallest
  set~$S_2$ such that $v_2 \in \interior(S_2)$ is
  $\SET{t_2, u_2, v_2}$. However, for $t_2 \in \interior(S_2)$ there
  exists no counterpart in~$S_1$.
\end{exa}

\begin{figure}[htbp]
  \centering

\begin{tikzpicture}[baseline=1,
  ->, >=stealth', semithick, shorten >=0.5pt,
  every state/.style={
    circle, minimum size=8pt, inner sep=2pt, draw},
  ]
  
  \node [state, fill=red, label={+180:{$x_1$}}] (x1) at (1,1) {} ;
  \node [state, fill=blue, label={+0:{$y_1$}}] (y1) at (3,1) {} ;
  \draw (x1) edge (y1) ;

  \begin{scope}[xshift=+3.5cm]
    \node [state, fill=red, label={+180:{$x_2$}}] (x2) at (1,1) {} ;
    \node [state, fill=blue, label={+0:{$y_2$}}] (y2) at (3,1.67) {} ;
    \node [state, fill=blue, label={-0:{$z_2$}}] (z2) at (3,0.33) {} ;
    \draw (x2) edge (y2) ;
    \draw (x2) edge (z2) ;
  \end{scope}

  \begin{scope}[xshift=+1.75cm, yshift=-1.67cm]
    \node [state, fill=red, label={+180:{$x_3$}}] (x3) at (1,1) {} ;
    \node [state, fill=green, label={+0:{$y_3$}}] (y3) at (3,1) {} ;
    \draw (x3) edge (y3) ;
  \end{scope}
\end{tikzpicture}
\qquad
\begin{tikzpicture}[baseline=1,
  ->, >=stealth', semithick, shorten >=0.5pt,
  every state/.style={
    circle, minimum size=8pt, inner sep=2pt, draw},
  ]
  
  \node [state, fill=green, label={+180:{$t_1$}}] (t1) at (1,1.67) {} ;
  \node [state, fill=green, label={+180:{$u_1$}}] (u1) at (1,0.33) {} ;
  \node [state, fill=red, label={+0:{$v_1$}}] (v1) at (3,1) {} ;
  \draw (t1) edge (v1) ;
  \draw (u1) edge (v1) ;

  \begin{scope}[xshift=+3.5cm]
    \node [state, fill=blue, label={+180:{$t_2$}}] (t2) at (1,1.67) {} ;
    \node [state, fill=green, label={+180:{$u_2$}}] (u2) at (1,0.33) {} ;
    \node [state, fill=red, label={+0:{$v_2$}}] (v2) at (3,1) {} ;
    \draw (t2) edge (v2) ;
    \draw (u2) edge (v2) ;
  \end{scope}

  \begin{scope}[xshift=+1.75cm, yshift=-1.67cm]
    \node [state, fill=green, label={+180:{$t_3$}}] (t3) at (1,1) {} ;
    \node [state, fill=red, label={+0:{$v_3$}}] (v3) at (3,1) {} ;
    \draw (t3) edge (v3) ;
  \end{scope}
\end{tikzpicture}
  \caption{$x_{1}$, $x_{2}$, and $x_{3}$ CM-bisimilar; $v_1$ and $v_3$
    not CM-bisimilar to~$v_2$}
  \label{fig:example-cm-bisim}
\end{figure}

The following lemma  follows directly from the relevant definitions:
\begin{lem}\label{lem:CMandNei}
Given a \cm{} $\model=(X, \closure, \peval)$ and neighbourhood model
$\calJ=(X,\nu_{\closure},\peval)$ with $\nu_{\closure}(x)=\ZET{U\subseteq X}{x \in \interior(U)}$ for all $x\in X$,
the following holds: for all $x_1,x_2\in X$,
$x_1$ and $x_2$ are \cm-bisimilar as elements of $\model$ if and only if they are 
neighbourhood bisimilar as elements of $\calJ$.
\end{lem}

\noindent
We will show that \cm-bisimilarity is characterized by an infinitary
modal logic called~\iml\@.
The result is an extension to (possibly) infinite structures of the known HM result for finite spaces~\cite{vB+17}. As we will see, in the proof of the HM result, we use a technique similar to that of van Benthem et al.\ but we need infinite conjunction in order to quantify over all elements of the structure.

In the present spatial context the
classical modality~$\Diamond$ is interpreted as a {\em proximity} modality, 
and it is denoted by~$\lnear$ standing for being ``near''.  The definition of \iml{} below is taken
from~\cite{Ci+20}.

\begin{defi}[\iml{}]
  \label{def:Iml}
  The abstract language of the modal logic \iml{} is given
    by
    \begin{displaymath}
      \form ::= p \mid \lneg \form \mid \liand_{i \in I} \: \form_i \mid
      \lnear \mkern1mu \form 
    \end{displaymath}
    where $p$ ranges over $\ap$ and $I$ ranges over an appropriate
    collection of index sets.
    
  The satisfaction relation for \iml{} with respect to a given
    \cm~$\model = (X, {\closure}, {\peval})$ with point~$x$ is
    recursively defined on the structure of $\form$ as follows:
    \begin{displaymath}
      \def\arraystretch{1.2}
      \begin{array}{r c l c l c l l}
        \model,x & \models_{\iml} & p & \text{iff} & x  \in \peval(p) \\
        \model,x & \models_{\iml}  & \lneg \,\form & \text{iff} & \model,x  \models_{\iml} \form \mbox{ does not hold} \\
        \model,x & \models_{\iml}  & \liand_{i\in I} \: \form_i & \text{iff} & \model,x  \models_{\iml} \form_i \mbox{ for all } i \in I \\
        \model,x & \models_{\iml}  & \lnear \form & \text{iff} & x \in \closure(\sem{\form}^{\model})
      \end{array}
      \def\arraystretch{1.0}
    \end{displaymath}
    with $\sem{\form}^{\model} = \ZET{x \in X}{\model,x
      \models_{\iml} \form}$.
\closedefi
\end{defi}

\begin{exa}
  As a simple illustration of the above definition, again with
  reference to the \cm{} of
  Figure~\ref{fig:example-cm-bisim} that we have used for Example~\ref{ex:AnExample}, it holds that
  $\sem{\mathit{blue}} = \SET{y_1, y_2, z_2, t_2}$. Hence
  $\closure( \sem{\mathit{blue}} ) = \sem{\mathit{blue}} \cup
  \SET{v_2}$. Therefore we have that
  $v_2 \models \lnear \mathit{blue}$ and
  $v_1,v_3 \nmodels \lnear \mathit{blue}$.
\end{exa}

\noindent
The logic \iml{} induces an equivalence on the carrier of a \cm{} in
the usual way.

\begin{defi}[\iml-equivalence]
  \label{def:ImlEq}
  For a \cm{} $\model$, the relation
  $\imleq^{\model} \, \subseteq \, X \times X$ is defined by
  \begin{displaymath}
    \begin{array}{rcl}
      x_1 \imleq^{\model} x_2
      & \text{iff}
      & \text{$\model, x_1 \models_{\iml} \form \Leftrightarrow
        \model, x_2 \models_{\iml} \form$, for all $\form \in \iml$.} 
    \end{array} \vspace{-0.25in}
  \end{displaymath}
\closedefi
\end{defi}

\noindent
We establish coincidence of CM-bisimilarity and \iml-equivalence in two
steps. First we prove that \iml-equivalence~$\imleq^{\model}$ includes
\cm-bisimilarity~$\cmbis^{\model}$. 

\begin{lem}
  \label{lem:CMbisIsImleq}
  For all points $x_1, x_2$ in a \cm~$\model$, if $x_1 \cmbis^{\model} x_2$
  then $x_1 \imleq^{\model} x_2$.
\end{lem}

\begin{proof}
  We proceed by induction on the structure of~$\form$ to show that
  $x_1 \models \form$ implies $x_2 \models \form$ when
  $x_1 \cmbis x_2$ for $x_1, x_2 \in X$, where we only consider the
  case for~$\lnear \form$, the others being straightforward, and we
  refrain from writing the superscript~${\model}$ for the sake of
  readability.
  
  So, suppose $x_1, x_2 \in X$ such that $x_1 \cmbis x_2$ and
  $x_1 \models \lnear \form$. By definition of satisfaction for~\iml,
  we have that $x_1 \in \closure \sem{\form}$. Equivalently, by
  Theorem~\ref{thm-point-in-closure}, every neighbourhood of~$x_1$
  intersects~$\sem{\form}$.
  Now, let $S_2$ be a neighbourhood of~$x_2$. Since $x_1 \cmbis x_2$,
  for some neighbourhood~$S_1$ of~$x_1$ it holds that for each
  point~$s_1 \in S_1$ a CM-bisimilar point~$s_2$ in~$S_2$ exists,
  i.e.,~$s_1 \cmbis s_2$. Let $x'_1 \in S_1 \cap \sem{\form}$ and
  $x'_2 \in S_2$ be such that $x'_1 \cmbis x'_2$. Since
  $x'_1 \in \sem{\form}$, we have $x'_1 \models \form$, and because
  $x'_1 \cmbis x'_2$, also $x'_2 \models \form$ by induction
  hypothesis. As $x'_2 \in S_2 \cap \sem{\form}$,
  $S_2 \cap \sem{\form}$ is non-empty.  Again with appeal to
  Theorem~\ref{thm-point-in-closure} we obtain
  $x_2 \in \closure \sem{\form}$, i.e.,\ $x_2 \models \lnear \form$, as
  was to be verified.
\end{proof}

\noindent
In order to obtain an inclusion into the other direction, we follow
the argument as in~\cite{vB+17} and introduce the auxiliary notion of
a characteristic formula for a point in a CM.

\begin{lem}
  \label{lem:ImleqIsCMbis}
  For a \qdcm~$\model$, it holds that $\imleq^{\model}$ is a
  \cm-bisimulation for~$\model$.
\end{lem}

\begin{proof}
  Suppose $\model = (X,\closureF,\peval)$. For~$x,y \in X$, 
  let  \iml-formula~$\delta_{x,y}$ be
  defined as follows: if $x \imleq y$, then set $\delta_{x,y}= \ltrue$, otherwise pick some 
  \iml-formula~$\psi$ such that $\model,x \models \psi$ and $\model,y \models \lneg\psi$, and 
  set $\delta_{x,y}= \psi$.
  For a point~$x$
  of~$\model$  define $\chi(x)$ as follows:
  $\chi(x) = \liand_{y \in X} \: \delta_{x,y}$. 
  It can be
  straightforwardly verified that 
  (i)~$y \models \chi(x)$ iff $x \imleq y$, and
  (ii)~$S \subseteq \sem{\lior_{s\in S} \: \chi(s)}$ for all $S \subseteq X$.
  Also, for a point~$x$ of~$\model$ and \iml-formula~$\form$ it holds
  that
  \begin{equation}
    \label{eq-interior-neg-lnear-neg}
    x \in \interior (\sem{\form})
    \quad \text{iff} \quad
    x \models \neg \mkern2mu \lnear \neg \mkern1mu \form
    \, .
  \end{equation}
  To see this, observe that
  $\overline{ \sem{\form} } = \sem{ \neg \mkern1mu \form
  }$. Consequently, $x \in \interior (\sem{\form})$ iff
  $x \notin \closureF (\overline{ \sem{\form} } )$ iff
  $x \notin \closureF (\sem{ \neg \mkern1mu \form })$ iff
  $x \not\models \lnear \neg \mkern1mu \form$ iff
  $x \models \neg \mkern3mu \lnear \neg \mkern1mu \form$.

  Assume~$x_1 \imleq x_2$. We check that the two conditions of
  Definition~\ref{def:cm-bisimulation} are fulfilled. As to the first
  condition, let~$p \in \ap$. Because $x_1 \imleq x_2$, we have
  $x_1 \models p$ iff $x_2 \models p$, thus $x_1 \in \peval(p)$ iff
  $x_2 \in \peval(p)$.

  As to the second condition, if $S_1 \subseteq X$ is a neighbourhood of~$x_1$,
  i.e.,\ $x_1 \in \interior(S_1)$, we put
  $S_2 = \ZET{s_2}{\exists \mkern1mu s_1 \in S_1 \colon s_1 \imleq
    s_2}$. Clearly, for $s_2 \in S_2$, $s_1 \in S_1$ exists  such
  that~$s_1 \imleq s_2$. Therefore, it suffices to verify that $S_2$
  is a neighbourhood of~$x_2$, i.e.,~$x_2 \in \interior (S_2)$.

  Put $\form = \bigvee_{s_1 {\in} S_1} \: \chi(s_1)$.
  To see that $S_2 = \sem{\form}$ we argue as follows:\\
  Case $S_2 \subseteq \sem{\form}$: For~$s_2 \in S_2$ we can choose
  $s_1 \in S_1$ such that $s_1 \imleq s_2$. Since
  $s_1 \models \chi(s_1)$ by property~(i) above, also
  $s_1 \models \form$. Thus $s_2 \models \form$, since
  $s_1 \imleq s_2$.\\
  Case $ \sem{\form} \subseteq S_2$: If $s \models \Phi$,
  then $s \models \chi(s_1)$ for some~$s_1 \in S_1$. Thus
  $s \imleq s_1$ for some~$s_1 \in S_1$ by property~(ii) above
  and~$s \in S_2$. This proves that $S_2 = \sem{\form}$.
  
  We continue to verify that $x_2 \in \interior(S_2)$. It holds that
  $x_1 \models \lneg \mkern2mu \lnear \mkern-1mu \lneg \mkern1mu
  \form$. To see this, note that $S_1 \subseteq \sem{\form}$ by
  property~(iii) above. Now, $S_1$~is a neighbourhood of~$x_1$. Hence
  $x_1 \in \interior(S_1) \subseteq \interior (\sem{\form})$ and
  $x_1 \models \lneg \mkern2mu \lnear \mkern-1mu \lneg \mkern1mu
  \form$ by Equation~\ref{eq-interior-neg-lnear-neg}.
  Now, since $x_1 \imleq x_2$ it follows that
  $x_2 \models \lneg \mkern2mu \lnear \mkern-1mu \lneg \mkern1mu
  \form$. Therefore, $x_2 \in \interior (\sem{\form})$ again by
  Equation~\ref{eq-interior-neg-lnear-neg}. Thus
  $x_2 \in \interior(S_2)$, because $S_2 = \sem{\form}$.
  \end{proof}

\noindent
Summarising the above, we have the following correspondence result of
CM-bisimilarity vs.\ \iml-equivalence.

\begin{thm}
  \label{thm:CMbisEqImleq}
  \iml-equivalence $\imleq^{\mkern1mu \model}$ coincides with
  \cm-bisimilarity $\cmbis^{\model}$, for all \qdcm~$\model$.
\end{thm}

\begin{proof}
  Lemma~\ref{lem:CMbisIsImleq} yields
  $\mathord{\cmbis^{\model}} \subseteq \mathord{\imleq^{\mkern1mu
      \model}}$. Lemma~\ref{lem:ImleqIsCMbis} yields
  $\mathord{\imleq^{\mkern1mu \model}} \subseteq
  \mathord{\cmbis^{\model}}$.
\end{proof}

\noindent
An obvious consequence of Theorem~\ref{thm:CMbisEqImleq} follows.

\begin{cor}
   \label{cor:CMbisEquiv}
For all  \qdcm{s} $\model$, $\cmbis^{\model}$ is an equivalence relation. 
\end{cor}

\begin{rem}
  The notion of CM-bisimilarity as given by
  Definition~\ref{def:cm-bisimulation} is the natural generalisation
  to closure models of the notion of topo-bisimilarity for
  topological models~\cite{vBB07}. The latter models are similar to
  closure models except that the underlying set is equipped with the open sets of a
  topology rather than the closed sets derived from a closure operator.
    We  note that a topological space~$(X,\tau)$ gives rise to a
    closure space with idempotent closure
    operator~$\closure_\tau$. Therefore, a topological model
    $(X,\tau,\peval)$ can be seen as a closure model
    $(X,\closure_\tau,\peval)$.
  For a topological model the definition of a Topo-bisimulation
  of~\cite{vBB07} and Definition~\ref{def:cm-bisimulation} coincide,
  i.e.,\ a relation~$B$ on~$X$ is a Topo-bisimulation if and only if it
  is a CM-bisimulation.
  However, closure spaces are more general than topological spaces. As
  we anticipated in Section~\ref{sec:Preliminaries}, it holds that a
  closure operator induces a topology if and only if it is idempotent,
  cf.~\cite{Cec66,Gal03}. Hence, not every closure space is a
  topological space and not every closure model is a topological
  model.
Finally, we point out that, given the close relationship between 
closure models and neighbourhood models (see, e.g. Lemma~\ref{lem:CMandNei}),
Theorem~\ref{thm:CMbisEqImleq} is the counterpart, for closure models,  of the
corresponding Hennessy-Milner
result for neighbourhood models~\cite{HKP09}.
\end{rem}


\section{\cmc-bisimilarity for \qdcm{s}}
\label{sec:CMCbisimilarity}

Definition~\ref{def:cm-bisimulation} of the previous section
defines \cm-bisimilarity of a closure model in terms of
its interior operator~$\interior$. In the case of
\qdcm{s}, an alternative formulation can
be given that uses the closure operator explicitly and directly, as we
will see below. This formulation exploits the symmetric nature of the
operators in such models.

\begin{defi}
  \label{def:cm-bisimulation-alt}
  Given a \qdcm{} $\model = (X, \closureF, \peval)$, a symmetric
  relation $B \subseteq X \times X$ is a \cm-bisimulation for~$\model$
  if, whenever $(x_1, x_2)\in B$, the following holds:
  \begin{enumerate}
  \item \label{CBCM1} for all $p\in\ap$ we have $x_1 \in \peval(p)$ if
    and only if $x_2 \in \peval(p)$;
  \item \label{CBCM2} for all $x'_1$ such that
    $x_1 \in \closureF(x'_1)$ exists $x'_2$ such that
    $x_2 \in \closureF(x'_2)$ and $(x'_1,x'_2)\in B$.
    \closedefi
\end{enumerate}
\end{defi}

\noindent
The above definition is justified by the next lemma.

\begin{lem}
  \label{lem:CM-def-with-interior-vs-closure}
  Let $\model = (X, \closureF, \peval)$ be a \qdcm{} and $B \subseteq X
  \times X$ a relation. It holds that $B$~is a \cm-bisimulation according to 
  Definition~\ref{def:cm-bisimulation} if and only if
  $B$~is a  \cm-bisimulation according to 
  Definition~\ref{def:cm-bisimulation-alt}.
\end{lem}

\begin{proof}
  (\textit{if}) Assume that $B$ is a \cm-bisimulation in the sense of
  Definition~\ref{def:cm-bisimulation}. Let $x_1, x_2 \in X$ such that
  $(x_1,x_2)\in B$. We verify condition~(\ref{CBCM2}) of
  Definition~\ref{def:cm-bisimulation-alt}: Let $x'_1 \in X$ such that
  $x_1 \in \closureF(x'_1)$. Hence $x'_1 \in \closureT(x_1)$, by
  Lemma~\ref{lem:FT}\ref{point:FT2}. For $S_2 = \closureT(x_2)$ we
  have $x_2 \in \interiorF(S_2)$ by
  Lemma~\ref{lem:swap-closure-and-interior}. By condition~(\ref{CM2})
  of Definition~\ref{def:cm-bisimulation}, with the roles of $x_1$
  and~$x_2$, and of $S_1$ and~$S_2$ interchanged, a subset
  $S_1 \subseteq X$ exists such that $x_1 \in \interiorF(S_1)$ and,
  for each $s_1 \in S_1$, $s_2 \in S_2$ exists such
  that~$(s_1,s_2)\in B$. In particular, there is
  $x'_2 \in S_2 = \closureT(x_2)$ such that $(x'_1,x'_2)\in B$. Thus
  $x'_2 \in X$ exists such that $x_2 \in \closureF(x'_2)$ and
  $(x'_1,x'_2)\in B$.

  (\textit{only if}) Assume that $B$ is a closure-based
  \cm-bisimulation in the sense of
  Definition~\ref{def:cm-bisimulation-alt}. Let $x_1, x_2 \in X$ be such
  that~$(x_1,x_2)\in B$. We verify condition~(\ref{CM2}) of
  Definition~\ref{def:cm-bisimulation}. Suppose subset
  $S_1 \subseteq X$ is such that $x_1 \in \interiorF(S_1)$. By
  Lemma~\ref{lem:swap-closure-and-interior} we have
  $\closureT(x_1) \subseteq S_1$. Let
  $S_2 = \closureT(x_2)$. Then $x_2 \in \interiorF(S_2)$,
  again by Lemma~\ref{lem:swap-closure-and-interior}. By
  condition~(\ref{CBCM2}) of Definition~\ref{def:cm-bisimulation-alt},
  for each $x'_2 \in \closureT(x_2)$ a point
  $x'_1 \in \closureT(x_1)$ exists such that~$(x'_1,x'_2)\in B$. Since
  $S_2 = \closureT(x_2)$ and $\closureT(x_1) \subseteq S_1$, 
  it follows that for each $s_2 \in S_2$ and  there is
  $s_1 \in S_1$ such that~$(s_1,s_2)\in B$.
\end{proof}

\begin{rem}
An alternative proof for Lemma~\ref{lem:CM-def-with-interior-vs-closure} can be carried out by exploiting
the relationship between \qdcm{s} and neighbourhood models on the one hand, and 
Kripke models and neighbourhood models on the other hand.
In fact, it is easy to see
that $x_1$ and $x_2$ are \cm-bisimilar in $(X,\closure_{R},\peval)$ according to Definition~\ref{def:cm-bisimulation-alt} if and only if they are Kripke bisimilar in $(X,\cnrf{R},\peval)$. Moreover,
by Lemma~\ref{KbisAndNbis}, they are Kripke bisimilar in $(X,\cnrf{R},\peval)$ if and only if
they are neighbourhood bisimilar in $(X,\nu_{\cnrf{R}},\peval)$ and, by Lemma~\ref{lem:AltQdCS},
this means they are also neighbourhood bisimilar in $(X,\nu_{\closure_R},\peval)$,
since $(X,\nu_{\closure_R})=(X,\nu_{\cnrf{R}})$. 
This, by Lemma~\ref{lem:CMandNei}, means exactly that 
$x_1$ and $x_2$ are \cm-bisimilar in $(X,\closure_{R},\peval)$ according to Definition~\ref{def:cm-bisimulation}
\end{rem}

\noindent
When dealing with \qdcm{s}, we can exploit the
symmetric nature of the operators involved. Recall that, whenever
$\model$~is quasi-discrete, there are actually two interior functions,
namely $\interiorF(S)$ and $\interiorT(S)$.  It is then natural to
use both functions for the definition of a notion of
\cm-bisimilarity specifically designed for \qdcm{s}, namely {\em
  \cm-bisimilarity with converse}, \cmc-bisimilarity for short,
presented below.

\begin{defi}[\cmc-bisimilarity] \label{def:cmc-bisimulation} Given \qdcm{}
  $\model = (X, \closureF, \peval)$, a symmetric relation
  $B \subseteq X \times X$ is a \emph{\cmc-bisimulation} for~$\model$
  if, whenever $(x_1,x_2)\in B$, the following holds:
  \begin{enumerate}
  \item \label{CMC1} for all $p\in\ap$ we have $x_1 \in \peval(p)$ if
    and only if $x_2 \in \peval(p)$;
  \item \label{CMC2} for all $S_1 \subseteq X$ such that
    $x_1 \in \interiorF(S_1)$ there is $S_2 \subseteq X$ such that
    $x_2 \in \interiorF(S_2)$ and for all $s_2 \in S_2$, there is
    $s_1 \in S_1$ with $(s_1,s_2)\in B$;
  \item \label{CMC3} for all $S_1 \subseteq X$ such that
    $x_1 \in \interiorT(S_1)$ there is $S_2 \subseteq X$ such that
    $x_2 \in \interiorT(S_2)$ and for all $s_2 \in S_2$, there is
    $s_1 \in S_1$ with $(s_1,s_2)\in B$.
  \end{enumerate}
  Two points $x_1,x_2 \in X$ are called CMC-bisimilar in $\model$, if
  $(x_1,x_2)\in B$ for some \cmc-bisimulation $B$ for
  $\model$. Notation, $x_1\,\cmcbis^{\model}\, x_2$.
  \closedefi
\end{defi}

\noindent
For a \qdcm~$\model$, as for \cm-bisimilarity, we have that
\cmc-bisimilarity 
is a \cmc-bisimulation itself,
viz.\ the largest \cmc-bisimulation for~$\model$, thus including each
\cmc-bisimulation for~$\model$. 

Also for \cmc-bisimilarity, a formulation directly in terms of
closures is useful\footnote{The notion characterized by
  Definition~\ref{def:cmc-bisimulation-alt} is called \cl-bisimilation
  in~\cite{Ci+21}}.

\begin{defi} \label{def:cmc-bisimulation-alt} Given a \qdcm{}
  $\model = (X, \closureF, \peval)$, a symmetric relation
  $B \subseteq X \times X$ is a \cmc-bisimulation for~$\model$ if,
  whenever $(x_1, x_2)\in B$, the following holds:
  \begin{enumerate}
  \item \label{CBCMC1} for all $p\in\ap$ we have $x_1 \in \peval(p)$
    if and only if $x_2 \in \peval(p)$;
  \item \label{CBCMC2} for all $x'_1 \in \closureF(x_1)$ there is
    $x'_2\in \closureF(x_2)$ such that $(x'_1,x'_2)\in B$;
  \item \label{CBCMC3} for all $x'_1 \in \closureT(x_1)$ there is
    $x'_2\in \closureT(x_2)$ such that $(x'_1,x'_2)\in B$.
    \closedefi
  \end{enumerate}
\end{defi}

\noindent
The next lemma shows the interchangeability of
Definitions~\ref{def:cmc-bisimulation}
and~\ref{def:cmc-bisimulation-alt}. The proof is in essence the same
as that of Lemma~\ref{lem:CM-def-with-interior-vs-closure}. Of course,
 the proof of Lemma~\ref{lem:CMCbisimilarityEqCMCbisimilarity-alt} can 
also be carried out by exploiting the connection 
between closure spaces and neighbourhood spaces.

\begin{lem} \label{lem:CMCbisimilarityEqCMCbisimilarity-alt}
  Let $\model = (X, \closureF, \peval)$ be a \qdcm{} and $B \subseteq X
  \times X$ a relation. It holds that 
  $B$~is a \cmc-bisimulation 
  according to  Definition~\ref{def:cmc-bisimulation}
  if and only if
  $B$~is a  \cmc-bisimulation
  according to  Definition~\ref{def:cmc-bisimulation-alt}. \qed
\end{lem}

\begin{exa}
  From the definitions of \cm-bisimulation and \cmc-bisimulation it
  can be immediately observed that in a closure model based on a
  quasi-discrete closure two points that are \cmc-bisimilar are also
  \cm-bisimilar. The reverse does not hold in general. In
  Figure~\ref{fig:example-cm-bisim}, the points $x_1$ and~$x_2$ on the
  one hand and the point~$x_3$ on the other hand are \cm-bisimilar,
  but not \cmc-bisimilar. E.g., the points $y_1 \in \closureF(x_1)$
  and $y_2 \in \closureF(x_2)$ do not have a match in
  $\closureF(x_3) = \SET{x_3,y_3}$.
\end{exa}

\begin{rem}
Definition~\ref{def:cmc-bisimulation-alt} was proposed
  originally in~\cite{Ci+20}, in a slightly different form, and
  resembles the notion of (strong) back-and-forth bisimulation
  of~\cite{De+90}, in particular for the presence of condition
  (\ref{CBCMC3}).
\end{rem}

In order to provide a logical characterisation of \cmc-bisimilarity,
we extend \iml{} with a ``converse'' of its modal operator. The
resulting logic is called \emph{Infinitary Modal Logic with Converse}
\imlc{}, a logic including the two spatial proximity modalities $\lnearF$, expressing forward ``near'', and~$\lnearT$, expressing backward ``near''.

\begin{defi}[\imlc{}]\label{def:Imlc} \mbox{}
  The abstract language of \imlc{} is defined as follows:
    \begin{displaymath}
      \form ::= p \mid \lneg \form \mid \textstyle{\liand_{i \in I}} \:
      \form_i \mid \lnearF \mkern1mu \form \mid \lnearT \mkern1mu \form
    \end{displaymath}
    where $p$ ranges over~$\ap$ and $I$ ranges over an appropriate
    collection of index sets.
  The satisfaction relation with respect to a \qdcm~$\model$,
    point~$x \in \model$, and \imlc{} formula~$\form$ is defined
    recursively on the structure of $\form$ as follows:
    \begin{displaymath}
      \def\arraystretch{1.2}
      \begin{array}{r c l c l c l l}
        \model,x & \models_{\imlc}  & p & \Leftrightarrow & x  \in \peval(p)
        \\ \model,x & \models_{\imlc}  & \lneg \,\form & \Leftrightarrow & \model,x  \models_{\imlc} \form \mbox{ does not hold}
        \\ \model,x & \models_{\imlc}  & \liand_{i\in I} \form_i  & \Leftrightarrow &
                                                                                   \model,x  \models_{\imlc} \form_i \mbox{ for all } i \in I
        \\ \model,x & \models_{\imlc}  & \lnearF \mkern1mu \form & \Leftrightarrow & x \in \closureF(\sem{\form}^{\model})
        \\ \model,x & \models_{\imlc}  & \lnearT \mkern1mu \form & \Leftrightarrow & x \in \closureT(\sem{\form}^{\model})
      \end{array}
      \def\arraystretch{1.0}
    \end{displaymath}
    with $\sem{\form}^\model = \ZET{x \in X}{\model,x \models_{\imlc} \form}$.
    {\closedefi}
\end{defi}

\begin{exa}
  Consider the
  \qdcm{} given in Figure~\ref{fig:example-ilmc} where 
  states $1$,$2$, $5$, and~$6$ satisfy $\mathit{red}$,
  states $3$,$4$, $7$, and~$8$ satisfy $\mathit{green}$,
  states $9$, $10$, $13$, and~$14$ satisfy $\mathit{blue}$, and
  states $11$,$12$, $15$, and~$16$ satisfy $\mathit{orange}$. The
  closure operator is as usual for directed graphs.
  The hatched upper-right area contains the states satisfying
  $\lnearT \mathit{green}$, the hatched lower-left area the states
  satisfying $\lnearF \mathit{blue}$, i.e.,\ 
  \begin{displaymath}
    \begin{array}{rcl}
      \sem{\mkern2mu \lnearF \mathit{green}}
      & =
      & \SET{3,4,6,7,8,11}
      \smallskip \\
      \sem{\mkern2mu \lnearT \mathit{blue}}
      & =
      & \SET{5,9,10,13,14,15}
        \mkern1mu .
    \end{array}
  \end{displaymath}
\end{exa}

\begin{figure}
  \centering

\scriptsize

\scalebox{0.85}{%
\begin{tikzpicture}[baseline=1,
  ->, >=stealth', semithick, shorten >=0.5pt,
  every state/.style={
    circle, minimum size=15pt, inner sep=2pt, draw},
  ]
  
  \draw [gray, pattern={Lines[angle=45, distance=6pt]}, pattern color=gray!50]
  (3.4,6.1) -- (6.1,6.1) -- (6.1,3.4) -- (4.6,3.4) --
  (4.6,1.9) -- (3.4,1.9) -- (3.4,3.4) -- (1.9,3.4) -- (1.9,4.6) --
  (3.4,4.6) -- cycle ;
  \node [draw=none] at (7.2,5.5) {\normalsize $\sem{\mkern2mu \lnearF \mathit{green}}$} ;
  
  \draw [gray, pattern={Lines[angle=-45, distance=6pt]}, pattern
  color=gray!50] (0.4,4.6) -- (1.6,4.6) -- (1.6,3.1) -- (3.1,3.1) --
  (3.1,1.6) -- (4.6,1.6) -- (4.6,0.4) --
  (3.1, 0.4) -- (0.4,0.4) -- cycle ;
  \node [draw=none] at (-0.6,1) {\normalsize $\sem{\mkern2mu \lnearT \mathit{blue}}$} ;

  \node [state, fill=red!60] (x1) at (1,5.5) {$1$} ;
  \node [state, fill=red!60] (x2) at (2.5,5.5) {$2$} ;
  \node [state, fill=green!50] (x3) at (4,5.5) {$3$} ;
  \node [state, fill=green!50] (x4) at (5.5,5.5) {$4$} ;

  \node [state, fill=red!60] (x5) at (1,4) {$5$} ;
  \node [state, fill=red!60] (x6) at (2.5,4) {$6$} ;
  \node [state, fill=green!50] (x7) at (4,4) {$7$} ;
  \node [state, fill=green!50] (x8) at (5.5,4) {$8$} ;

  \node [state, fill=blue!60] (x9) at (1,(2.5) {$9$} ;
  \node [state, fill=blue!60] (x10) at (2.5,2.5) {$10$} ;
  \node [state, fill=orange!50] (x11) at (4,2.5) {$11$} ;
  \node [state, fill=orange!50] (x12) at (5.5,2.5) {$12$} ;

  \node [state, fill=blue!60] (x13) at (1,1) {$13$} ;
  \node [state, fill=blue!60] (x14) at (2.5,1) {$14$} ;
  \node [state, fill=orange!50] (x15) at (4,1) {$15$} ;
  \node [state, fill=orange!50] (x16) at (5.5,1) {$16$} ;

  \draw (x1) edge (x2) ;
  \draw (x2) edge (x3) ;
  \draw (x3) edge (x4) ;

  \draw (x6) edge (x5) ;
  \draw (x7) edge (x6) ;
  \draw (x8) edge (x7) ;

  \draw (x9) edge (x10) ;
  \draw (x10) edge (x11) ;
  \draw (x11) edge (x12) ;

  \draw (x14) edge (x13) ;
  \draw (x15) edge (x14) ;
  \draw (x16) edge (x15) ;

  \draw (x1) edge (x5) ;
  \draw (x5) edge (x9) ;
  \draw (x9) edge (x13) ;

  \draw (x14) edge (x10) ;
  \draw (x10) edge (x6) ;
  \draw (x6) edge (x2) ;

  \draw (x3) edge (x7) ;
  \draw (x7) edge (x11) ;
  \draw (x11) edge (x15) ;

  \draw (x16) edge (x12) ;
  \draw (x12) edge (x8) ;
  \draw (x8) edge (x4) ;

\end{tikzpicture}
} 
  \caption{Points satisfying $\lnearF \mathit{green}$ and
    $\lnearT \mathit {blue}$.}
  \label{fig:example-ilmc}
\end{figure}

\noindent
Equivalence for the logic~\imlc{} is defined as usual.

\begin{defi}[\imlc-equivalence]
  \label{def:ImlcEq}
  For a \qdcm~$\model$, the relation
  $\imlceq^{\model} \, \subseteq \, X \times X$ is defined by
  \begin{displaymath}
    \begin{array}{rcl}
      x_1 \imlceq^{\model} x_2
      & \text{iff}
      & \text{${\model, x_1 \models_{\imlc} \form} \Leftrightarrow {\model, x_2
        \models_{\imlc} \form}$, for all $\form \in \imlc$.}
    \end{array}\vspace{-0.25in}
  \end{displaymath}\closedefi
\end{defi}

\noindent
Next we formulate two lemmas which are used to prove that
\cmc-bisimilarity and \imlc-equivalence coincide.

\begin{lem}
  \label{lem:CMCbisIsImlceq}
  For all points $x_1, x_2$ in \qdcm~$\model$, if $x_1
  \cmcbis^{\model} x_2$ then $x_1 \imlceq^{\model} x_2$. 
\end{lem}

\begin{proof}
  The proof is similar to that of Lemma~\ref{lem:CMbisIsImleq}.  Let
  $\model = (X, \closureF, \peval)$. We verify, by induction on the
  structure of the formula~$\form$, that $x_1 \models \form$ if and
  only if $x_2 \models \form$ for $x_1, x_2 \in X$ such
  that~$x_1 \cmcbis x_2$. We only cover the case for~$\lnearT \form$.

  For the case of $\lnearT \form$ we will exploit
  condition~(\ref{CBCMC2}) of
  Definition~\ref{def:cmc-bisimulation-alt}.  Suppose
  $x_1 \cmcbis x_2$ and $x_1 \models \lnearT \form$. Then
  $x_1 \in \closureT (\sem{\form})$ by Definition~\ref{def:Imlc}. Thus
  exists $x'_1 \in X$ such that $x'_1 \models \form$ and
  $x_1 \in \closureT (x'_1)$.
  By Lemma~\ref{lem:FT}\ref{point:FT2}, we also have
  $x'_1 \in \closureF(x_1)$. Since $x'_1 \in \closureF(x_1)$ and
  $x_1 \cmcbis x_2$ we obtain, from condition~(\ref{CBCMC2}) of
  Definition~\ref{def:cmc-bisimulation-alt}, that
  $x'_2 \in \closureF(x_2)$ exists such that $x'_1 \cmcbis x'_2$. From
  $x'_1 \models \form$ we derive $x'_2 \models \form$ by induction
  hypothesis for~$\form$. Therefore, $x'_2 \in \sem{\form}$,
  $x_2\in \closureT(x'_2)$ by Lemma~\ref{lem:FT}\ref{point:FT2}, and
  thus $x_2 \in \closureT (\sem{\form})$, which implies
  $x_2 \models \lnearT \form$.
\end{proof}

For what concerns the other direction, i.e.,\ going from
\imlc-equivalence to \cmc-bisimilarity, we show, as in the previous
section, that logic equivalence is a bisimulation.

\begin{lem}
  \label{lem:ImlceqIsCMCbis}
  For a \qdcm~$\model$,  $\imlceq^{\model}$
  is a \cmc-bisimulation for~$\model$.
\end{lem}

\begin{proof}
  Let $\model = (X,\closureF,\peval)$. 
  Define, for  points $x,y \in X$, \imlc-formula~$\delta_{x,y}$
  as follows: if $x \imlceq y$, then set $\delta_{x,y}= \ltrue$, otherwise pick some 
  \imlc-formula~$\psi$ such that $\model,x \models \psi$ and $\model,y \models \lneg\psi$, and 
  set $\delta_{x,y}= \psi$.
  For a point~$x$
  of~$\model$  define
  $\chi(x) = \liand_{y \in X} \: \delta_{x,y}$. 
  Note, for $x,y\in X$, it holds that $y \in \sem{\chi(x)}$ if and only
  if $x \imlceq y$. 

  In order to verify that $\imlceq$ is a \cmc-bisimulation we check
  the conditions of
  Definition~\ref{def:cmc-bisimulation-alt}. Suppose
  $x_1 \imlceq x_2$ for $x_1, x_2 \in X$; (i)~clearly, $x_1 \in \peval(p)$
  iff $x_2 \in \peval(p)$ since $x_1 \models p$ iff $x_2 \models p$.
  (ii)~let $x'_1 \in \closureF(x_1)$. Since $x_1 \in \closureT(x'_1)$,
  by Lemma~\ref{lem:FT}\ref{point:FT2}, and
  $x'_1 \in \sem{\chi(x'_1)}$ it holds that
  $x_1 \models \lnearT \chi(x'_1)$. By assumption also
  $x_2 \models \lnearT \chi(x'_1)$. Hence, for some~$x'_2 \in X$ we
  have $x_2 \in \closureT(x'_2)$ and $x'_2 \in \sem{\chi(x'_1)}$. Thus
  $x'_2 \in \closureF(x_2)$ and $x'_1 \imlceq x'_2$.
  (iii)~similar to (ii).
\end{proof}

With the two lemmas above in place, we can establish the correspondence of
\cmc-bisimilarity and \imlc-equivalence.

\begin{thm}
  \label{thm:CMCbisEqImleq}
  For a \qdcm{} $\model$ it holds that $\imlceq^{\model}$ coincides
  with $\cmcbis^{\model}$.
\end{thm}

\begin{proof}
  On the one hand,
  $\mathord{\cmcbis^{\model}} \subseteq \mathord{\imlceq^{\mkern1mu
      \model}}$ by Lemma~\ref{lem:CMCbisIsImlceq}.  On the other hand,
  $\mathord{\imlceq^{\mkern1mu \model}} \subseteq
  \mathord{\cmcbis^{\model}}$ by Lemma~\ref{lem:ImlceqIsCMCbis}.
\end{proof}

\noindent
The following statement is an obvious consequence of Theorem~\ref{thm:CMCbisEqImleq}.

\begin{cor}
   \label{cor:CMCbisEquiv}
For all  \qdcm{s} $\model$, $\cmcbis^{\model}$ is an equivalence relation. 
\end{cor}

\begin{rem}
  \label{rem:rho}
  Previous work of Ciancia et al.\ concerns (various iterations of)
  the Spatial Logic for Closure Spaces (\slcs) that is based on the \emph{surrounded} operator~$\lsurr$ and the
  \emph{reachability} operator~$\lreach$, see
  e.g.,~\cite{Ci+21,Ci+20,Be+19,Ci+16}. A point~$x$ satisfies
  $\form_1 \, \lsurr\, \form_2$ if it lays in an area whose points
  satisfy~$\form_1$, and that is delimited, i.e., surrounded, by points that
  satisfy~$\form_2$. The point~$x$ satisfies
  $\lreach \,\form_1[\form_2]$ if there is a path starting in~$x$ that
  \emph{reaches} a point satisfying~$\form_1$ and whose intermediate
  points---if any---satisfy~$\form_2$.

  In~\cite{Be+19} it has been shown that the operator~$\lsurr$ can be
  derived from the logical operator~$\lreach$. More specifically,
  $\form_1 \, \lsurr\, \form_2$ is equivalent to
  $\form_1 \land \neg \lreach(\neg(\form_1 \lor \form_2))[\lneg
  \form_2]$. Furthermore, for \qdcm{s}, the operator~$\lreach$ gives
  rise to a pair of operators, namely $\ltothru$, coinciding
  with~$\lreach$, and its \lq{}converse\rq{}~$\lfromthru$, meaning
  that a point~$x$ \emph{can be reached from} a point
  satisfying~$\form_1$ via a path whose intermediate points---if
  any---satisfy~$\form_2$. It is easy to see that, for such spaces,
  $\lnearF\, \form$ and~$\lnearT\, \form$ are equivalent to
  $\lfromthru \,\form[\lfalse]$ and $\ltothru \, \form[\lfalse]$,
  respectively, and that $\ltothru \, \form_1[\form_2]$ is equivalent
  to a suitable (possibly) infinite disjunction of nested
formulas using only conjunction and the \imlc{} $\lnearT$
proximity operator, involving $\form_1$ and $\form_2$; 
similarly $\lfromthru \, \form_1[\form_2]$ is equivalent to a formula
using only conjunction and the \imlc{} $\lnearF$ operator.
 Thus, on \qdcm{s}, \imlc{} and \islcs{}---the
  infinitary version of \slcs{}~\cite{Ci+21}---share the same
  expressive power. The interested reader is referred to~\cite{Ci+22}.
As for Theorem~\ref{thm:CMbisEqImleq}, note that the
Hennessy-Milner result of Theorem~\ref{thm:CMCbisEqImleq} is the counterpart for closure
models of similar results for Kripke and neighbourhood models.
\end{rem}


\section{\cop-Bisimilarity for \qdcm{}}
\label{sec:COPAbisimilarity}

\cm-bisimilarity, and its refinement \cmc-bisimilarity, are a
fundamental starting point for the study of spatial bisimulations due
to their strong link with topo-bisimulation. However, \cm{} and
\cmc-bisimilarity are rather strong, i.e.,\ fine-grained, notions of
equivalence regarding their use for reasoning about general properties
of space.
For instance, in the \qdcm{} given in Figure~\ref{fig:grid02}
(with points at the border satisfying the atomic proposition
$\mathit{green}$ and inner points satisfying atomic proposition
$\mathit{red}$), point~$13$ in the centre is {\em not} \cmc-bisimilar to
any other red point around it.
This is because \cmc-bisimilarity is based on reachability ``in one
step'', so to speak, as can be seen from
Definition~\ref{def:cmc-bisimulation-alt}. This, in turn, provides
\cmc-bisimilarity with the ability to distinguish points based on
their distance to a set that can be characterized by an \imlc-formula.
This is inconsistent with the wish to consider in this model, for
instance, all red points to be spatially equivalent. Along the same
lines, one could desire to identify all `$\mathit{green}$' points at
the border as well as all inner `$\mathit{red}$' points, obtaining a
model of two points only, a `$\mathit{green}$' one and a
`$\mathit{red}$'.

\begin{figure}
  \centering
  \scalebox{0.85}{%

\scriptsize

\begin{tikzpicture}[baseline=1,
  <->, >=stealth', semithick, shorten >=0.5pt, shorten <=0.5pt,
  every state/.style={
    circle, minimum size=12pt, inner sep=0.5pt, draw},
  ]
  
  \node [state, fill=green!50] (x11) at (1.0,1.0) {21};
  \node [state, fill=green!50] (x12) at (1.0,2.5) {16};
  \node [state, fill=green!50] (x13) at (1.0,4.0) {11};
  \node [state, fill=green!50] (x14) at (1.0,5.5) {6};
  \node [state, fill=green!50] (x15) at (1.0,7.0) {1};

  \node [state, fill=green!50] (x21) at (2.5,1.0) {22};
  \node [state, fill=red!50  ] (x22) at (2.5,2.5) {17};
  \node [state, fill=red!50  ] (x23) at (2.5,4.0) {12};
  \node [state, fill=red!50  ] (x24) at (2.5,5.5) {7};
  \node [state, fill=green!50] (x25) at (2.5,7.0) {2};

  \node [state, fill=green!50] (x31) at (4.0,1.0) {23};
  \node [state, fill=red!50  ] (x32) at (4.0,2.5) {18};
  \node [state, fill=red!50  ] (x33) at (4.0,4.0) {13};
  \node [state, fill=red!50  ] (x34) at (4.0,5.5) {8};
  \node [state, fill=green!50] (x35) at (4.0,7.0) {3};

  \node [state, fill=green!50] (x41) at (5.5,1.0) {24};
  \node [state, fill=red!50  ] (x42) at (5.5,2.5) {19};
  \node [state, fill=red!50  ] (x43) at (5.5,4.0) {14};
  \node [state, fill=red!50  ] (x44) at (5.5,5.5) {9};
  \node [state, fill=green!50] (x45) at (5.5,7.0) {4};

  \node [state, fill=green!50] (x51) at (7.0,1.0) {25};
  \node [state, fill=green!50] (x52) at (7.0,2.5) {20};
  \node [state, fill=green!50] (x53) at (7.0,4.0) {15};
  \node [state, fill=green!50] (x54) at (7.0,5.5) {10};
  \node [state, fill=green!50] (x55) at (7.0,7.0) {5};

  \draw (x11) edge (x12) ;  \draw (x12) edge (x13) ;
  \draw (x13) edge (x14) ;  \draw (x14) edge (x15) ;

  \draw (x21) edge (x22) ;  \draw (x22) edge (x23) ;
  \draw (x23) edge (x24) ;  \draw (x24) edge (x25) ;

  \draw (x31) edge (x32) ;  \draw (x32) edge (x33) ;
  \draw (x33) edge (x34) ;  \draw (x34) edge (x35) ;

  \draw (x41) edge (x42) ;  \draw (x42) edge (x43) ;
  \draw (x43) edge (x44) ;  \draw (x44) edge (x45) ;

  \draw (x51) edge (x52) ;  \draw (x52) edge (x53) ;
  \draw (x53) edge (x54) ;  \draw (x54) edge (x55) ;

  \draw (x14) edge (x25) ;

  \draw (x13) edge (x24) ;
  \draw (x24) edge (x35) ;

  \draw (x12) edge (x23) ;
  \draw (x23) edge (x34) ;
  \draw (x34) edge (x45) ;

  \draw (x11) edge (x22) ;
  \draw (x22) edge (x33) ;
  \draw (x33) edge (x44) ;
  \draw (x44) edge (x55) ;

  \draw (x21) edge (x32) ;
  \draw (x32) edge (x43) ;
  \draw (x43) edge (x54) ;

  \draw (x31) edge (x42) ;
  \draw (x42) edge (x53) ;

  \draw (x41) edge (x52) ;

  \draw (x11) edge (x21) ; \draw (x21) edge (x31) ;
  \draw (x31) edge (x41) ; \draw (x41) edge (x51) ;

  \draw (x12) edge (x22) ; \draw (x22) edge (x32) ;
  \draw (x32) edge (x42) ; \draw (x42) edge (x52) ;

  \draw (x13) edge (x23) ; \draw (x23) edge (x33) ;
  \draw (x33) edge (x43) ; \draw (x43) edge (x53) ;

  \draw (x14) edge (x24) ; \draw (x24) edge (x34) ;
  \draw (x34) edge (x44) ; \draw (x44) edge (x54) ;

  \draw (x15) edge (x25) ; \draw (x25) edge (x35) ;
  \draw (x35) edge (x45) ; \draw (x45) edge (x55) ;

  \draw (x12) edge (x21) ;

  \draw (x13) edge (x22) ;
  \draw (x22) edge (x31) ;

  \draw (x14) edge (x23) ;
  \draw (x23) edge (x32) ;
  \draw (x32) edge (x41) ;

  \draw (x15) edge (x24) ;
  \draw (x24) edge (x33) ;
  \draw (x33) edge (x42) ;
  \draw (x42) edge (x51) ;

  \draw (x25) edge (x34) ;
  \draw (x34) edge (x43) ;
  \draw (x43) edge (x52) ;

  \draw (x35) edge (x44) ;
  \draw (x44) edge (x53) ;

  \draw (x45) edge (x54) ;
\end{tikzpicture}
  } 
  \caption{Red midpoint~$13$ not \cmc-bisimilar to
    red points around it.
  }
  \label{fig:grid02}
\end{figure}

In order to overcome the `counting' capability of \cmc-bisimilarity,
one may think of considering \emph{paths} instead of single
`steps'. In fact in~\cite{Ci+21} we introduced such a bisimilarity,
called path bisimilarity. This bisimilarity requires that, in order
for two points to be bisimilar, for every path starting from one point,
there must be a path starting from the other point having
bisimilar end-points. However, as we  discussed in
Section~\ref{sec:Introduction}, path bisimilarity is too weak. In
particular because no constraints are put on intermediate points.

In order to come to a more constraining definition, and yet a 
notion of bisimulation weaker than \cmc-bisimulation, a deeper insight into the structure of a path
is desirable as well as some, relatively high-level, requirements
regarding paths. For this purpose we resort to a form of compatibility
between paths that essentially requires each of them to be composed of
non-empty, adjacent and interrelated `zones'.

Informally, two such paths under consideration should share the same
structure as in Figure~\ref{fig:Zones}. We see that (i)~both paths can
be seen as comprised of corresponding zones---the number of such zones in each path being equal and positive, (ii)~points in
corresponding zones are bisimilar, although (iii)~the length of
corresponding zones---and thus of the two paths---may be different.

We first formalise the notion of compatibility of paths in a
quasi-discrete closure model.

\begin{defi}[Path-compatibility]
  \label{def:path-compatibility}
  Given \qdcm{} $\model=(X,\closureF, \peval)$ and
  $B \subseteq {X \times X}$ a relation.
    Two forward (respectively backward) paths $\pi_1 = ( x'_i )_{i{=}0}^\ell$ and
    $\pi_2 = ( x''_j )_{j{=}0}^m$ in~$\model$ are called {\em compatible}
    with respect to~$B$ in~$\model$ if, for some~$N > 0$, two total
    monotone surjections $f : [0;\ell \mkern1mu] \to [1;N]$ and
    $g : [0;m] \to [1;N]$ exist such that~$(x'_i,x''_j)\in B$ for all
    indices $i \in [0;\ell \mkern1mu ]$ and~$j \in [0;m]$ satisfying
    $f(i) = g(j)$.
\end{defi}

\noindent
The functions $f$ and~$g$ are referred to as \emph{matching
  functions}. Note that both the number~$N$ and functions $f$ and~$g$
need not be unique. The minimal number~$N > 0$ for which matching
functions exist is defined to be the number of {\em zones} of the two paths
$\pi_1$ and~$\pi_2$.

It is easy to see that, given \qdcm{} $\model=(X,\closureF, \peval)$ and relation
$B \subseteq X\times X$,
whenever two paths 
$( x'_i )_{i{=}0}^\ell$ and $( x''_j )_{j{=}0}^m$ over $X$
are compatible with respect to $B$, for any pair of matching function $f$ and~$g$ the
following holds, by virtue of monotonicity and surjectivity:
$f(0) = 1$ and~$g(0) = 1$ and $f(\ell) = g(m)$.  Hence $(x'_0,x''_0)\in B$
and $(x'_\ell,x''_m)\in B$. In other words, for paths that are compatible with
respect to a relation, the start and end points are
related by that relation.

\begin{figure}
  \centering
\scriptsize
\begin{tikzpicture}[baseline=1,
  <->, >=stealth', semithick, shorten >=0.5pt, shorten <=0.5pt,
  every state/.style={
    circle, minimum size=12pt, inner sep=0.5pt, draw},
  scale=0.67,
  ]
  
  \node [state, fill=green!50] (x11) at (1.0,1.0) {21};
  \node [state, fill=green!50] (x12) at (1.0,2.5) {16};
  \node [state, fill=green!50] (x13) at (1.0,4.0) {11};
  \node [state, fill=green!50] (x14) at (1.0,5.5) {6};
  \node [state, fill=green!50] (x15) at (1.0,7.0) {1};

  \node [state, fill=green!50] (x21) at (2.5,1.0) {22};
  \node [state, fill=red!50  ] (x22) at (2.5,2.5) {17};
  \node [state, fill=red!50  ] (x23) at (2.5,4.0) {12};
  \node [state, fill=red!50  ] (x24) at (2.5,5.5) {7};
  \node [state, fill=green!50] (x25) at (2.5,7.0) {2};

  \node [state, fill=green!50] (x31) at (4.0,1.0) {23};
  \node [state, fill=red!50  ] (x32) at (4.0,2.5) {18};
  \node [state, fill=red!50  ] (x33) at (4.0,4.0) {13};
  \node [state, fill=red!50  ] (x34) at (4.0,5.5) {8};
  \node [state, fill=green!50] (x35) at (4.0,7.0) {3};

  \node [state, fill=green!50] (x41) at (5.5,1.0) {24};
  \node [state, fill=red!50  ] (x42) at (5.5,2.5) {19};
  \node [state, fill=red!50  ] (x43) at (5.5,4.0) {14};
  \node [state, fill=red!50  ] (x44) at (5.5,5.5) {9};
  \node [state, fill=green!50] (x45) at (5.5,7.0) {4};

  \node [state, fill=green!50] (x51) at (7.0,1.0) {25};
  \node [state, fill=green!50] (x52) at (7.0,2.5) {20};
  \node [state, fill=green!50] (x53) at (7.0,4.0) {15};
  \node [state, fill=green!50] (x54) at (7.0,5.5) {10};
  \node [state, fill=green!50] (x55) at (7.0,7.0) {5};

  \node [draw=none] (pi1) at (0.125,7.0) {\normalsize $\pi_1$} ;
  \draw [->] (x11) edge (x12) ;  \draw [->] (x12) edge (x13) ;
  \draw [->] (x13) edge (x14) ; 

  \draw (x22) [->] edge (x23) ;

  \draw [<-] (x42) edge (x43) ;
  \draw [<-] (x43) edge (x44) ; 

  \draw [<-] (x51) edge (x52) ;  \draw [<-] (x52) edge (x53) ;
  \draw [<-] (x53) edge (x54) ;  \draw [<-] (x54) edge (x55) ;

  \draw [<-] (x11) edge (x21) ; \draw [<-] (x21) edge (x31) ;
  \draw [<-] (x31) edge (x41) ; \draw [<-] (x41) edge (x51) ;

  \draw [<-] (x22) edge (x32) ;
  \draw [<-] (x32) edge (x42) ;

  \draw [->] (x23) edge (x33) ;
  \draw [->] (x33) edge [loop right, looseness=12] (x33) ;

  \draw [->] (x14) edge (x24) ; \draw [->] (x24) edge (x34) ;
  \draw [->] (x34) edge (x44) ;

  \draw [->] (x15) edge (x25) ; \draw [->] (x25) edge (x35) ;
  \draw [->] (x35) edge (x45) ; \draw [->] (x45) edge (x55) ;

\end{tikzpicture}
  \hspace*{1.5cm}
\scriptsize
\begin{tikzpicture}[baseline=1,
  <->, >=stealth', semithick, shorten >=0.5pt, shorten <=0.5pt,
  every state/.style={
    circle, minimum size=12pt, inner sep=0.5pt, draw},
  scale=0.67,
  ]
  
  \node [state, fill=green!50] (x11) at (1.0,1.0) {21};
  \node [state, fill=green!50] (x12) at (1.0,2.5) {16};
  \node [state, fill=green!50] (x13) at (1.0,4.0) {11};
  \node [state, fill=green!50] (x14) at (1.0,5.5) {6};
  \node [state, fill=green!50] (x15) at (1.0,7.0) {1};

  \node [state, fill=green!50] (x21) at (2.5,1.0) {22};
  \node [state, fill=red!50  ] (x22) at (2.5,2.5) {17};
  \node [state, fill=red!50  ] (x23) at (2.5,4.0) {12};
  \node [state, fill=red!50  ] (x24) at (2.5,5.5) {7};
  \node [state, fill=green!50] (x25) at (2.5,7.0) {2};

  \node [state, fill=green!50] (x31) at (4.0,1.0) {23};
  \node [state, fill=red!50  ] (x32) at (4.0,2.5) {18};
  \node [state, fill=red!50  ] (x33) at (4.0,4.0) {13};
  \node [state, fill=red!50  ] (x34) at (4.0,5.5) {8};
  \node [state, fill=green!50] (x35) at (4.0,7.0) {3};

  \node [state, fill=green!50] (x41) at (5.5,1.0) {24};
  \node [state, fill=red!50  ] (x42) at (5.5,2.5) {19};
  \node [state, fill=red!50  ] (x43) at (5.5,4.0) {14};
  \node [state, fill=red!50  ] (x44) at (5.5,5.5) {9};
  \node [state, fill=green!50] (x45) at (5.5,7.0) {4};

  \node [state, fill=green!50] (x51) at (7.0,1.0) {25};
  \node [state, fill=green!50] (x52) at (7.0,2.5) {20};
  \node [state, fill=green!50] (x53) at (7.0,4.0) {15};
  \node [state, fill=green!50] (x54) at (7.0,5.5) {10};
  \node [state, fill=green!50] (x55) at (7.0,7.0) {5};

  \node [draw=none] (pi2) at (7.875,4.0) {\normalsize $\pi_2$} ;

  \draw [->] (x53) edge (x43) ;
  \draw [<-] (x32) edge (x43) ;
  \draw [<-] (x43) edge (x34) ;
  \draw [<-] (x34) edge (x23) ;
  \draw [<-] (x23) edge (x32) ;

 \end{tikzpicture}
  \caption{F-compatible path for points $1$ and $15$}
  \label{fig:compatible-paths}
\end{figure}

\begin{exa}
  With reference to the \qdcm~$\model$ given in
  Figure~\ref{fig:grid02}, let us consider forward path $\pi_1$ of length~$27$
  starting in point~$1$ and forward path $\pi_2$ of length~$10$ starting in
  point~$15$, as indicated in Figure~\ref{fig:compatible-paths}.  In
  more detail,
  \begin{displaymath}
    \begin{array}{l}
      \pi_1 = (1,2,3,4,5, 10,15,20,25, 24,23,22,21,
      \\ \qquad \qquad \qquad
       16,11,6, 7,8,9, 14,19, 18,17, 12, 13, 13, 13, 13 )
      \smallskip \\
      \pi_2 = (15, 14,18,12,8, 14,18,12,8, 14,18 )
    \end{array}
  \end{displaymath}
  The paths are compatible.
  Possible matching functions are $f : [0;27] \to [1;2]$ and $g : [0;10]
  \to [1;2]$ where $f(i) = 1$ for $0 \leqslant i \leqslant 15$,
  $f(i) = 2$ for $16 \leqslant i \leqslant 27$ and $g(0) = 1$ and
  $g(j) = 2$ for $1 \leqslant j \leqslant 10$. Thus `green' states are
  in zone~$1$, `red' states are in zone~$2$.
\end{exa}

\noindent
The above notion of compatibility give rise to what is called
{\em compatible path bisimilarity} or \cop-bisimilarity, as defined below.

\begin{defi}[\cop-bisimilarity]
  \label{def:CoPabisimilarity}
  Let \qdcm{} $\model=(X,\closureF, \peval)$ be given.  A symmetric
  relation~$B \subseteq {X \times X}$ is a \emph{\cop-bisimulation}
  for~$\model$ if, whenever $(x_1,x_2)\in B$ for points $x_1, x_2 \in X$,
  the following holds:
  \begin{enumerate}
  \item $x_1 \in \peval(p)$ iff $x_2 \in \peval(p)$ for all
    $p \in \ap$.
  \item \label{zones} For each forward path$~\pi_1$ from~$x_1$ there exists
    a compatible forward path~$\pi_2$ from~$x_2$ with respect to~$B$
    in~$\model$.
  \item For each backward path~$\pi_1$ from~$x_1$ there exists a compatible
    backward path~$\pi_2$ from~$x_2$ with respect to~$B$ in~$\model$.
  \end{enumerate}
\item Two points $x_1, x_2 \in X$ are called \emph{\cop-bisimilar}
  in~$\model$ if $(x_1,x_2)\in B$ for some \cop-bisimulation $B$ for
  $\model$. Notation, $x_1\,\copbis^{\model}\, x_2$. \closedefi
\end{defi}

\begin{exa}\label{exa:CoPaBis}
  To see that, e.g., the points $7$ and~$13$ in the closure
  model~$\model$ of Figure~\ref{fig:grid02} are \cop-bisimilar, we
  consider the relation
  $B = {\bigl( \peval(\mathit{green}) \times \peval(\mathit{green})
    \bigr)} \cup {\bigl( \peval(\mathit{red}) \times
    \peval(\mathit{red}) \bigr)} \subseteq {X \times X}$, among others
  relating $7$ and~$13$. We see that in~$\model$, all pairs of `green'
  points and all pairs of `red' points are related by~$B$.\\
  (i)~Suppose $(x_1,x_2)\in B$ for two points $x_1$ and~$x_2$. Then
  $x_1 \in \peval(p)$ iff $x_2 \in \peval(p)$ by definition
  of~$\peval$ and~$B$.\\
  (ii)~Suppose $(x_1,x_2)\in B$ for two points $x_1$ and~$x_2$. We only
  treat the case where $x_1$~is the point~$13$ in the middle and
  $x_2$~is not; the other cases are similar and/or simpler. Let
  $\pi_1 = ( x'_i )_{i{=}0}^\ell$ be a forward path in~$\model$
  starting from~$x_1$. The path~$\pi_1$ shows an alternation of red
  and green points, not necessarily strict, starting from red. We
  construct a path $\pi_2$ with the same alternation of colour.
  Choose $x_r \in \closureF(x_2) \setminus \SET{13}$ such that
  $\peval(x_r) = \mathit{red}$ and $x_g \in \closureF(x_r)$ such that
  $\peval(x_g) = \mathit{green}$.
  We define the forward path $\pi_2 = ( x''_j )_{j{=}0}^\ell$
  in~$\model$, of the same length as~$\pi_1$, as follows:
  $x''_0 = x_2$ and for $j \in (0;\ell]$, $x''_j = x_g$ if
  $x'_j \in \peval(\mathit{green})$ and $x''_j = x_r$ if
  $x'_j \in \peval(\mathit{red})$. Note that $\pi_2$~is a forward path
  from~$x_2$ in~$\model$ indeed, by choice of the points $x_r$
  and~$x_g$.
  To check the compatibility of $\pi_1$ and~$\pi_2$ with respect to
  the relation~$B$ we choose as matching functions
  $f,g : [0;\ell \mkern1mu ] \to [1;\ell+1]$ satisfying $f(i) = i{+}1$
  and $g(j) = j{+}1$ for $0 \leqslant i,j \leqslant \ell$. Clearly, if
  $f(i) = g(j)$ then $i=j$, from which $\peval(x'_i) = \peval(x''_j)$
  and $(x'_i,x''_j)\in B$~follow.\\  
  (iii)~Compatibility of backward paths for $x_1$ and~$x_2$ can be
  checked in a similar way.
\end{exa}

In~\cite{Ci+22a} a definition of \cop-bisimilarity has been presented that is different from
Definition~\ref{def:CoPabisimilarity}. In the Appendix, we recall the definition given in~\cite{Ci+22a} and we show that  it is equivalent to Definition~\ref{def:CoPabisimilarity}. We prefer the latter since it better conveys the intuitive notion of the ``zones'' in the relevant paths.

The next lemma captures the fact that \cmc-bisimilarity is included in
\cop-bisimilarity.

\begin{lem}
  \label{lem:CMCbisImplCoPa}
  ${\cmcbis^{\model}} \subseteq {\copbis^{\model}}$ for every
  \qdcm~$\model$.
\end{lem}

\begin{proof}
  Suppose $\model = (X, \closureF, \peval)$ is a \qdcm.  Let $x_1$
  and~$x_2$ be two points of~$\model$ such that $(x_1,x_2)\in B$ for some
  \cmc-bisimulation $B \subseteq X \times X$. We check the conditions
  for the symmetric relation~$B$ to be a \cop-bisimulation.

  (i)~Clearly, since $B$ is a \cmc-bisimulation, $x_1 \in \peval(p)$
  iff $x_2 \in \peval(p)$ for all~$p \in \ap$.

  (ii)~Let $\pi_1 = ( x'_i )_{i{=}0}^\ell$ be a forward path 
  from~$x_1$. Define a path $\pi_2 = ( x''_j )_{j{=}0}^\ell$
  in~$\model$ from~$x_2$ as follows: $x''_0 = x_2$ and
  $x''_{j{+}1} \in X$ is such that $x''_{j{+}1} \in \closureF(x''_j)$
  and~$(x'_{j{+}1},x''_{j{+}1})\in B$ for $0 < j \leqslant \ell$. This is
  possible because $B$~is a \cmc-bisimulation. Then $\pi_2$~is a
  forward path from~$x_2$. Moreover, the pair of functions
  $f,g : [0;\ell \mkern1mu] \to [1;\ell {+} 1]$ with $f(i) = i{+}1$
  and $g(j) = j{+}1$ for $0 \leqslant i,j \leqslant \ell$, show that
  $\pi_1$ and~$\pi_2$ are compatible.

  (iii)~Similar to case (ii).
\end{proof}

\begin{figure}
  \centering

\begin{tikzpicture}[baseline=1,
  ->, >=stealth', semithick, shorten >=0.5pt,
  every state/.style={
    circle, minimum size=8pt, inner sep=2pt, draw},
  ]
  
  \node [state, fill=red!50, label={+270:{$x$}}] (x1) at (1,1) {} ;
  \node [state, fill=red!50, label={+270:{$y$}}] (y1) at (3,1) {} ;
  \node [state, fill=green!50, label={+270:{$z$}}] (z1) at (5,1) {} ;
  \draw (x1) edge (y1) ; \draw (y1) edge (z1) ;
\end{tikzpicture}
\qquad \qquad
\begin{tikzpicture}[baseline=1,
  ->, >=stealth', semithick, shorten >=0.5pt,
  every state/.style={
    circle, minimum size=8pt, inner sep=2pt, draw},
  ]
  
  \node [state, fill=red!50, label={+270:{$u$}}] (u1) at (1,1) {} ;
  \node [state, fill=green!50, label={+270:{$v$}}] (v1) at (3,1) {} ;
  \draw (u1) edge (v1) ;

\end{tikzpicture}
  \caption{$x \copbis u$ and $x \not\cmcbis u$.}
  \label{fig:example-copa-bisim}
\end{figure}

\noindent
The converse of Lemma~\ref{lem:CMCbisImplCoPa} does not hold. With
reference to the quasi-discrete closure model in
Figure~\ref{fig:example-copa-bisim} for which
$\peval(\mathit{red}) = \SET{x,y,u}$ and
$\peval(\mathit{green}) = \SET{z,v}$, it is easy to see that the
symmetric closure of the relation $B = \SET{(x,u), (y,u), (z,v)}$ is a
\cop-bisimulation, and so $x \copbis u$.  However, there is no
relation $B$ such that $x \cmcbis u$, because there is a green point
in $\closureF(u)$ whereas no green point belongs to the
$\closureF(x)$.

More formally, in order to check that  for every forward
path~$\pi_1 = ( x'_i )_{i{=}0}^\ell$ from~$u$ a compatible path
$\pi_2 = ( x''_j )_{j{=}0}^m$ from~$x$ exists, we reason as
follows, considering two cases:

Case~1: If $\range(\pi_1) = \lbrace u \rbrace$ we choose~$\pi_2$ such
that $\range(\pi_2) = \lbrace x \rbrace$. The constant mappings $f :
[0;\ell \mkern1mu] \to \SET{1}$ and $g : [0;m] \to \SET{1}$ are
matching functions for $\pi_1$ and~$\pi_2$. 

Case~2: Otherwise, for some $k > 0$, $x'_i = v$ for all~$i \geqslant k$. Put
$m = \ell{+}1$, $x''_0 = x$, $x''_1 = y$, and~$x''_j = z$
for~$j \geqslant 2$. Now, $f : [0;\ell \mkern1mu] \to [1;2]$ such that
$f(i) = 1$ for $0 \leqslant i < k$, $f(i) = 2$ for
$k \leqslant i \leqslant \ell$ and $g : [0;\ell{+}1] \to [1;2]$ such
that $g(j) = 1$ for $0 \leqslant j \leqslant k$, $g(j) = 2$ for
$k < j \leqslant \ell{+}1$ are matching functions.

Although $x \copbis u$, we have that $x \not\cmcbis u$.  In fact, we
have $u \models \lnearT \mathit{green}$ whereas
$x \not\models \lnearT \mathit{green}$ and therefore, by
Theorem~\ref{thm:CMCbisEqImleq}, the point~$u$ is not CMC-bisimilar to
the point~$x$.

One can argue that \cop-bisimilarity for \qdcs{s} is conceptually the same as
divergence-blind stuttering equivalence for Kripke structures, \dbs-equivalence for short. 
We recall the definition of \dbs-equivalence from~\cite{Gr+17} below:

\begin{defi}[Definition 2.2 of \cite{Gr+17}]
Let $K = (S, \ap, \rightarrow, L)$ be a Kripke model.
A symmetric relation $R \subseteq S \times S$ is a \dbs-equivalence if and only if, for all $s, t \in $ such that $R(s,t)$:
\begin{enumerate}
\item $L(s)= L(t)$, and
\item for all $s' \in S$, if $s \rightarrow s'$, then there are 
$t_0, \ldots,t_k \in S$ for some $k \in \nats$ such that
$t_0 = t$, $R(s,t_i)$, $t_i \rightarrow t_{i+1}$ for all $i<k$, and $R(s',t_k)$.
\end{enumerate}
We say that two states $s,t \in S$ are \dbs-equivalent if and only if 
there is a \dbs-equivalence relation $R$ such that $R(s,t)$.\closedefi
\end{defi}

Taking the specific structure of closure spaces into account, as well as the back-and-forth aspect of such spaces, the standard
definition of \dbs-equivalence 
can be rearranged for \qdcm{s} as follows.

\begin{defi}[\dbs-bisimilarity for \qdcm{s}]
  \label{def:dbs-bisimilarity}
  Let $\model = (X, \closureF, \peval)$ be a \qdcm. A symmetric
  relation $B \subseteq {X \times X}$ is called a \dbs-bisimulation
  for \qdcm~$\model$ if whenever $(x_1,x_2)\in B$ for points
  $x_1, x_2 \in X$, the following holds:
  \begin{enumerate} [(1)]
  \item $x_1 \in \peval(p)$ if and only if $x_2 \in \peval(p)$ for all
    $p \in \ap$.
  \item For all $x'_1 \in \closureF(x_1)$ exist $n \geqslant 0$ and
    $y_0, \ldots, y_n \in X$ such that $y_0 = x_2$,
    $y_{i{+}1} \in \closureF(y_i)$ and~$(x_1,y_i)\in B$ for $0 \leqslant i
    < n$, and~$(x'_1,y_n)\in B$.
   \item 
   For all $x'_1 \in \closureT(x_1)$ exist $n \geqslant 0$ and
    $y_0, \ldots, y_n \in X$ such that $y_0 = x_2$,
    $y_{i{+}1} \in \closureT(y_i)$ and~$(x_1,y_i)\in B$ for $0 \leqslant i
    < n$, and~$(x'_1,y_n)\in B$.
  \end{enumerate}
  Two points $x_1, x_2 \in X$ are called \dbs-bisimilar in~$\model$,
  if $(x_1,x_2)\in B$ for some \dbs-bisimulation~$B$ for~$M$. Notation,
  $x_1 \dbseq^\model x_2$.\closedefi
\end{defi}

\noindent
It holds that \cop-bisimilarity and \dbs-bisimilarity coincide. We
split the proof of this into two lemmas.

\begin{lem}
  \label{lem:copa-implies-dbs}
  For each \qdcm~$\model$ it holds that ${\copbis^\model} \subseteq
  {\dbsbis^\model}$. 
\end{lem}

\begin{proof}
  A \cop-bisimulation $B \subseteq {X \times X}$ is a
  \dbs-bisimulation as well: 
  
  (i)~The first requirement of
  Definition~\ref{def:dbs-bisimilarity} is immediate. 
  
  (ii)~In order to verify the second requirement of
  Definition~\ref{def:dbs-bisimilarity}, suppose $(x_1,x_2)\in B$ and
  $x'_1 \in \closureF(x_1)$. We note that $\pi_1 = (x_1,x'_1)$ is a
  forward path from~$x_1$ in~$\model$. Since~$(x_1,x_2)\in B$ and $B$~is a
  \cop-bisimulation, a forward path $\pi_2 = (x''_j)_{j{=}0}^m$
  from~$x_2$ exists  that is compatible to~$\pi_1$ with respect to~$B$. So,
  $x''_{j{+}1} \in \closureF(x''_j)$ for $0 \leqslant j < m$. Let
  $N > 0$ and $f : [0;1] \to [1;N]$, $g : [0;m] \to [1;N]$ be matching
  functions for $\pi_1$ and~$\pi_2$. 

Since $f$~is surjective, then either $N=1$ or $N=2$.
In the first case,  $(x_1,x''_j)\in B$ and $(x'_1,x''_j)\in B$ for
$0 \leqslant j \leqslant m$, which suffices.
If, instead, $N=2$, then for a suitable index~$k$, $0 < k \leqslant m$, $g(j) = 1$ for
$0 \leqslant j < k$ and $g(j) = 2$ for $k \leqslant j \leqslant m$, as shown in Figure~\ref{fig:lem:copa-implies-dbs-Fig1}.

\begin{figure}
     \begin{tikzpicture}
       \tikzset{->}
       \tikzstyle{point}=[circle,draw=black,fill=white,inner sep=0pt,minimum width=4pt,minimum height=4pt]
          \node (p0)[label=left:{$x_2$}] at (3,0){}; 
          \node (p1)[point,label=below:{$x_0''$},fill=black] at (3.05,0){}; 
          \node (p2)[point,label=below:{$x_1''$},fill=black, right of=p1] {}; 
          \node (p3)[point,fill=black, right of=p2] {}; 
          \node (p4)[point,label=below:{$x_{k-1}''$},fill=black, right of=p3] {}; 
          \node (p5)[point,label=below:{$x_k''$},fill=black, right of=p4] {}; 
          \node (p6)[point,fill=black, right of=p5] {}; 
          \node (p7)[point,fill=black, right of=p6] {}; 
          \node (p8)[point,label=below:{$x_m''$},fill=black, right of=p7] {}; 
          \draw [thick] (p1) edge (p2);
          \draw [thick] (p2) edge (p3);
          \draw [dotted] (p3) edge (p4);
          \draw [thick] (p4) edge (p5);
          \draw [thick] (p5) edge (p6);
          \draw [dotted] (p6) edge (p7);
          \draw [thick] (p7) edge (p8);
          \node (p9)[above of=p1]{$0$};
          \node (p10)[right of=p9]{$1$};
          \node (p11)[right of=p10]{};
          \node (p12)[right of=p11]{$k-1$};
          \node (p13)[right of=p12]{$k$};
          \node (p14)[right of=p13]{};
          \node (p15)[right of=p14]{};
          \node (p16)[right of=p15]{$m$};
          \draw (p9) edge[|->,thin,right] node{$\pi_2$} (p1);
          \draw (p10) edge[|->,thin,right] node{$\pi_2$} (p2);
          \draw (p12) edge[|->,thin,right] node{$\pi_2$} (p4);  
          \draw (p13) edge[|->,thin,right] node{$\pi_2$} (p5); 
          \draw (p16) edge[|->,thin,right] node{$\pi_2$} (p8);   
          \node (p17)[] at (2,2.5){$1$};
          \node (p18)[right of=p17]{};
          \node (p19)[right of=p18]{$2$};   
          \draw (p9) edge[|->,thin,right] node{$g$} (p17);  
          \draw (p10) edge[|->,thin,right] node{$g$} (p17);
          \draw (p12) edge[|->,thin,right] node{$g$} (p17);
          \draw (p13) edge[|->,thin,right] node{$g$} (p19);
          \draw (p16) edge[|->,thin,right] node{$g$} (p19);  
          \node (p20)[label=left:{$x_1$}] at (0,5){}; 
          \node (p21)[point,label=above:{$x_0'$},fill=black] at (0.05,5){}; 
          \node (p22)[point,label=above:{$x_1'$},fill=black, right of=p21] {};     
          \draw [thick] (p21) edge (p22); 
          \node (p23)[below of=p21]{$0$};
          \node (p24)[right of=p23]{$1$};
          \draw (p23) edge[|->,thin,right] node{$\pi_1$} (p21);
          \draw (p24) edge[|->,thin,right] node{$\pi_1$} (p22);         
          \draw (p23) edge[|->,thin,right] node{$f$} (p17);  
          \draw (p24) edge[|->,thin,right] node{$f$} (p19);           
     \end{tikzpicture}
\caption{\label{fig:lem:copa-implies-dbs-Fig1} Example illustrating the proof of Lemma~\ref{lem:copa-implies-dbs} for
$N=2$; $f,g$ are the matching functions and $\pi_2$ an F-compatible path for $\pi_1$}
\end{figure}
 
  Thus, writing $x'_0$ for~$x_1$, it holds that $(x'_0,x''_j)\in B$
  for $0 \leqslant j < k$ and $(x'_1,x''_j)\in B$ for
  $k \leqslant j \leqslant m$. For $x''_0, \ldots , x''_k \in X$ we
  thus have $x''_0 = x_2$, $x''_{j{+}1} \in \closureF(x''_j)$ and
  $(x_1,x''_j)\in B$ for $0 \leqslant j < k$ and $(x'_1,x''_k)\in B$ as
  required (see Figure~\ref{fig:lem:copa-implies-dbs-Fig2}).

\begin{figure} 
     \begin{tikzpicture}
       \tikzstyle{point}=[circle,draw=black,fill=white,inner sep=0pt,minimum width=4pt,minimum height=4pt]
          \node (p0)[label=left:{$x_2$}] at (3,0){}; 
          \node (p1)[point,label=below:{$x_0''$},fill=black] at (3.05,0){}; 
          \node (p2)[point,label=below:{$x_1''$},fill=black, right of=p1] {}; 
          \node (p3)[point,fill=black, right of=p2] {}; 
          \node (p4)[point,label=below:{$x_{k-1}''$},fill=black, right of=p3] {}; 
          \node (p5)[point,label=below:{$x_k''$},fill=black, right of=p4] {}; 
          \node (p6)[point,fill=black, right of=p5] {}; 
          \node (p7)[point,fill=black, right of=p6] {}; 
          \node (p8)[point,label=below:{$x_m''$},fill=black, right of=p7] {}; 
          \draw [->,thick] (p1) edge (p2);
          \draw [->,thick] (p2) edge (p3);
          \draw [->,dotted] (p3) edge (p4);
          \draw [->,thick] (p4) edge (p5);
          \draw [->,thick] (p5) edge (p6);
          \draw [->,dotted] (p6) edge (p7);
          \draw [->,thick] (p7) edge (p8);
          \node (p9)[above of=p1]{$0$};
          \node (p10)[right of=p9]{$1$};
          \node (p11)[right of=p10]{};
          \node (p12)[right of=p11]{$k-1$};
          \node (p13)[right of=p12]{$k$};
          \node (p14)[right of=p13]{};
          \node (p15)[right of=p14]{};
          \node (p16)[right of=p15]{$m$};
          \draw (p9) edge[|->,thin,right] node{$\pi_2$} (p1);
          \draw (p10) edge[|->,thin,right] node{$\pi_2$} (p2);
          \draw (p12) edge[|->,thin,right] node{$\pi_2$} (p4);  
          \draw (p13) edge[|->,thin,right] node{$\pi_2$} (p5); 
          \draw (p16) edge[|->,thin,right] node{$\pi_2$} (p8);   
          \node (p17)[] at (2,2.5){$1$};
          \node (p18)[right of=p17]{};
          \node (p19)[right of=p18]{$2$};   
          \draw (p9) edge[|->,thin,right] node{$g$} (p17);  
          \draw (p10) edge[|->,thin,right] node{$g$} (p17);
          \draw (p12) edge[|->,thin,right] node{$g$} (p17);
          \draw (p13) edge[|->,thin,right] node{$g$} (p19);
          \draw (p16) edge[|->,thin,right] node{$g$} (p19);  
          \node (p20)[label=left:{$x_1$}] at (0,5){}; 
          \node (p21)[point,label=above:{$x_0'$},fill=black] at (0.05,5){}; 
          \node (p22)[point,label=above:{$x_1'$},fill=black, right of=p21] {};     
          \draw [->,thick] (p21) edge (p22); 
          \node (p23)[below of=p21]{$0$};
          \node (p24)[right of=p23]{$1$};
          \draw (p23) edge[|->,thin,right] node{$\pi_1$} (p21);
          \draw (p24) edge[|->,thin,right] node{$\pi_1$} (p22);         
          \draw (p23) edge[|->,thin,right] node{$f$} (p17);  
          \draw (p24) edge[|->,thin,right] node{$f$} (p19);           
          \draw[line width=0.1mm,red] (p1) -- (p21);
          \draw[line width=0.1mm,red] (p2) -- (p21);
          \draw[line width=0.1mm,red] (p3) -- (p21);
          \draw[line width=0.1mm,red] (p4) -- (p21);
          \draw[line width=0.1mm,red] (p5) -- (p22);
          \draw[line width=0.1mm,red] (p6) -- (p22);
          \draw[line width=0.1mm,red] (p7) -- (p22);
          \draw[line width=0.1mm,red] (p8) -- (p22);
     \end{tikzpicture}
\caption{\label{fig:lem:copa-implies-dbs-Fig2} The same as Figure~\ref{fig:lem:copa-implies-dbs-Fig1} but showing also relation $B$, in red.}
\end{figure}
  (iii)~Similar to case (ii).
\end{proof}

\noindent
The converse of Lemma~\ref{lem:copa-implies-dbs} is captured by the
next result.

\begin{lem}
  \label{lem:dbs-implies-copa}
  For each \qdcm~$\model$ it holds that
  ${\dbsbis^\model} \subseteq {\copbis^\model}$.
\end{lem}

\begin{proof}
  Let $\model = (X, \closureF, \peval)$ be a \qdcm. We show that a
  \dbs-bisimulation $B \subseteq {X \times X}$ for~$\model$ is a
  \cop-bisimulation for~$\model$ as well. We verify the three
  requirements of Definition~\ref{def:CoPabisimilarity}.

  (i)~Clearly, if $(x_1,x_2)\in B$ for two points~$x_1, x_2 \in X$ we have
  $x_1 \in \peval(p)$ iff $x_2 \in \peval(p)$ for all $p \in \ap$ by
  definition of a \dbs-bisimulation.

  (ii)~To see that the second requirement for a \cop-bisimulation is
  met, we prove the following claim by induction on~$\ell_1$:
  (\textit{Claim}) For all~$\ell_1 \geqslant 0$, if~$(x_1,x_2)\in B$ for
  two points~$x_1, x_2 \in X$ and $\pi_1$~is a forward path of
  length~$\ell_1$ from~$x_1$ in~$\model$, then  a forward
  path~$\pi_2$ from~$x_2$ in~$\model$ exists that is compatible
  with~$\pi_1$ with respect to~$B$.

  Base, $\ell_1 = 0$: The path~$\pi_1$ consists of $x_1$ only. The
  path~$\pi_2$ consisting of $x_2$ only, is a forward path from~$x_2$
  that is compatible to~$\pi_1$ with respect to~$B$. 
  
  Induction step,~$\ell_1 > 0$: Put $x'_1 = \pi_1(1)$. We have
  $x'_1 \in \closureF(x_1)$ since $\pi_1$~is a forward path. Because
  $B$ is a \dbs-bisimulation and $(x_1,x_2)\in B$, exist $n \geqslant 0$
  and $y_0, \ldots, y_n \in X$ such that $y_0 = x_2$,
  $y_{i{+}1} \in \closureF(y_i)$ and~$(x_1,y_i)\in B$ for
  $0 \leqslant i < n$, and~$(x'_1,y_n)\in B$.
  We split the path~$\pi_1$ into the subpaths $\pi'_{1}$ and
  $\pi''_{1}$ where $\pi'_{1}$ is the restriction of~$\pi_1$
  to~$[0;1]$, i.e.,\ $\pi'_{1}=\pi_{1}|[0;1]$, and $\pi''_{1}$ is the
  $1$-shift of $\pi_1$, i.e.,\ $\pi_{1b}=\pi_{1} {\uparrow} 1$. Then
  $\pi'_2 = ( y_i )_{i{=}0}^n$ is a forward path in~$\model$ from~$x_2$
  that is compatible with respect to~$B$ with the path
  $\pi'_{1} = (x_1,x'_1)$ from~$x_1$. The subpath~$\pi''_{1}$
  in~$\model$ is of length~$\ell_1{-}1$ and is a forward path from $x'_1$. 
  Since~$(x'_1,y_n)\in B$, by induction hypothesis a
  path~$\pi''_2$ from $y_n$ exists that is compatible with
  $\pi''_{1}$ with respect to~$B$. Let the path
  $\pi_2 = \pi'_2 \cdot \pi''_2$ be the concatenation of the forward
  paths $\pi'_2$ and~$\pi''_2$. Then $\pi_2$ is a forward path
  in~$\model$ from~$x_2$ that is compatible with
  $\pi'_{1} \cdot \pi''_{1} = \pi_1$ with respect to~$B$.
  
  (iii)~Similar to case (ii).
\end{proof}

\blankline

\noindent
In order to provide a logical characterisation of \cop-bisimilarity,
we replace the proximity modalities $\lnearF$ and~$\lnearT$ 
in~\imlc{} by the (forward and backward) {conditional reachability} modalities
$\lstothru$ and~$\lsfromthru$ to obtain the Infinitary Compatible
Reachability Logic, or \icrl{} for short.

\begin{defi}[\icrl{}]
  \label{def:Icrl}  
  The abstract language of \icrl{} is defined as follows:
    \begin{displaymath}
      \form ::= p \mid \lneg \, \form \mid \liand_{i \in I} \form_i \mid
      \lstothru \, \form_1[\form_2] \mid \lsfromthru \,\form_1[\form_2]
    \end{displaymath}
    where $p$~ranges over~$\ap$ and $I$~ranges over a collection of index sets.
  The satisfaction relation with respect to a \qdcm~$\model$,
    point~$x \in \model$, and \icrl{} formula~$\form$ is defined
    recursively on the structure of $\form$ as follows:  
    \begin{displaymath}
      \def\arraystretch{1.2}
      \begin{array}{r c l c l c l l}
        \model,x & \models_{\icrl} & p & \Leftrightarrow & x  \in \peval(p) \\
        \model,x & \models_{\icrl} & \lneg \,\form & \Leftrightarrow & \model,x  \models_{\icrl} \form \mbox{ does not hold} \\
        \model,x & \models_{\icrl} & \liand_{i\in I} \form_i & \Leftrightarrow & \model,x  \models_{\icrl} \form_i \mbox{ for all $i \in I$} \\
\model,x  & \models_{\icrl} & \lstothru \, \form_1 [\form_2] & \Leftrightarrow &
\mbox{a forward path } (x_i)_{i=0}^{\ell} \mbox{ from } x \mbox{ exists }\mbox{ such that }\\
&&&&\model, x_{\ell} \models_{\icrl} \form_1, \mbox{ and } 
\model, x_j \models_{\icrl} \form_2 \mbox{ for } j\in [0,\ell)\\
\model,x  & \models_{\icrl} & \lsfromthru \, \form_1 [\form_2] & \Leftrightarrow &
\mbox{a backward path } (x_i)_{i=0}^{\ell} \mbox{ from } x \mbox{ exists }\mbox{ such that }\\
&&&&\model, x_{\ell} \models_{\icrl} \form_1, \mbox{ and } 
\model, x_j \models_{\icrl} \form_2 \mbox{ for } j\in [0,\ell)\\
      \end{array}\vspace{-0.2in}
      \def\arraystretch{1.2}
    \end{displaymath}
    \closedefi
\end{defi}

\noindent
Informally, $x$~satisfies $\lstothru \form_1 [\form_2]$ if it
satisfies~$\form_1$, or it satisfies~$\form_2$ and there is a point
satisfying $\form_1$ \emph{that can be reached from~$x$} via a path
whose intermediate points, if any, satisfy~$\form_2$.  Similarly,
$x$~satisfies $\lsfromthru \form_1 [\form_2]$ if it
satisfies~$\form_1$, or it satisfies~$\form_2$ and there is a point
satisfying~$\form_1$ \emph{from which $x$ can be reached} via a path
whose intermediate points, if any, satisfy~$\form_2$.

\begin{figure}
  \centering

\begin{tikzpicture}[baseline=1,
  ->, >=stealth', semithick, shorten >=0.5pt,
  every state/.style={
    circle, minimum size=8pt, inner sep=2pt, draw},
  ]
  
  \node [state, fill=red!50, label={+270:{$x_1$}}] (x1) at (1,1) {} ;
  \node [state, fill=red!50, label={+270:{$x_2$}}] (x2) at (2.5,1) {} ;
  \node [state, fill=green!50, label={+90:{$x_3$}}] (x3) at (4,2) {} ;
  \node [state, fill=blue!50, label={+90:{$x_4$}}] (x4) at (5.5,2) {} ;
  \node [state, fill=blue!50, label={+90:{$x_5$}}] (x5) at (7,2) {} ;
  \node [state, fill=red!50, label={+90:{$x_6$}}] (x6) at (8.5,2) {} ;
  \node [state, fill=red!50, label={+270:{$x_7$}}] (x7) at (4,0) {} ;
  \node [state, fill=blue!50, label={+270:{$x_8$}}] (x8) at (5.5,0) {} ;
  \node [state, fill=blue!50, label={+270:{$x_9$}}] (x9) at (7,0) {} ;
  \draw (x1) edge (x2) ;
  \draw (x2) edge (x3) (x3) edge (x4) (x4) edge (x5) (x5) edge (x6) ;
  \draw (x2) edge (x7) (x7) edge (x8) (x8) edge (x9) ;
\end{tikzpicture}
  \caption{}
  \label{fig:example-icrl-satisfaction}
\end{figure}

\begin{exa}
  Consider the \qdcm~$\model$ depicted in
  Figure~\ref{fig:example-icrl-satisfaction}. With respect to~$\model$
  we have, for example, that 
$\lstothru \mathit{red} \, [\mathit{red}]$ is satisfied by  $x_1, x_2, x_6$, and $x_7$, 
$\lstothru \mathit{red} \, [\mathit{blue}]$   is satisfied by $x_1, x_2, x_4, x_5, x_6$, and $x_7$,
$\lstothru ( \lstothru \mathit{red} \, [\mathit{blue}] )\, [\neg \,\mathit{blue}]$ is satisfied by
$x_1, x_2, x_3, x_4, x_5, x_6$, and $ x_7$,
$\lneg \bigl( \lstothru ( \lstothru \mathit{red} \, [\mathit{blue}] ) [\neg \mathit{blue}] \bigr)$ 
is satisfied by $x_8$, and $x_9$, and
$\lsfromthru \mathit{red} \, [\mathit{blue}] $ is satisfied by
 $x_1, x_2, x_6, x_7, x_8$, and $x_9$. 
\end{exa}

\noindent
Also for \icrl{} we introduce a notion of equivalence.

\begin{defi}[\icrl-equivalence]
  \label{def:IcrlEq}
  Given a \qdcm{} $\model = (X,\closureF,\peval)$, the 
  relation $\icrleq^{\model} \, \subseteq \, X \times X$ is defined by
  \begin{displaymath}
    x_1 \icrleq^{\model} x_2
    \ \text{iff} \ 
    \model, x_1 \models_{\icrl} \form
    \Leftrightarrow
    \model, x_2 \models_{\icrl} \form,
    \text{for all $\form \in \icrl$}. \tag*{$\bullet$}
  \end{displaymath}
\end{defi}

\noindent
We first check that \cop-bisimulation respects \icrl-equivalence.

\begin{lem}\label{lem:CoPabisIsIcrleq}
For $x_1,x_2$ in  \qdcm{} $\model$,if
$x_1 \copbis^{\model} x_2$ then $x_1 \icrleq^{\model} x_2$.\qed
\end{lem}

\begin{proof}
  Let $\model = (X,\closureF,\peval)$ be a \qdcm. We verify that for
  all $x_1, x_2 \in X$ such that $x_1 \copbis x_2$ and
  $x_1 \models \form$ implies $x_2 \models \form$, by induction on the
  structure of~$\form$ in~\icrl. We only cover the case
  $\lstothru \form_1 [\form_2]$. The proof for the case $\lsfromthru \form_1 [\form_2]$ is similar, while that for the other cases is straightforward.

  Let $x_1$ and~$x_2$ be two points of~$\model$ such that
  $x_1 \copbis x_2$. Suppose
  $x_1 \models \lstothru \form_1 [\form_2]$. Let
  $\pi_1=(x'_i)_{i=0}^{k_1}$ be a forward path from~$x_1$ satisfying
  $\pi_1(k_1) \models \form_1$ and $\pi_1(i) \models \form_2$ for
  $0 \leqslant i < k_1$. Since $x_1 \copbis x_2$, a forward path
  $\pi_2=(x''_i)_{i=0}^{k_2}$ from~$x_2$ exists that is compatible
  with~$\pi_1$ with respect to~$\copbis$. Let, for
  appropriate~$N > 0$, $f : [0;k_1] \to [1;N]$ and
  $g : [0;k_2] \to [1;N]$ be matching functions for $\pi_1$
  and~$\pi_2$. Without loss of generality,
  $g^{-1}(\SET{N}) = \SET{k_2}$. Since $f(k_1) = g(k_2) = N$, we have
  $\pi_1(k_1) \copbis \pi_2(k_2)$. Thus $\pi_2(k_2) \models \form_1$
  by induction hypothesis. Moreover, if $0 \leqslant j < k_2$, then
  $g(j) < N$ by assumption and $f(i) = g(j)$ and
  $\pi_1(i) \copbis \pi_2(j)$ for some~$i$, $0 \leqslant i <
  k_1$. Since $\pi_1(i) \models \form_2$ it follows that
  $\pi_2(j) \models \form_2$ by induction hypothesis.  Therefore
  path~$\pi_2$ witnesses $x_2 \models \lstothru \form_1 [\form_2]$.
\end{proof}

\noindent
The converse of Lemma~\ref{lem:CoPabisIsIcrleq} stating that
\icrl-equivalent points are \cop-bisimilar, is given below.

\begin{lem}
  \label{lem:IcrleqIsCoPabis}
  For a \qdcm{} $\model$, $\icrleq$ is a \cop-bisimulation for
  $\model$. \qed
\end{lem}

\begin{proof}
  Let~$\model$ be a quasi-discrete closure model. We check that
  $\icrleq$ satisfies requirement~(ii) of
  Definition~\ref{def:CoPabisimilarity}. Requirement~(i) is immediate;
  requirement~(iii) can be verified, with appropriate auxiliary
  definitions, similar to requirement~(ii).

  Let, for points $x,y \in X$,  \icrl-formula~$\delta_{x,y}$ be
  defined as follows: if $x \icrleq y$, then set $\delta_{x,y}= \ltrue$, otherwise pick some 
  \icrl-formula~$\psi$ such that $\model,x \models \psi$ and $\model,y \models \lneg\psi$, and 
  set $\delta_{x,y}= \psi$.
  Put $\chi(x) = \bigwedge_{y \in X} \: \delta_{x,y}$. As before, for
  points $x, y \in X$, it holds that $x \models \chi(y)$ iff
  $x \icrleq y$.

  Let the function $\zones : \fpthsF \to \mathbb{N}$ 
  be such that, for 
  a forward path $\pi = (x_i )_{i{=}0}^n$,
  \begin{displaymath}
    \begin{array}{r@{\,}c@{\,}l@{\quad}l}
      \zones(\pi) & = & 1 & \text{if $n=0$} \\
      \zones(\pi) & = & \zones(\pi') &
      \text{if $n > 0$ and $x_0 \icrleq x_1$} \\
      \zones(\pi) & = & \zones(\pi') + 1 & 
      \text{if $n > 0$ and $x_0 \not\icrleq x_1$} \\
    \end{array}
  \end{displaymath}
  where $\pi'= ( x_i )_{i{=}1}^n$. A forward path~$\pi$ is said to have
  $n$~zones, if $\zones(\pi) = n$.

  \textit{Claim} For all~$k \geqslant 1$, for all $x_1, x_2 \in X$, if
  $x_1 \icrleq x_2$ and $\pi_1$ is a forward path from~$x_1$ of
  $k$~zones, then  a forward path~$\pi_2$ from~$x_2$ exists that is
  compatible to~$\pi_1$ with respect to~$\icrleq$. The claim is
  proven by induction on~$k$.

  Base case, $k = 1$: If $x_1 \icrleq x_2$ and
  $\pi_1 = ( x'_i )_{i{=}0}^n$~is a forward path from~$x_1$ of 1~zone,
  then $x_1 \icrleq x'_i$ for $0 \leqslant i \leqslant n$. Let~$\pi_2$
  be the forward path consisting of~$x_2$ only. Since
  $x_1 \icrleq x_2$, also $x_2 \icrleq x'_i$ for
  $0 \leqslant i \leqslant n$. Hence, $\pi_2$~is compatible
  with~$\pi_1$ with respect to~$\icrleq$.

  Induction step, $k {+} 1$: Suppose $x_1 \icrleq x_2$ and
  $\pi_1 = ( x'_i )_{i{=}0}^n$~is a forward path from~$x_1$ of
  $k{+}1$~zones. So, let $m > 0$ be such that $x_1 \icrleq x'_i$ for
  $0 \leqslant i < m$ and $x_1 \not\icrleq x'_m$. Then it holds that
  $x_1 \models \lstothru {\chi(x'_m)} {[ \mkern1mu \chi(x_1)]}$. Since
  $x_2 \icrleq x_1$ also
  $x_2 \models \lstothru {\chi(x'_m)} {[ \mkern1mu \chi(x_1)]}$. Thus,
  a forward path~$\pi'$ from~$x_2$ exists such that
  $\pi'(\ell) \models \chi(x'_m)$ and $\pi'(j) \models \chi(x_1)$ for
  $0 \leqslant j < \ell$. We have that
  $x'_m \icrleq \pi'(\ell)$---because
  $\pi'(\ell) \models \chi(x'_m)$---and the forward path $\pi_1[m;n]$
  from~$x'_m$ has $k$~zones. By induction hypothesis exists a forward
  path~$\pi''$ from~$\pi(\ell)$ that is compatible
  to~$\pi'_1[m;n]$. Then the path~$\pi' \cdot \pi''$, the
  concatenation of $\pi'$ and~$\pi''$, is a forward path from~$x_2$ that is compatible with the
  path~$\pi_1$.

  Requirement~(ii) of Definition~\ref{def:CoPabisimilarity} follows
  directly from the claim. We conclude that $\icrleq$ is a
  \cop-bisimulation, as was to be shown. 
\end{proof}

\noindent
The correspondence between \icrl-equivalence and \cop-bisimilarity is
summarised by the following theorem.

\begin{thm}
  \label{thm:CoPabisEqIcrleq}
  For every \qdcm~$\model$ it holds that \icrl-equivalence
  $\icrleq^{\model}$ coincides with \cop-bisimilarity $\copbis^{\model}$.
\end{thm}

\begin{proof}
  Direct from Lemma~\ref{lem:CMCbisImplCoPa} and
  Lemma~\ref{lem:IcrleqIsCoPabis}.
\end{proof}


\section{Conclusions}
\label{sec:conclusions}

In this paper we have developed a uniform, coherent and self-contained formal framework for the design of logic-based, automatic verification techniques, notably model checking, to be applied on suitable representations of space. In particular we have studied three bisimilarities for closure
spaces, namely (i) {\em Closure Model bisimilarity} (\cm-bisimilarity), (ii) its specialisation for \qdcm{s}, namely
{\em \cm-bisimilarity with converse} (\cmc-bisimilarity), and (iii) {\em Compatible Paths bisimilarity} (\cop-bisimilarity), 
a conditional form of path bisimilarity.
For each bisimilarity we introduced a spatial logic and we proved that the associated logical equivalence coincides with the relevant bisimilarity.

\cm-bisimilarity is a generalisation for \cm{s} of classical
topo-bisimilarity for topological spaces and can also be interpreted as an instance of neighbourhood bisimilarity, when \cm{s} are seen as neighbourhood models.  \cmc-bisimilarity 
 takes into consideration the fact that, in \qdcm{s}, there is
a notion of ``direction'' given by the binary relation underlying the
closure operator. This can be exploited in order to obtain an
equivalence---namely \cmc-bisimilarity---that, for
\qdcm{s}, refines \cm-bisimilarity.  In several applications, e.g., image analysis, weaker notions, such as spatial conditional reachability, are more appropriate. 
This notion cannot be captured conveniently using \cm-bisimilarity or \cmc-bisimilarity since both relations are too strong, in the sense that they can ``count'' the number of single closure steps. 
To capture such weaker notion we introduce \cop-bisimilarity, 
a stronger version of path bisimilarity --- first proposed
in~\cite{Ci+21} --- but weaker than \cmc-bisimilarity. \cop-bisimilarity expresses 
path ``compatibility'' in a way that resembles the concept of stuttering
equivalence for Kripke models~\cite{BCG88}. We have proven 
that \cop-bisimilarity coincides with (an adaptation to closure models of) divergence-blind
stuttering equivalence~\cite{Gr+17}.  

The practical relevance of the study presented in this paper stems from the fact that
bisimilarity gives rise to the notion of a minimal model. Such minimal models, when computable,  
can be effectively used for optimising spatial model-checking procedures, when the bisimilarity preserves
logical equivalence. This is indeed the case for all spatial equivalences proposed in this paper. 
The issue of model minimisation via bisimilarity has been explicitly addressed in~\cite{Ci+20} and in~\cite{Ci+23}.

In~\cite{Ci+20} a minimisation algorithm has been presented for \cmc-bisimilarity, based on the coalgebraic interpretation of the equivalence. That paper also introduces \minilogica, a minimisation tool for \cmc-bisimilarity. 

In~\cite{Ci+23} an encoding of (finite) \qdcm{s} into  labelled transition systems has been defined. 
Intuitively, given finite model $\model=(X,\closureF,\peval)$, a labelled transition system
$\lts(\model)$ is generated such that
whenever $y\in \closureF(x)$ and $\invpeval(\SET{x})=\invpeval(\SET{y})$,
there is  a $\tau$-transition in $\lts(\model)$
from the state encoding $x$ to that encoding $y$, 
whereas, if $\invpeval(\SET{x})\not=\invpeval(\SET{y})$, there is a transition
with a different, ad hoc, label encoding the change in $\invpeval$. A similar procedure is performed for  encoding $\closureT$. In addition, a state $s$ encoding a point $x$ is 
 equipped with a self-loop labelled by $p$, for each $p \in \invpeval(\SET{x})$.
It has been shown that two points
are \cop-bisimilar in $\model$ if and only if the states they are encoded into are branching equivalent in $\lts(\model)$. The size of the resulting labelled transition system grows linearly with
that of the original model. In particular, depending on the nature of the relation underlying 
the closure space, the number of states of the transition system can vary from the
same as to twice that of the points in the space.
This  enables the exploitation of very efficient minimisation algorithms for 
branching  equivalence, as those in~\cite{Gr+17}, and related tools, such as {\tt mCRL2},  in order to compute the 
minimal \qdcm{} with respect to \cop-bisimilarity
and then apply spatial model-checking to the latter.
As a proof-of-concept, an experimental tool chain, based on the above mentioned encoding and {\tt mCRL2} has been used for obtaining the minimal model for a maze similar to that of
Figure~\ref{fig:Maze}, for various sizes ranging from  $128 \times 128=16$K pixels to $8192 \times 8192=64$M pixels.  The labelled
transition systems resulting from the encoding have the same number of states as the number of pixels in the images, whereas the reduced models have, in all cases, $7$ states, like the 
model in Figure~\ref{fig:PathReducedMaze}, with different kinds of blue and white states, 
identifying different situations with respect to the possibility of reaching an exit. 
The resulting model has been model-checked using the (non-optimised) prototypical tool \graphlogica{.}
The time for model-checking the original models with \voxlogica{} ranged from $0.49$s to $5.34$s.
The time for model-checking the minimal model with \graphlogica{} was about $0.40$s leading to a maximal speed-up of $12.77$, for the largest model. The time for minimisation ranges from $0.0$s to $21.53$s. 
We underline here that \voxlogica{} is inherently much faster than \graphlogica{} as the former is specialised for images, exploiting state-of-the-art imaging libraries and automatic parallelisation. This poses a further challenge to the speed-up via minimisation and is the reason why we used~\voxlogica{} instead of \graphlogica{} for the full model.
It should also be noted that our interest  is more on space efficiency than on computation time, since  the size of space models is often huge, beyond the capacity of current analysis tools. Moreover, we focus more on the time taken for actual model-checking than on that for model reduction because, typically, one is interested in several model-checking sessions on the same model, so that the (time) cost of generating the minimal model is then distributed over such sessions.

Another experiment has been performed, using 
images of the classical \emph{Philips 5544} monoscope test pattern. Also in this case, the experiment has been carried out for different resolutions of the image, ranging from $30.37$K pixels to $31.64$M pixels. Given the specific nature of this image (e.g. the presence of very thin test lines), the topological relationships of different parts of the image change when changing the resolution. Consequently, in this case, for different resolutions we get different minimal models, with number of states ranging from $155$ to $955$. In this case, 
model-checking has been
performed for some formulas with a high level of nesting (up to 16) of reachability operators. The
model-checking time for the largest model
was $14.96$s whereas that for the minimised one was $0.65$s with
a speed-up of $22.87$. The minimisation of the largest monoscope model took $9.88$s. For further details, the reader is referred to~\cite{Ci+23}.

As a closing remark concerning the experiments briefly described above, we underline once more the importance of the development of a uniform theoretical framework for models of space, that includes model reduction based on suitable bisimulation-based equivalences  that enjoy the Hennessy-Milner property with respect to specific logics. As an example, the approach described in~\cite{Ci+23} and outlined above makes it possible, for users interested in modelling 
space and analysing properties of such models, to think and reason {\em directly} in terms of {\em only} spatial notions. A user interested in certain spatial features of a model of a space does not need to have specific competence in labelled transition systems, silent $\tau$-transitions or branching equivalence. Neither she is required to think about coding her problem into a concurrency theory problem. She only needs knowledge on models for space and related theory. All the coding necessary for exploiting branching equivalence --- or, in general any other notion developed in concurrency theory or general modal logics and related models --- is taken care of by software tools based on solid mathematical notions and correctness proofs. In other words, the approach discussed in this paper, although based on notions that are closely related to temporal logic and bisimilarity in concurrency theory --- that, in turn, have their roots in modal logic and related models --- provides a conceptual ``front end'' for the user that requires only knowledge on models of space and related logics whereas, whenever necessary or convenient, an encoding to similar models and logics for different domains, typically those of concurrency theory or general modal logics, is provided by a ``back end'' like that presented, for instance, in~\cite{Ci+23}. Such a ``front end'' and ``back end'' relationship can then be inherited as a basis for the design and implementation of efficient automatic spatial analysis tools, as illustrated in~\cite{Ci+23}.

Many results we have shown in this paper concern \qdcm{s;} we think
the investigation of their extension to continuous or general closure
spaces is an interesting line of future research.  A significant first step
in that direction is proposed in~\cite{Bez+22} where a technique  for 
 model-checking a variant of \slcs{} on polyhedra in Euclidean spaces has been proposed
as well as an efficient  model checker for 2D and 3D spaces, that enhances the spatial model-checker  \voxlogica{.} 
Notions of bisimilarity for polyhedral models have also been proposed that 
coincide with the equivalence induced by variants of \slcs{} and for which discrete versions have been presented that can be used for computing the minimal model of a cell poset model, i.e. of the discrete representation of a polyhedral model used for model-checking~\cite{Ci+23a,Be+24a,Be+24b}.

In~\cite{Ci+20} we
investigated a coalgebraic view of \qdcm{s} that was useful for the
definition of the minimisation algorithm for \cmc-bisimilarity. It
would be interesting to study a similar approach for \pth-bisimilarity
and \cop-bisimilarity. Steps in this direction have been performed in the context of closure hyperdoctrines in recent work by Miculan et al.~\cite{CaM21}.

In~\cite{HKP09}, coalgebraic bisimilarity is developed for a general kind of models, generalising the topological ones, known as Neighbourhood Frames. The development of the 
corresponding Hennessy-Milner style theorems for logics with reachability such as \slcs{,} by enriching the class of Neighbourhood Frames with a notion of path, could be an interesting addition to that theory.

\bibliographystyle{alphaurl}
\bibliography{clmv}

\appendix
\section{}

Below, we recall the definition of \cop-bisimilarity we used in ~\cite{Ci+22a}\footnote{Here, for notational consistency with the rest of the paper,  we use a terminology that is slightly different than, but equivalent to, that used in~\cite{Ci+22a}.} and show that it is equivalent to Definition~\ref{def:CoPabisimilarity}.

\begin{defi}[Definition 15 of \cite{Ci+22a}]\label{def:CoPabisimilarityOld}
Given \qdcs{} $\model=(X,\closureF, \peval)$, a symmetric relation 
$B \subseteq X \times X$ is a {\em \cop-bisimulation  for $\model$} if, whenever $B(x_1,x_2)$, the following holds:
\begin{enumerate}
\item for all $p\in\ap$ we have $x_1 \in \peval(p)$ in and only if $x_2 \in \peval(p)$;
\item\label{Ci+22a-zones} 
for each forward path $\pi_1=(x_i')_{i=0}^{\ell_1}$ from $x_1$ such that
$B(\pi_1(i),x_2)$ for all $i \in  [0;\ell_1)$ 
there is a forward path $\pi_2=(x_j'')_{j=0}^{\ell_2}$ from $x_2$ such that
the following holds:
$B(x_1,\pi_2(j))$ for all $j\in [0;\ell_2)$ and
$B(\pi_1(\ell_1),\pi_2(\ell_2))$;
\item 
for each backward path $\pi_1=(x_i')_{i=0}^{\ell_1}$ from $x_1$ such that
$B(\pi_1(i),x_2)$ for all $i\in [0;\ell_1)$ 
there is a backward path $\pi_2=(x_j'')_{j=0}^{\ell_2}$ from $x_2$ such that
the following holds:
$B(x_1,\pi_2(j))$ for all $j\in [0;\ell_2)$ and
$B(\pi_1(\ell_1),\pi_2(\ell_2))$.
\end{enumerate}
Two points $x_1, x_2 \in X$ are called {\em \cop-bisimilar} in $\model$ if $B(x_1, x_2)$
for some \cop-bisimulation $B$ for $\model$. Notation, $x_1\,\copbis^{\model}\, x_2$.
\closedefi
\end{defi}

In~\cite{Ci+22a} it has been shown that, for all \qdcm{s} $\model$, $\copbis^{\model}$, defined according to Definition~\ref{def:CoPabisimilarityOld}, is
an equivalence relation.  
The following proposition establishes the equivalence of the two definitions.

\begin{prop}
  \label{prop:CoPA-rec-vs-zones}
  Let $\model = (X, \closureF, \peval)$ be a \qdcm{} and $x_1,x_2 \in X$. It holds that 
  $x_1 \copbis x_2$  according to 
  Definition~\ref{def:CoPabisimilarity} if and only if
  $x_1 \copbis x_2$ according to 
  Definition~\ref{def:CoPabisimilarityOld}.
\end{prop}

\begin{proof}
In this proof, for the sake of clarity, we use the following notation:
$\copbisold$ denotes relation $\copbis$ defined according to 
  Definition~\ref{def:CoPabisimilarityOld}, 
  and
 $\copbisnew$ denotes relation $\copbis$ defined according to 
  Definition~\ref{def:CoPabisimilarity}.
  
  We first prove that $\copbisnew \subseteq \copbisold$. Let $x_1,x_2 \in X$ and assume
  $x_1 \copbisnew x_2$. We have to show that the three conditions of Definition~\ref{def:CoPabisimilarityOld} 
  are satisfied.\\
 (i) The first condition is trivially satisfied since $x_1 \copbisnew x_2$ requires that 
 $x_1 \in \peval(p)$ iff $x_2 \in \peval(p)$ for all $p \in \ap$.\\
(ii) As for the second condition, let $\pi_1=(x'_i)_{i=0}^{\ell_1}$  be a forward path from $x_1$ such that 
$x'_i \, \copbisnew\, x_2$ for all $i \in [0;\ell_1)$. We distinguish two cases:\\
{\bf Case 1:} $x'_{\ell_1} \, \copbisnew x_2$.\\
Put $\ell_2 = 0$ and consider the one-point path $\pi_2=(x_2)$. Then we have that
$\pi_1(\ell_1) = x'_{\ell_1} \, \copbisnew  x_2 = \pi_2(\ell_2)$. Note that the other requirement of 
Definition~\ref{def:CoPabisimilarityOld}(\ref{Ci+22a-zones}) is vacuous.\\
{\bf Case 2:} $x'_{\ell_1} \, \not\copbisnew x_2$.\\
Since, by assumption, $x_1 \, \copbisnew x_2$, a forward path $\pi_2=(x''_j)_{j=0}^{\ell_2}$ from $x_2$ exists, for some $\ell_2 \in \nats$, that is compatible with $\pi_1$.
Define $m_2 = \min\ZET{j\in (0;\ell_2]}{x''_{j-1}\not\copbisnew x''_j}$
(see Figure~\ref{fig:lem:CoPA-rec-vs-zonesNIO}). 
\begin{figure}
\begin{tikzpicture}[%
  pstate/.style = {draw, circle, minimum size=1.5mm, inner sep=0pt},
  qstate/.style = {draw, circle, minimum size=1mm, inner sep=0pt, fill=black},
  semithick,
  ]


  \node (x1) at (0.5,3) {$x_1$} ;
  \node [pstate, label={+90:$x'_0$}] (x0p) at (1,3) {} ;
  \node [qstate, label={+90:$x'_{\ell_1{-}1}$}] (xellminus1p) at (4,3) {} ;
  \node [pstate, label={+90:$x'_{\ell_1}$}] (xellp) at (5,3) {} ;
  
  \node (x2) at (0.5,1) {$x_2$} ;
  \node [pstate, label={-90:$x'_0$}] (x0d) at (1,1) {} ;
  \node [qstate, label={-90:$x''_{m_2{-}1}$}] (xm2minus1d) at (4,1) {} ;
  \node [pstate, label={-90:$x''_{m_2}$}] (xm2d) at (5,1) {} ;
  \node [qstate, label={-90:$x''_{\ell_2}$}] (xell2d) at (8,1) {} ;
 
  \draw [decorate, decoration=snake] (x0p) -- (xellminus1p) ;
  \draw [decorate, decoration=snake] (x0d) -- (xm2minus1d) ;
  \draw [decorate, decoration=snake] (xm2d) -- (xell2d) ;

  \draw [draw=red] (x0p) -- (x0d) ;
  \coordinate (xmid1d) at ($(x0d)!0.5!(xm2minus1d)$) ;
  \draw [draw=red] (x0p) -- ([xshift=0.25pt, yshift=-0.25pt]xmid1d) ;
  \draw [draw=red] (x0p) -- (xm2minus1d) ;
  \coordinate (xmid1p) at ($(x0p)!0.5!(xellminus1p)$) ;
  \draw [draw=red] (x0d) -- ([xshift=-2.85pt, yshift=-2.5pt] xmid1p) ;
  \draw [draw=red] (x0d) -- (xellminus1p) ;
  
  \draw [draw=red] (xellp) -- (xm2d) ;
  \coordinate (xmid2d) at ($(xm2d)!0.5!(xell2d)$) ;
  \draw [draw=red] (xellp) -- ([xshift=0.15pt, yshift=-0.15pt] xmid2d) ;
  \draw [draw=red] (xellp) -- (xell2d) ;

\end{tikzpicture}
\caption{Example illustrating the proof of Proposition~\ref{prop:CoPA-rec-vs-zones} ($\copbisnew \subseteq \copbisold$).}\label{fig:lem:CoPA-rec-vs-zonesNIO}
\end{figure}
Then $m_2$ is well-defined 
since $x_2 \not\copbisnew x'_{\ell_1}$ by hypothesis and $x'_{\ell_1} \copbisnew x''_{\ell_2}$ because
$\pi_1$ is compatible with $\pi_2$. Consider the forward path $\widehat{\pi}_2 = (x''_j)_{j=0}^{m_2}$ from $x_2$, i.e., the prefix of length $m_2$ of $\pi_2$. For all $j \in [0; m_2)$ it holds that
$x''_j \copbisnew x_2$, by definition of $m_2$ and  $x_2 \copbisnew x_1$, by assumption.
Let, for some $N \in \nats$,  $f:[0;\ell_1] \to [1;N]$ and $g:[0;\ell_2] \to [1;N]$ be matching functions
for $\pi_1$ and $\pi_2$. Let $i\in [0;\ell_1]$ be such that  $f(i)=g(m_2)$, implying that
$x'_i \copbisnew x''_{m_2}$. 
Then $i= \ell_1$, since for all $i \in [0,\ell_1)$ it holds $x'_i \, \copbisnew \, x_2$ by hypothesis and
 $x_2 \, \not\copbisnew x''_{m_2}$ by definition of $m_2$.
Thus, we have shown that there is a forward path $\widehat{\pi}_2$ from $x_2$ such that 
$\widehat{\pi}_2(j) \copbisnew x_1$ for all $j \in [0,m_2)$ and 
$\widehat{\pi}_2(m_2) \copbisnew \pi_1(\ell_1)$.\\
(iii) The proof is similar to that for the second condition.

We now prove that $  \copbisold \subseteq \copbisnew$.\\
 (i) Also in this case, the first condition is trivially satisfied since $x_1 \copbisold x_2$ requires that 
 $x_1 \in \peval(p)$ iff $x_2 \in \peval(p)$ for all $p \in \ap$.\\
 (ii) Regarding the second condition, for any forward path $\pi=(x_i)_{i=0}^{\ell}$ the number $Z(\pi)$
 of times $\copbisold$ is not preserved while going from $\pi(0)$ to $\pi(\ell)$, i.e., 
 $$
 Z(\pi) =\big|\ZET{i\in [0,\ell)}{x_i \not\copbisold x_{i+1}}\big|.
 $$
We have to show that for all $x_1,x_2 \in X$ such that $x_1 \copbisold x_2$ and
 forward path $\pi_1 = (x'_i)_{i=0}^{\ell_1}$ from $x_1$ there is a forward path 
 $\pi_2 = (x''_j)_{j=0}^{\ell_2}$ from $x_2$ that is compatible with $\pi_1$. We proceed
 by induction on $Z(\pi_1)$.\\
 {\bf Base Case:} $Z(\pi_1)=0$.\\
 In this case, since $Z(\pi_1)=0$, we have that $x'_i \copbisold x_1$ for all $i\in [0;\ell_1]$ 
 and, by hypothesis, we have that $x_1 \copbisold x_2$. Let then $\pi_2$ be the forward
 path of length 0 from $x_2$, i.e., $\pi_2=(x_2)$. Then $\pi_2$ is compatible with $\pi_1$ because
 the constant maps $f:[0;\ell_1] \to [1;1]$ and $g:[0;0] \to [1;1]$ are matching functions for $\pi_1$ and $\pi_2$.\\
 {\bf Induction Step:} $Z(\pi_1)>0$.\\
 In the sequel, we  construct a forward path $\pi_2 = (x''_j)_{j=0}^{\ell_2}$ from $x_2$ as the concatenation of two forward paths --- namely $\widehat{\pi}_{2a}$, of length $m_{2a}$, from $x_2$ and  $\widehat{\pi}_{2b}$, of length $m_{2b}$ from $\widehat{\pi}_{2a}(m_{2a})$, with 
$\ell_2 = m_{2a}+m_{2b}$ --- and we show that $\pi_2$ is compatible with $\pi_1$.
Define $m_1 = \max\ZET{i\in [0;\ell_1)}{x'_{i}\not\copbisold x'_{i+1}}$
(see Figure~\ref{fig:lem:CoPA-rec-vs-zonesNOI}). 
\begin{figure}
\begin{tikzpicture}[%
  pstate/.style = {draw, circle, minimum size=1.5mm, inner sep=0pt},
  qstate/.style = {draw, circle, minimum size=1mm, inner sep=0pt, fill=black},
  semithick,
  ]


  \node (x1) at (2,3) {$x_1$} ;
  \node [pstate] (x0p) at (2.5,3) {} ;
  \node [qstate] (x1p) at (4.5,3) {} ;
  \node [pstate] (x2p) at (6,3) {} ;
  \node [qstate, label={+90:$x'_{m_1}$}] (xm1p) at (8,3) {} ;
  \node [pstate, label={+90:$x'_{m_1{+}1}$}] (xm1plus1p) at (11,3) {} ;
  \node [qstate, label={+90:$x'_{\ell_1}$}] (xell1p) at (13,3) {} ;
  
  \node (x2) at (2,1) {$x_2$} ;
  \node [pstate] (x0d) at (2.5,1) {} ;
  \node [qstate] (x1d) at (4.5,1) {} ;
  \node [pstate] (x2d) at (6,1) {} ;
  \node [qstate, label={-90:$x''_{m_2}$}] (xm2d) at (8,1) {} ;
  \node [qstate, label={-90:$x''_{\ell_2{-}1}$}] (xell2minus1d) at (10,1) {} ;
  \node [pstate, label={-90:$x''_{\ell_2}$}] (xell2d) at (11,1) {} ;
 
  \draw [decorate, decoration=snake] (x0p) -- (x1p) ;
  \draw [decorate, decoration=snake] (x2p) -- (xm1p) ;
  \draw [decorate, decoration=snake] (xm1plus1p) -- (xell1p) ;
  \draw [decorate, decoration=snake] (x0d) -- (x1d) ;
  \draw [decorate, decoration=snake] (x2d) -- (xm2d) ;
  \draw [decorate, decoration=snake] (xm2d) -- (xell2minus1d) ;

  \draw [dotted] (4.8,3) -- (5.8,3) ;
  \draw [dotted] (4.8,1) -- (5.8,1) ;

  \draw [draw=red] (x0p) -- (x0d) ;
  \coordinate (xmid1d) at ($(x0d)!0.5!(x1d)$) ;
  \draw [draw=red] (x0p) -- ([xshift=-2.75pt, yshift=+2.25pt]xmid1d) ;
  \draw [draw=red] (x0p) -- (x1d) ;
  \coordinate (xmid1p) at ($(x0p)!0.5!(x1p)$) ;
  \draw [draw=red] (x0d) -- ([xshift=2pt, yshift=-2.5pt] xmid1p) ;
  \draw [draw=red] (x0d) -- (x1p) ;
  
  \draw [draw=red] (x2p) -- (x2d) ;
  \coordinate (xmid2d) at ($(x2d)!0.5!(xm2d)$) ;
  \draw [draw=red] (x2p) -- ([xshift=-2.75pt, yshift=2.5pt] xmid2d) ;
  \draw [draw=red] (x2p) -- (xm2d) ;
  \coordinate (xmid2p) at ([xshift=2pt, yshift=-2.75pt]$(x2p)!0.5!(xm1p)$) ;
  \draw [draw=red] (x2d) -- (xmid2p) ;
  \draw [draw=red] (x2d) -- (xm1p) ;
  \draw [draw=red] (xm1p) -- (xm2d) ;
  \coordinate (xmid3d) at ([xshift=-3.5pt, yshift=2.75pt]$(xm2d)!0.5!(xell2minus1d)$) ;
  \draw [draw=red] (xm1p) -- (xmid3d) ;
  \draw [draw=red] (xm1p) -- (xell2minus1d) ;

  \draw [draw=red] (xm1plus1p) -- (xell2d) ;
  \coordinate (xmid3p) at ([xshift=2.25pt, yshift=-2.25pt]$(xm1plus1p)!0.5!(xell1p)$) ;
  \draw [draw=red] (xmid3p) -- (xell2d) ;

  \draw [draw=red] (xell1p) -- (xell2d) ;

\end{tikzpicture}
\caption{Example illustrating the proof of Proposition~\ref{prop:CoPA-rec-vs-zones} ($\copbisold \subseteq \copbisnew$).}\label{fig:lem:CoPA-rec-vs-zonesNOI}
\end{figure}
Note that $m_1$ is
well-defined since $Z(\pi_1)>0$. Consider the forward path $\widehat{\pi}_1 = (x'_i)_{i=0}^{m_1}$ from $x_1$, i.e., the prefix of length $m_1$ of $\pi_1$. 
Since 
$Z(\widehat{\pi}_1)<Z(\pi_1)$, a forward path $\widehat{\pi}_{2a}=(x''_j)_{j=0}^{m_{2a}}$  exists 
that is compatible with $\widehat{\pi}_1$, by the induction hypothesis.
Let, for some $N\in \nats$,  $f':[0;m_1] \to [1;N]$ and $g':[0;m_{2a}] \to [1;N]$ be matching functions for $\widehat{\pi}_1$ and $\widehat{\pi}_{2a}$ with respect to $\copbisold$. Note that, by 
compatibility of  $\widehat{\pi}_1$ and $\widehat{\pi}_{2a}$, with respect to $\copbisold$,
we have that $x'_{m_1} \copbisold  x''_{m_{2a}}$. Consider the two point forward path 
$(x'_{m_1}, x'_{m_1+1})$  from $x'_{m_1}$. Since $x'_{m_1} \copbisold x''_{m_{2a}}$, 
a forward path $\widehat{\pi}_{2b}=(z_j)_{j=0}^{m_{2b}}$ exists
from $x''_{m_{2a}}$  such that
$x'_{m_1} \copbisold z_j$ for all $j\in [0;m_{2b})$ and $x'_{m_1+1} \copbisold z_{m_{2b}}$.
We have  that $m_{2b} > 0$  since $x''_{m_{2a}}\copbisold x'_{m_1} \not\copbisold x'_{m_1+1}  \copbisold z_{m_{2b}}$.
Define path $\pi_2 = (x''_j)_{j=0}^{\ell_2}$, with $\ell_2=m_{2a}+m_{2b}$, as the concatenation 
$\widehat{\pi}_{2a}\cdot\widehat{\pi}_{2b}$ of $\widehat{\pi}_{2a}$ and $\widehat{\pi}_{2b}$. 
Obviously, $\pi_2$ is a forward path from $x_2$. Moreover, $\pi_1$ and $\pi_2$ are
compatible as shown below.
Let $f:[0;\ell_1] \to [1;N+1]$ and $g:[0;\ell_2] \to [1;N+1]$ be defined as follows, where,
for notational convenience, we let $m_2=m_{2a}$:
$$
f(i) =
\left\{
\begin{array}{l l}
f'(i) & \mbox{ if } i\in [0;m_1],\\
N+1 &  \mbox{ if } i\in (m_1;\ell_1].
\end{array}
\right.
$$
$$
g(j) =
\left\{
\begin{array}{l l}
g'(j) & \mbox{ if } j\in [0;m_2],\\
N &  \mbox{ if } j\in (m_2;\ell_2),\\
N+1 &  \mbox{ if } j=\ell_2.
\end{array}
\right.
$$

The functions $f$ and $g$ are monotone. Surjectivity of $f$ is obvious.
For what concerns  surjectivity of $g$, recall that $g'(m_2)=N$ and that $\ell_2 > m_2$. 
We now show that for all $i \in [0;\ell_1]$ and $j\in [0;\ell_2]$ satisfying $f(i)=g(j)$ it holds that
$x'_i=\pi_1(i) \copbisold \pi_2(j)=x''_j$.\\
If $f(i)=g(j) < N$ and $i\in [0;m_1],j\in [0;m_2]$,
then $f'(i)=f(i)=g(j)=g'(j)$, which implies $x'_i \copbisold x''_j$, since $f'$ and $g'$ are
matching functions for $\widehat{\pi}_1$ and $\widehat{\pi}_{2a}$ with respect to $\copbisold$.\\
If $f(i)=g(j) = N$ and $i\in [0;m_1],j\in (m_2;\ell_2)$,
then $f(i)=g(m_2)$; so $x'_i \copbisold x''_{m_2}$ 
and $x''_{m_2} \copbisold x'_{m_1}$.
Moreover $x'_{m_1}\copbisold x''_j$ by choice of $\widehat{\pi}_{2b}$.
Thence $x'_i \copbisold x''_j$.\\
If  $f(i)=g(j) = N+1$ then $i \in (m_1,\ell_1]$ and $j=\ell_2$. Since, by definition of $m_1$,
we have that 
$x'_i \copbisold x'_{m_1+1}$ and, by definition of $\widehat{\pi}_{2b}$, 
$x'_{m_1+1} \copbisold x''_{\ell_2}$, it follows that $x'_i \copbisold x''_j$.
In conclusion, we have that $\pi_1$ and $\pi_2$ are compatible.\\
(iii) Similar to case (ii).
\end{proof}

\end{document}